\documentclass[11pt]{article}
\usepackage{smile}

\usepackage{fullpage}
\usepackage{lscape}
\usepackage{bigints}
\usepackage{framed}
\usepackage{mdframed}
\usepackage{enumerate}
\usepackage[inline]{enumitem}
\usepackage[T1]{fontenc}
\usepackage{moresize}
\usepackage{bm}
\usepackage{bbm}
\usepackage{dsfont}
\usepackage{amsmath}
\usepackage{amssymb}
\usepackage{amsthm}
\usepackage{amsfonts}
\usepackage{stmaryrd}
\usepackage{array}
\usepackage{mathrsfs}
\usepackage{mathtools} % \prescript
\usepackage{extarrows}
\usepackage{stackrel}
\usepackage{relsize,exscale}
\usepackage{scalerel}
\usepackage[nodisplayskipstretch]{setspace}
\usepackage{color}
\usepackage[usenames,dvipsnames]{xcolor}
\usepackage{cancel}
\usepackage{soul}
\usepackage{undertilde}
\usepackage{xfrac}
\usepackage{siunitx}
\usepackage{graphicx}
\usepackage{float}
\usepackage{rotating}
\usepackage{subcaption}
\usepackage{overpic}
\usepackage[all]{xy}
\usepackage{tikz}
\usetikzlibrary{arrows,matrix,positioning,calc,automata,patterns}
\usepackage{booktabs}
\usepackage{dcolumn}
\usepackage{multirow}
\usepackage{diagbox}
\usepackage{tabularx}
\usepackage{tablefootnote}
\usepackage{threeparttable}
\usepackage{verbatim}
\usepackage{listings}
\usepackage[ruled,vlined]{algorithm2e}
\usepackage{fancyvrb}
\usepackage{hyperref}
\usepackage[round]{natbib}
\usepackage{sectsty}

\hypersetup{
    bookmarks=true,         % show bookmarks bar?
    unicode=false,          % non-Latin characters in Acrobat’s bookmarks
    pdftoolbar=true,        % show Acrobat’s toolbar?
    pdfmenubar=true,        % show Acrobat’s menu?
    pdffitwindow=false,     % window fit to page when opened
    pdfstartview={FitH},    % fits the width of the page to the window
    pdftitle={My title},    % title
    pdfauthor={Author},     % author
    pdfsubject={Subject},   % subject of the document
    pdfcreator={Creator},   % creator of the document
    pdfproducer={Producer}, % producer of the document
    pdfkeywords={key1, key2}, % list of keywords
    pdfnewwindow=true,      % links in new PDF window
    colorlinks=true,        % false: boxed links; true: colored links
    linkcolor=blue,         % color of internal links (change box color with linkbordercolor)
    citecolor=blue,         % color of links to bibliography
    filecolor=blue,         % color of file links
    urlcolor=cyan           % color of external links
}

\usepackage{stackengine}
\stackMath
\newcommand\tenq[2][1]{%
\def\useanchorwidth{T}%
\ifnum#1>1%
\stackunder[0pt]{\tenq[\numexpr#1-1\relax]{#2}}{\!\scriptscriptstyle\thicksim}%
\else%
\stackunder[1pt]{#2}{\!\scriptstyle\thicksim}%
\fi%
}

\makeatletter
\DeclareRobustCommand\widecheck[1]{{\mathpalette\@widecheck{#1}}}
\def\@widecheck#1#2{%
    \setbox\z@\hbox{\m@th$#1#2$}%
    \setbox\tw@\hbox{\m@th$#1%
       \widehat{%
          \vrule\@width\z@\@height\ht\z@
          \vrule\@height\z@\@width\wd\z@}$}%
    \dp\tw@-\ht\z@
    \@tempdima\ht\z@ \advance\@tempdima2\ht\tw@ \divide\@tempdima\thr@@
    \setbox\tw@\hbox{%
       \raise\@tempdima\hbox{\scalebox{1}[-1]{\lower\@tempdima\box
\tw@}}}%
    {\ooalign{\box\tw@ \cr \box\z@}}}
\makeatother

\def\given{\,|\,}

\def\tr{\mathop{\text{tr}}\kern.2ex}

\def\P{{\mathrm P}}

\def\E{{\mathrm E}}

\def\d{{\mathrm d}}
\def\cI{{\mathcal I}}

\newcommand{\indep}{\perp \!\!\!\perp}

\newcolumntype{L}[1]{>{\raggedright\let\newline\\\arraybackslash\hspace{0pt}}m{#1}}
\newcolumntype{C}[1]{>{  \centering\let\newline\\\arraybackslash\hspace{0pt}}m{#1}}
\newcolumntype{R}[1]{>{ \raggedleft\let\newline\\\arraybackslash\hspace{0pt}}m{#1}}
\newcolumntype{d}[1]{D{.}{.}{#1}}
\newcolumntype{H}{>{\setbox0=\hbox\bgroup}c<{\egroup}@{}}
\newcolumntype{Z}{>{\setbox0=\hbox\bgroup}c<{\egroup}@{\hspace*{-\tabcolsep}}}
\newcolumntype{b}{X}
\newcolumntype{s}{>{\hsize=.5\hsize}X}

\newcommand{\htau}{\widehat\tau}

\numberwithin{equation}{section}

\newtheorem{theorem}{Theorem}[section]
\newtheorem{lemma}{Lemma}[section]

\newtheorem{assumption}{Assumption}[section]

\newtheorem{innercustomgeneric}{\customgenericname}
\providecommand{\customgenericname}{}
\newcommand{\newcustomtheorem}[2]{%
  \newenvironment{#1}[1]
  {%
   \renewcommand\customgenericname{#2}%
   \renewcommand\theinnercustomgeneric{##1}%
   \innercustomgeneric
  }
  {\endinnercustomgeneric}
}
\newcustomtheorem{customdefinition}{Definition}
\newcustomtheorem{customdefinitions}{Definitions}
\newcustomtheorem{customtheorem}{Theorem}
\newcustomtheorem{customassumption}{Assumption}
\newcustomtheorem{customlemma}{Lemma}
\newcustomtheorem{customexample}{Example}
\theoremstyle{definition}
\newtheorem{definition}{Definition}[section]
\newtheorem{example}{Example}[section]
\newtheorem{remark}{Remark}[section]

\usepackage{enumitem}
\makeatletter
\newcommand{\mylabel}[2]{#2\def\@currentlabel{#2}\label{#1}}
\makeatother

\graphicspath{{./fig3/}}

\allowdisplaybreaks

\begin{document}

\setlength{\abovedisplayskip}{5pt}
\setlength{\belowdisplayskip}{5pt}
\setlength{\abovedisplayshortskip}{5pt}
\setlength{\belowdisplayshortskip}{5pt}
\hypersetup{colorlinks,breaklinks,urlcolor=blue,linkcolor=blue}

\title{\LARGE Introducing the b-value: combining unbiased and biased estimators from a sensitivity analysis perspective }

\author{Zhexiao Lin\thanks{Department of Statistics, University of California, Berkeley, CA 94720, USA; e-mail: {\tt zhexiaolin@berkeley.edu}},~~~Peter J. Bickel\thanks{Department of Statistics, University of California, Berkeley, CA 94720, USA; e-mail: {\tt bickel@stat.berkeley.edu}},~~~and~
Peng Ding\thanks{Department of Statistics, University of California, Berkeley, CA 94720, USA; e-mail: {\tt pengdingpku@berkeley.edu}}
}

\maketitle

\vspace{-1em}

\begin{abstract}
  In empirical research, when we have multiple estimators for the same parameter of interest, a central question arises: how do we combine unbiased but less precise estimators with biased but more precise ones to improve the inference? Under this setting, the point estimation problem has attracted considerable attention. In this paper, we focus on a less studied inference question: how can we conduct valid statistical inference in such settings with unknown bias? We propose a strategy to combine unbiased and biased estimators from a sensitivity analysis perspective. We derive a sequence of confidence intervals indexed by the magnitude of the bias, which enable researchers to assess how conclusions vary with the bias levels. Importantly, we introduce the notion of the b-value, a critical value of the unknown maximum relative bias at which combining estimators does not yield a significant result. We apply this strategy to three canonical combined estimators: the precision-weighted estimator, the pretest estimator, and the soft-thresholding estimator. For each estimator, we characterize the sequence of confidence intervals and determine the bias threshold at which the conclusion changes. Based on the theory, we recommend reporting the b-value based on the soft-thresholding estimator and its associated confidence intervals, which are robust to unknown bias and achieve the lowest worst-case risk among the alternatives.
\end{abstract}


\section{Introduction}

In empirical research, it is common for researchers to employ different methods to estimate the same parameter of interest. These differences may arise from the use of distinct datasets or from imposing different model assumptions on the same dataset. We motivate our paper with the following two examples of combining estimators.
\begin{example}\label{ex:treatment}
  Randomized controlled trials (RCTs) are the gold standard for estimating treatment effects due to their ability to eliminate unmeasured confounding. However, RCTs often suffer from limited sample sizes, as large-scale experiments can be costly or infeasible. In contrast, observational data are more readily available from the target population of interest. However, estimates using observational data may be biased in estimating treatment effects due to unmeasured confounding, raising concerns about the internal validity. See \cite{brantner2023methods} and \cite{colnet2024causal} for recent reviews on motivations and methods for combining RCTs and observational studies.
\end{example}

\begin{example}\label{ex:regression}
  The ordinary least squares (OLS) estimator is biased in estimating the unknown parameters under the linear model when the error term is correlated with the regressor. In contrast, the instrumental variables (IV) estimator can provide unbiased estimates for the parameters of interest with a valid instrumental variable that is uncorrelated with the error term but correlated with the regressor. However, the IV estimator is usually much less precise than the OLS estimator, especially when the IV is weakly correlated with the regressor \citep{bound1995problems}. In empirical studies, e.g., \cite{angrist1991does}, researchers often report results from both OLS and IV estimators. How to combine OLS and IV estimators is gaining increasing interest \citep{armstrong2025adapting}.
\end{example}

When the estimators are from different datasets, e.g., Example~\ref{ex:treatment}, the estimators are independent as long as the datasets are independent. When the estimators are from the same dataset but with different model assumptions, e.g., Example~\ref{ex:regression}, the estimators are dependent in general. Given access to multiple potentially dependent estimators, some unbiased but less precise and others biased but more precise, a natural question is: How can we combine the unbiased and potentially biased estimators to improve the inference with unknown bias? From the point estimation perspective, this problem has been extensively studied \citep{bickel1984parametric,green1991james, giles1993pre,chen2015data,athey2020combining,de2020empirical, rosenman2023combining,gao2023pretest, yang2025cross}. Many methods have been proposed for constructing combined estimators that perform well when the bias is small and have bounded risks when the bias is large. From the statistical inference perspective, this problem is less studied. In this paper, we answer the following question: how can we conduct valid statistical inference after combining the estimators?

This question has received considerably less attention. The primary difficulty lies in the impossibility of characterizing the distribution of the combined estimator with unknown bias \citep{armstrong2025adapting}.
% This impossibility is closely related to results in nonparametric regression, where honest confidence intervals cannot be constructed without some restriction on the function class, such as a bound on its derivative \citep{low1997nonparametric,cai2004adaptation,armstrong2018optimal,armstrong2020simple}. The key insight is that once even limited information about the bias is introduced, e.g., an upper bound on its magnitude, confidence intervals for the parameter of interest become attainable. 
Once the information about the bias is introduced, e.g., an upper bound on its magnitude, confidence intervals for the parameter of interest become possible. In the absence of such information, we focus on the following question:
\begin{quote}
  \textit{How can we construct a sequence of confidence intervals for the parameter of interest across bias levels?}
\end{quote}

The sequence of confidence intervals provides a way to quantify how the level of bias impacts the uncertainty in point estimation. One important application of this sequence is in hypothesis testing. Suppose the null hypothesis is not rejected based on the unbiased but less precise estimator. Now consider a scenario where the statistical test based on the biased but more precise estimator rejects the null hypothesis. If we have prior knowledge suggesting the bias is small, incorporating the biased estimator may yield a more precise estimator to reject the null hypothesis. In such cases, the sequence of confidence intervals enables us to address the following question:
\begin{quote}
  \textit{How large must the bias be to change the conclusion of a hypothesis test—from rejection to non-rejection?}
\end{quote}

The idea of constructing the sequence of confidence intervals and examining how conclusions change as the assumed level of bias varies is related to sensitivity analysis in causal inference with unmeasured confounding, e.g., \cite{cornfield1959smoking,rosenbaum1983assessing,vanderweele2017sensitivity}. In observational studies, sensitivity analysis assesses
how the causal conclusions change with respect to different degrees of unmeasured confounding by varying the sensitivity parameter \citep{rosenbaum2002observational,ding2016sensitivity}. Our proposed framework has a similar flavor: by indexing inference results over a continuum of bias levels, we can assess the robustness of statistical inference.

We formalize two statistical inference questions, confidence interval and hypothesis testing, in the context of combining unbiased and biased estimators. 
% Our formulation is of broad interest as it encompasses both randomization-based and sampling-based frameworks. In randomization-based inference, the estimators can be viewed as functions of treatment assignments drawn from a sequence of finite populations. In sampling-based inference, they arise from random samples from an underlying superpopulation.
Under regularity conditions,
% \citep{madow1948limiting,erdHos1959central,hajek1960limiting, lehmann1999elements,van2000asymptotic}
the estimators satisfy a joint central limit theorem. Consequently, we present our formulation in a finite-sample Gaussian setting, assuming exact normality for both the unbiased and biased estimators. This reduction to a Gaussian model is motivated by Le Cam’s classical asymptotic argument \citep{le1956asymptotic}, 
% , which shows that general statistical experiments can be locally approximated by Gaussian shift experiments. Hence, 
and because of this, our Gaussian formulation should be viewed as an asymptotic idealization rather than a restrictive finite-sample assumption.

We develop a general framework that applies to any point estimator formed by combining such estimators. Within this framework, we construct a sequence of confidence intervals indexed by the bias level, and importantly, we introduce the notion of the b-value, a critical value of the unknown maximum relative bias at which combining estimators does not yield a significant result. We examine three canonical combined estimators: the precision-weighted estimator, the pretest estimator, and the soft-thresholding estimator. For each estimator, we derive either analytically or numerically the sequence of confidence intervals and the b-value. Among the three, we advocate for the soft-thresholding estimator, as it offers robustness to unknown bias compared with the precision-weighted estimator, and exhibits lower worst-case risk and more desirable properties for confidence interval construction than the pretest estimator. We provide a Python package for the proposed methods, available at \url{https://github.com/zhexiaolin/b-value}.

\textbf{Notation.} For a vector $a = (a_1,\ldots,a_d)^\top \in \mathbb{R}^d$, let $\lVert a \rVert_1 = \sum_{i=1}^d \lvert a_i \rvert$, $\lVert a \rVert_2 = ( \sum_{i=1}^d a_i^2 )^{1/2}$, and $\lvert a \rvert = (\lvert a_1 \rvert, \ldots, \lvert a_d \rvert)^\top$. For vectors $a = (a_1,\ldots,a_d)^\top$ and $b = (b_1,\ldots,b_d)^\top$, let $a \odot b = (a_1 b_1, \ldots, a_d b_d)^\top$ be the Hadamard  (element-wise) product, and $a \le b$ denote $a_i \le b_i$ for all $i=1,2,\ldots,d$. For a scalar $b \in \mathbb{R}$ and a set $A \subset \mathbb{R}$, we write
$b - A = \{b - a : a \in A\}$, which extends naturally to vectors and sets in $\mathbb{R}^d$. We use $\Phi(\cdot)$ and $\phi(\cdot)$ to denote the cumulative distribution function and density function of the standard normal distribution, respectively. We use $\phi_{\bm{\mu}, \bm{\Sigma}}(\cdot)$ to denote the density function of the multivariate normal distribution with mean $\bm{\mu}$ and covariance matrix $\bm{\Sigma}$. We use $\Psi_d(\cdot; \lambda)$ to denote the cumulative distribution function of the noncentral chi-squared distribution with noncentrality parameter $\lambda$ and degrees of freedom $d$. We use $c_{\alpha}$ to denote the $(1-\alpha)$ upper quantile of the standard normal distribution.

\section{Problem setup and a review of point estimation}\label{sec:setup}

\subsection{Problem Setup}

We consider the following setting:
\begin{assumption}\label{ass:normal}
  Suppose we observe two independent random variables: one unbiased estimator $\htau_0 \sim N(\tau, \sigma_0^2)$ and one biased estimator $\htau_1 \sim N(\tau + \Delta, \sigma_1^2)$. Here, $\tau \in \bR$ is the unknown parameter of interest. We assume that $\sigma_0^2$ and $\sigma_1^2$ are known, whereas $\Delta$ is unknown.
\end{assumption}

In practice, both the unbiased and biased estimators for the same parameter of interest are constructed from data. Under regularity conditions, these estimators are jointly asymptotically normal with an unknown covariance matrix. As long as this covariance matrix can be consistently estimated, the problem of combining estimators reduces to an exact normality framework under Le Cam’s asymptotic framework \citep{le1956asymptotic}. \cite{le1956asymptotic} showed that a wide class of estimators and test statistics can be approximated in large samples by Gaussian experiments with known covariances. Within this framework, we treat the estimators as arising from Gaussian experiments in which the variances are replaced by their consistent estimates, and the validity of the inference procedures is preserved asymptotically. For this reason, we present our analysis in the exact normality setting. Nevertheless, our results represent the asymptotic limits of a broad class of more general and practically relevant inference problems. The analysis can be generalized to dependent (Section~\ref{sec:dependent}), multivariate (Section~\ref{sec:multivariate}) and multiple estimators (Section~\ref{sec:data_fusion}) cases.

Let $\gamma = \sigma_0^2/\sigma_1^2$ be the variance ratio between $\htau_0$ and $\htau_1$. In general, the unbiased estimator is less precise but the biased estimator is more precise. Therefore, we focus on the regime in which $\sigma_0^2$ is large and $\sigma_1^2$ is small, which indicates that $\gamma$ is large.

Under Assumption~\ref{ass:normal}, we address the following central question:
\begin{quote}
  \textit{How can we construct a sequence of two-sided confidence intervals for $\tau$ across different levels of bias ($\Delta$)?}
\end{quote}

We focus on constructing two-sided symmetric confidence intervals centered at prespecified point estimators. In other words, we do not let the point estimator itself depend on the bias level $\Delta$. An alternative approach would be to use a bias-dependent point estimator, where the estimator incorporates $\Delta$. We will discuss the advantages of using prespecified point estimators later in Appendix~\ref{sec:bias_dependent}.

Since $\htau_0$ is an unbiased estimator of $\tau$, one natural approach is to construct a confidence interval based solely on $\htau_0$, which is invariant to the bias $\Delta$. Given a significance level $\zeta \in (0,1)$, we can construct a standard two-sided confidence interval
\[
[\htau_0 - \sigma_0 c_{\zeta/2},\; \htau_0 + \sigma_0 c_{\zeta/2}],
\]
where $c_{\zeta/2}$ is the $(1-\zeta/2)$ upper quantile of the standard normal distribution. While being valid regardless of the bias $\Delta$, this confidence interval may be too wide when $\htau_0$ is not precise, whose length $2\sigma_0 c_{\zeta/2}$ scales proportionally with $\sigma_0$. As we have access to the additional biased but more precise estimator $\htau_1$, a natural question arises: can we shorten the confidence interval by incorporating information from $\htau_1$?

This motivates us to combine the two estimators $\htau_0$ and $\htau_1$ to construct a two-sided confidence interval with shorter length. Consider a generic combined estimator 
\[
\hat\tau = \hat\tau(\htau_0, \htau_1, \sigma_0^2, \sigma_1^2)
\]
of $\tau$, which depends only on observed data $(\htau_0, \htau_1)$ and known variances $(\sigma_0^2, \sigma_1^2)$, but not on the unknown bias $\Delta$. If $\Delta$ were known, constructing an exact two-sided confidence interval would be straightforward, as the distributions of both estimators are known, which lead to known distribution of $\hat\tau$. However, since $\Delta$ is unknown, the exact confidence interval depends on the magnitude of $\Delta$.

To analyze the problem, we impose bounds on the bias $\Delta$ and construct two-sided confidence intervals for different levels of bias. We assume that $\lvert \Delta/\sigma_0 \rvert \le b$ for some $b \ge 0$, and study how the confidence interval based on $\hat\tau$ changes as a function of the bias bound $b$. Here we focus on the relative bias $\Delta/\sigma_0$ instead of the absolute bias $\Delta$ to ensure that the parameterization is invariant to the scale of the estimators. For a given $b$, we aim to construct a two-sided confidence interval that achieves correct coverage uniformly over all $\Delta$ satisfying $\lvert \Delta/\sigma_0\rvert \le b$. We thus define the confidence interval below.
\begin{definition}\label{def:CI}
  Given a significance level $\zeta \in (0,1)$ and a maximum relative bias $b \ge 0$, we want to construct an interval $\cI(b, \zeta) = \cI(b, \zeta, \htau_0, \htau_1, \sigma_0^2, \sigma_1^2)$ such that
  \begin{align}\label{eq:coverage}
    \inf_{\Delta: \lvert \Delta/\sigma_0 \rvert \le b} \P_\Delta(\tau \in \hat\tau - \cI(b, \zeta)) 
    = \inf_{\Delta: \lvert \Delta/\sigma_0  \rvert \le b} \P_\Delta(\hat\tau - \tau \in \cI(b, \zeta)) \ge 1 - \zeta,
  \end{align}
  where $\P_\Delta$ explicitly denotes the dependence of the distribution of $\hat\tau$ on $\Delta$. The confidence interval for $\tau$ based on $\hat\tau$ is then given by $\hat\tau - \cI(b, \zeta)$.
\end{definition}

We next impose two natural monotonicity conditions on the interval $\cI(b, \zeta)$ below.
\begin{assumption}\label{asp:monotonicity}
We assume:
\begin{enumerate}
  \item For fixed $(\htau_0, \htau_1, \sigma_0^2, \sigma_1^2)$ and $\zeta$, we have $\cI(b, \zeta) \subset \cI(b', \zeta)$ whenever $b \le b'$.
  \item For fixed $(\htau_0, \htau_1, \sigma_0^2, \sigma_1^2)$ and $b$, we have $\cI(b, \zeta) \subset \cI(b, \zeta')$ whenever $\zeta \ge \zeta'$.
\end{enumerate}
\end{assumption}

The first condition in Assumption~\ref{asp:monotonicity} requires that as we allow greater bias in $\htau_1$, the confidence interval becomes wider and contains the previous intervals. The second condition in Assumption~\ref{asp:monotonicity} requires that to guarantee higher coverage rate, the confidence interval must widen. A common class of intervals satisfying both conditions is the class of symmetric fixed-length centered intervals, whose length does not depend on $(\htau_0, \htau_1)$, i.e., $\cI(b,\zeta) = [-c(b,\zeta, \sigma_0^2, \sigma_1^2), c(b,\zeta, \sigma_0^2, \sigma_1^2)]$ for some $c(b,\zeta, \sigma_0^2, \sigma_1^2) \ge 0$ depending only on $b$, $\zeta$, $\sigma_0^2$, and $\sigma_1^2$. In this case, the confidence interval $\hat\tau - \cI(b, \zeta)$ in Definition~\ref{def:CI} is given by $[\hat\tau - c(b,\zeta, \sigma_0^2, \sigma_1^2), \hat\tau + c(b,\zeta, \sigma_0^2, \sigma_1^2)]$. 

Therefore, for given $(\htau_0, \htau_1, \sigma_0^2, \sigma_1^2)$ and $\zeta$, we can regard the confidence interval $\hat\tau - \cI(b, \zeta)$ in Definition~\ref{def:CI} as a function of the bias bound $b$. Thus, to address the central question of constructing a sequence of two-sided confidence intervals for $\tau$ across different levels of bias, we compute this confidence interval for a range of values of $b$: $\{\hat\tau - \cI(b, \zeta)\}_{b \ge 0}$.

% \nb{
% The confidence interval in Definition~\ref{def:CI} is closely related to existing results in the nonparametric regression literature. Specifically, $\hat\tau - \cI(b, \zeta)$ is bias-aware \citep{armstrong2018optimal,armstrong2020bias,armstrong2022robust} in the sense that $\cI(b, \zeta)$ explicitly accounts for a possible bias of magnitude $b$. The uniform coverage requirement over all values within the bias bound in Definition~\ref{def:CI} mirrors the notion of honesty introduced by \cite{li1989honest}, which requires confidence intervals to cover the true parameter uniformly (asymptotically) over the parameter space. This is also conceptually aligned with minimax optimality in the classical decision theory.
% }

Once we construct the sequence of confidence intervals, a natural application is hypothesis testing about the parameter $\tau$. We focus on the two-sided test:
\begin{align}\label{eq:two_sided_test}
  H_0: \tau = 0 \quad\text{versus}\quad H_1: \tau \neq 0.
\end{align}
We generalize the discussion to one-sided tests, such as testing $\tau = 0$ versus $\tau > 0$ or testing $\tau \le 0$ versus $\tau > 0$, in Section~\ref{sec:one_sided}. It is straightforward to extend the two-sided test to hypotheses of the form $\tau = \tau^*$ versus $\tau \neq \tau^*$ for any given $\tau^* \in \mathbb{R}$. A common practice for the two-sided test involves constructing confidence intervals for $\tau$ and checking whether the null value $0$ lies within these intervals.

Under Assumption~\ref{asp:monotonicity}, the width of the confidence interval increases with the bias bound $b$. This leads to a central question when using the combined estimator $\hat\tau$ with the confidence interval $\hat\tau - \cI(b, \zeta)$:
\begin{quote}
  \textit{How large must the bias bound $b$ be to change the conclusion of the hypothesis test \eqref{eq:two_sided_test}?}
\end{quote}

By the monotonicity of $\cI(b, \zeta)$ in $b$ (holding $(\htau_0, \htau_1)$ and $\zeta$ fixed) from Assumption~\ref{asp:monotonicity}, we define the limiting interval as $\cI(\infty, \zeta) = \lim_{b \to \infty} \cI(b, \zeta)$. Then there may exist a critical value of $b$ at which the confidence interval $\hat\tau - \cI(b, \zeta)$ contains the null value tested in \eqref{eq:two_sided_test}. We focus on cases where the critical value exists. Specifically, we consider the case when $0 \notin \hat\tau - \cI(0, \zeta)$ but $0 \in \hat\tau - \cI(\infty, \zeta)$. In this scenario, we define the b-value below.

\begin{definition}\label{def:b}
  Define the b-value as the critical value $b^*$ of testing $\tau = 0$ versus $\tau \neq 0$:
  \begin{align}\label{eq:b}
    b^*(\zeta) = b^*(\zeta, \htau_0, \htau_1, \sigma_0^2, \sigma_1^2) = \inf\left\{b \ge 0: 0 \in \hat\tau - \cI(b, \zeta)\right\}.
  \end{align}
\end{definition}

By the monotonicity conditions in Assumption~\ref{asp:monotonicity}, we have $0 \notin \hat\tau - \cI(b, \zeta)$ for $0 \le b < b^*$ and $0 \in \hat\tau - \cI(b, \zeta)$ for $b > b^*$. Thus, we reject the null hypothesis $\tau = 0$ when $b < b^*$ and fail to reject it for $b > b^*$. The b-value $b^*$ thus represents the maximum relative bias beyond which the null hypothesis can no longer be rejected. We propose to report $b^*$ defined in \eqref{eq:b} for given estimators $(\htau_0,\htau_1)$ and significance level $\zeta$. 

Compared with the sensitivity analysis literature, the bias bound $b$ in our framework plays a role analogous to the sensitivity parameter. The corresponding confidence interval $\hat\tau - \cI(b, \zeta)$ serves as a sensitivity curve, illustrating how inference changes as the bias varies. In this context, the b-value $b^*$ serves a role similar to key robustness metrics in prior work: it parallels
% the breakdown point and breakdown frontier by 
% \cite{horowitz1995identification} and \cite{masten2020inference}, 
the design sensitivity by \cite{rosenbaum2004design}, the E-value by \cite{vanderweele2017sensitivity}, and the robustness value by \cite{cinelli2020making}.

The b-value $b^*$ provides a criterion for comparing competing strategies of confidence interval construction. Different choices of the combined estimator $\hat\tau$ and different formulations of $\cI(b, \zeta)$ can lead to different values of $b^*$. We prefer procedures that yield a larger b-value $b^*$, since this indicates that the resulting confidence interval is more robust to potential bias in the biased estimator. This role of $b^*$ is analogous to that use of the design sensitivity in \cite{rosenbaum2004design}, where it serves as a criterion for comparing different test statistics and matched designs in observational studies.

\begin{remark}
  Besides the scenario of primary interest: $0 \notin \hat\tau - \cI(0, \zeta)$ and $0 \in \hat\tau - \cI(\infty, \zeta)$, two other less interesting scenarios may occur. First, if $0 \in \hat\tau - \cI(0, \zeta)$, then $0 \in \hat\tau - \cI(b, \zeta)$ for all $b \ge 0$, and we always fail to reject the null hypothesis regardless of the bias magnitude. Second, if $0 \notin \hat\tau - \cI(\infty, \zeta)$, then $0 \notin \hat\tau - \cI(b, \zeta)$ for all $b \ge 0$, and we always reject the null hypothesis regardless of the bias magnitude. In these two scenarios, combining the estimators or not does not change the statistical result qualitatively. Therefore, we ignore them in our discussion. Depending on the choice of $\hat\tau$ and construction of $\cI(b, \zeta)$, not all three scenarios may arise. We focus on the case when the b-value $b^*$ defined in \eqref{eq:b} is in the interval $(0, \infty)$.
\end{remark}

\subsection{Point estimation: a review}\label{sec:point}

Before diving into the details of the method for constructing confidence intervals and obtaining the b-value, we first review the existing approaches to point estimation. Specifically, we review three point estimators: the precision-weighted estimator, the pretest estimator, and the soft-thresholding estimator.

First, we recall the precision-weighted estimator:
\[
\htau_{\rm PW} := \frac{\sigma_1^2}{\sigma_0^2 + \sigma_1^2} \htau_0 + \frac{\sigma_0^2}{\sigma_0^2 + \sigma_1^2} \htau_1 = \htau_0 + \frac{\sigma_0^2}{\sigma_0^2 + \sigma_1^2} (\htau_1 - \htau_0) = \htau_0 + \frac{\gamma}{1+\gamma} (\htau_1 - \htau_0).
\]
When $\Delta$ is known to be 0, the precision-weighted estimator is the maximum likelihood estimator and best linear unbiased estimator of $\tau$. Moreover, by the classical Le Cam asymptotic decision theory \citep{le1956asymptotic} and classical results for the Gaussian shift model, it is asymptotically admissible and minimax under $L_2$ risk in general, i.e., $\E[(\hat\tau - \tau)^2]$, among all regular estimators under standard regularity conditions.

However, $\htau_{\rm PW}$ is not robust to bias: its risk is large when $\lvert \Delta \rvert$ is large. 
% To see that, since $\htau_{\rm PW} \sim N(\tau +  (1+\gamma)^{-1} \gamma \Delta, (1+\gamma)^{-1} \sigma_0^2)$, the bias of $\htau_{\rm PW}$ in estimating $\tau$ is $ (1+\gamma)^{-1} \gamma \Delta$, and its risk is
% \[
% \E_\Delta[(\htau_{\rm PW} - \tau)^2] = \gamma^2(1+\gamma)^{-2} \Delta^2 + (1+\gamma)^{-1} \sigma_0^2,
% \]
% which diverges as $\lvert \Delta \rvert \to \infty$. Consequently, the length of the confidence interval $[\htau_{\rm PW} - \hat{L}_{\rm PW} (1+\gamma)^{-1/2} \sigma_0, \htau_{\rm PW} + \hat{L}_{\rm PW} (1+\gamma)^{-1/2} \sigma_0]$ also diverges, as $\hat{L}_{\rm PW}$ defined in Theorem~\ref{thm:L_PW} grows without bound when $b \to \infty$. 
This motivates us to consider a combined estimator that performs nearly as well as $\htau_{\rm PW}$ when the bias is small, and is robust to unknown bias $\Delta$, ensuring that the worst-case risk $\sup_{\Delta \in \bR} \E_\Delta[(\hat\tau - \tau)^2]$ remains bounded. This motivates the following two estimators. 

Second, we recall the pretest estimator, which involves incorporating a pretest for $\Delta = 0$ versus $\Delta \neq 0$, a procedure that is commonly used \citep{bancroft1944biases,wallace1977pretest, bancroft1977inference,giles1993pre}. Under $\Delta = 0$, given the independence between $\htau_1$ and $\htau_0$, their difference follows: $\htau_1 - \htau_0 \sim N(0, \sigma^2)$, where $\sigma^2 = \sigma_0^2 + \sigma_1^2$. For a fixed significance level $\alpha \in (0,1)$, we consider the test statistic $(\htau_1 - \htau_0)/\sigma$. Let $$\mathcal{A} = \{\lvert \htau_1 - \htau_0 \rvert \le \sigma c_{\alpha/2}\}$$ denote the event that the pretest fails to reject the null hypothesis ($\Delta = 0$). If the pretest fails to reject the null hypothesis, i.e., $\lvert \htau_1 - \htau_0 \rvert \le \sigma c_{\alpha/2}$, the pretest estimator uses the precision-weighted estimator:
\[
  \htau_0 + \frac{\gamma}{1+\gamma} (\htau_1 - \htau_0).
\]
If the pretest rejects the null hypothesis, i.e., $\lvert \htau_1 - \htau_0 \rvert > \sigma c_{\alpha/2}$, the pretest estimator employs hard-thresholding by reverting to the unbiased estimator $\htau_0$. Combining the two cases, the pretest estimator is:
\[
\htau_{\rm PT} = \htau_0 + \frac{\gamma}{1+\gamma} (\htau_1 - \htau_0) \ind\left(\mathcal{A}\right).
\]

Third, we recall the soft-thresholding estimator, which ensures continuity at the pretest boundary ($\lvert \htau_1 - \htau_0 \rvert = \sigma c_{\alpha/2}$). If the pretest fails to reject the null hypothesis, the soft-thresholding estimator also uses the precision-weighted estimator. If the pretest rejects the null hypothesis, the soft-thresholding estimator employs soft-thresholding by:
\[
\htau_0 + \frac{\gamma}{1+\gamma} \sigma c_{\alpha/2}\,\sign(\htau_1 - \htau_0).
\]
Combining the two cases, the soft-thresholding estimator is:
\[
\htau_{\rm ST}= \htau_0 + \frac{\gamma}{1+\gamma} (\htau_1 - \htau_0)\ind\left(\mathcal{A}\right) 
+ \frac{\gamma}{1+\gamma} \sigma c_{\alpha/2}\,\sign(\htau_1 - \htau_0)\ind\left(\mathcal{A}^\textup{c}\right).
\]

% Compared with $\htau_{\rm PT}$, the estimator $\htau_{\rm ST}$ aligns more closely with estimators in other contexts, e.g., Efron–Morris estimator \citep{efron1971limiting,j1981estimation, bickel1983minimax}, optimization and regularization \citep{duan2023adaptive}, Bayesian hierarchical models \citep{park2008bayesian}, shrinkage estimators \citep{green1991james,de2020empirical}. Furthermore, 
While $\htau_{\rm PT}$ reduces to $\htau_0$ with probability nearly one when the bias is large, the risk of $\htau_{\rm PT}$ is much higher than that of $\htau_{\rm ST}$ when the bias is moderate. See \cite{bickel1983minimax} and \cite{armstrong2025adapting} for a comparison of worst-case risk between $\htau_{\rm PT}$ and $\htau_{\rm ST}$ and the superior performance of $\htau_{\rm ST}$.

In both $\htau_{\rm PT}$ and $\htau_{\rm ST}$, the role of $\alpha$ is different from its usual interpretation in hypothesis testing. 
% It is important to note that the confidence intervals and hypothesis testings of interest concern the target parameter $\tau$, not the bias parameter $\Delta$. The relevant confidence level is determined by $\zeta$, as in the confidence interval $\htau - \cI(b, \zeta)$ in Definition~\ref{def:CI}, which controls the coverage probability for $\tau$ given a bias bound $b$. In contrast, $\alpha$ appears only in the construction of the combined estimator. In $\htau_{\rm PT}$, $\alpha$ controls the probability of rejecting the null hypothesis that $\Delta=0$, thereby introducing a discontinuity in the estimator at the rejection boundary. In $\htau_{\rm ST}$, the same critical value $c_{\alpha/2}$ determines the size of the shrinkage applied to $\htau_1 - \htau_0$. 
Rather than a Type I error rate, $\alpha$ here acts as a tuning parameter that balances the bias and variance of the point estimator. In this sense, $\alpha$ simply indexes a family of estimators, much like the penalty level in penalized regression methods (e.g., the Lasso). Choosing $\alpha$ optimally depends on $\Delta$, which is unknown in practice. In our empirical studies, we set $\alpha = 0.05$ as a default, following the standard convention, and examine how the estimator’s performance varies across different bias levels.

\section{Confidence intervals, hypothesis testing, and the b-value}\label{sec:confidence}

Given the canonical point estimators introduced in Section~\ref{sec:point}, we now discuss the problem of constructing confidence intervals, hypothesis testing, and the b-value.

\subsection{Confidence interval based on the precision-weighted estimator}\label{sec:simple}

As a warm-up, we use the precision-weighted estimator $\htau_{\rm PW}$ to construct the sequence of confidence intervals and to illustrate how the confidence interval changes with respect to $b$. Based on the theory below, we do not recommend using the precision-weighted estimator and its corresponding confidence intervals in practice. Under Assumption~\ref{ass:normal}, we have $\htau_{\rm PW} \sim N(\tau +  (1+\gamma)^{-1} \gamma \Delta, (1+\gamma)^{-1} \sigma_0^2)$. The following theorem provides the sequence of confidence intervals based on $\htau_{\rm PW}$ depending on $b$.

\begin{theorem}\label{thm:L_PW}
  Let $\hat{L}_{\rm PW} = \hat{L}_{\rm PW}(b, \zeta, \gamma) \ge 0$ denote the solution to the equation of $L$:
  \[
    \Phi\Big(L - \frac{\gamma}{\sqrt{1+\gamma}} b \Big) - \Phi\Big(-L - \frac{\gamma}{\sqrt{1+\gamma}} b \Big) = 1-\zeta.
  \]
  The $\hat{L}_{\rm PW}$ always exists and is unique. The shortest length symmetric centered confidence interval based on $\htau_{\rm PW}$ for $\tau$ satisfying \eqref{eq:coverage} is given by $[\htau_{\rm PW} - \hat{L}_{\rm PW} (1+\gamma)^{-1/2} \sigma_0, \htau_{\rm PW} + \hat{L}_{\rm PW} (1+\gamma)^{-1/2} \sigma_0]$.
\end{theorem}

The $\hat{L}_{\rm PW}$ in Theorem~\ref{thm:L_PW} corresponds to the $(1-\zeta)$ quantile of the folded normal distribution $\lvert N(\frac{\gamma}{\sqrt{1+\gamma}} b, 1) \rvert$. This distribution also arises in the regression literature, where worst-case bias is incorporated into the construction of confidence intervals \citep{armstrong2020bias, armstrong2022robust}.

In the special case when $b=0$, we have $\hat{L}_{\rm PW}(0, \zeta, \gamma) = c_{\zeta/2}$, which yields the length of the confidence interval based on $\htau_{\rm PW}$ to be $2 c_{\zeta/2} (1+\gamma)^{-1/2} \sigma_0$. For comparison, recall that the length of the confidence interval solely based on $\htau_0$ is $2c_{\zeta/2} \sigma_0$. Thus the length of the confidence interval reduces by a factor of $(1+\gamma)^{-1/2}$ when using $\htau_{\rm PW}$ instead of $\htau_0$, which is small when $\gamma$ is large, i.e., $\sigma_0^2$ is much larger than $\sigma_1^2$. This indicates that by using a more precise but potentially biased estimator, if we have prior knowledge that the bias is small, then we can achieve a much shorter confidence interval. Given $\zeta$ and $\gamma$, the $\hat{L}_{\rm PW}(b, \zeta, \gamma)$ in Theorem~\ref{thm:L_PW} increases as $b$ increases, and goes to infinity as $b \to \infty$. Therefore, the precision-weighted estimator is not robust to unknown bias since the confidence interval is not bounded as the bias diverges.

One interesting observation is that $\hat{L}_{\rm PW}$ depends on $\sigma_0^2$ and $\sigma_1^2$ only through $\gamma$, the variance ratio. This explains why the bias $\Delta$ is scaled by $\sigma_0$ in the definition of $b$, and the confidence interval is represented as $[\htau_{\rm PW} - \hat{L}_{\rm PW} (1+\gamma)^{-1/2} \sigma_0, \htau_{\rm PW} + \hat{L}_{\rm PW} (1+\gamma)^{-1/2} \sigma_0]$.

% In Figure~\ref{fig:CI}, we illustrate how the confidence interval $[\htau_{\rm PW} - \hat{L}_{\rm PW} (1+\gamma)^{-1/2} \sigma_0, \htau_{\rm PW} + \hat{L}_{\rm PW} (1+\gamma)^{-1/2} \sigma_0]$ changes with the bias bound $b$. For comparison, we also include the confidence interval based on $\htau_0$, i.e., $[\htau_0 - c_{\zeta/2} \sigma_0,\; \htau_0 + c_{\zeta/2}\sigma_0]$. Assume $\sigma_0^2 = 1$ and $\htau_0 = 1$. Let $\zeta = 0.05$. We set $\htau_1 = 2$ and examine two scenarios: $\gamma = 10$ and $\gamma = 100$. We observe that as $b$ increases, the confidence interval based on $\htau_{\rm PW}$ widens, while the confidence interval based on $\htau_0$ remains fixed.

% \begin{figure}[htbp]
%   \centering
%   \includegraphics[width=0.9\linewidth]{figs/CI.png}
%   \caption{Confidence Interval against Maximum Bias $b$, where the left panel has $\gamma = 10$ and the right panel has $\gamma = 100$.}
%   \label{fig:CI}
% \end{figure}

\subsection{Confidence interval based on the pretest estimator}

We now construct the confidence interval for the pretest estimator $\htau_{\rm PT}$. In the following lemma, we present the distribution of $\hat\tau_{\rm PT} - \tau$. 

\begin{lemma}\label{lemma:clt_PT}
  Let $\cZ_1, \cZ_2$ be two independent standard normal random variables. Then $\hat\tau_{\rm PT} - \tau$ is distributed as
\begin{align*}
  \frac{1}{\sqrt{1+\gamma}}\sigma_0\cZ_1 +  \frac{\gamma}{1+\gamma} \Delta\ind\left(\left\lvert \cZ_2 + \sqrt{\frac{\gamma}{1+\gamma}}\frac{\Delta}{\sigma_0}\right\rvert \le c_{\alpha/2}\right) - \sqrt{\frac{\gamma}{1+\gamma}}\sigma_0\cZ_2\ind\left(\left\lvert \cZ_2 + \sqrt{\frac{\gamma}{1+\gamma}}\frac{\Delta}{\sigma_0}\right\rvert > c_{\alpha/2}\right).
\end{align*} 
\end{lemma}

By Lemma~\ref{lemma:clt_PT}, the distribution of $\hat\tau_{\rm PT}$ is a mixture of a normal distribution and a truncated normal distribution. We define $\hat{L}_{\rm PT} = \hat{L}_{\rm PT}(b, \zeta, \sigma_0^2, \sigma_1^2, \alpha)$ as the smallest length such that the confidence interval $[\htau_{\rm PT} - \hat{L}_{\rm PT} (1+\gamma)^{-1/2} \sigma_0, \htau_{\rm PT} + \hat{L}_{\rm PT} (1+\gamma)^{-1/2} \sigma_0]$ achieves correct coverage for all $\Delta$ satisfying $\lvert \Delta/\sigma_0 \rvert \le b$, where we choose the same scale as in Theorem~\ref{thm:L_PW} for comparison. By Definition~\ref{def:CI}, we can formulate $\hat{L}_{\rm PT}$ as the optimization problem
\begin{align}\label{eq:L_PT}
  \hat{L}_{\rm PT} = \hat{L}_{\rm PT}(b, \zeta, \sigma_0^2, \sigma_1^2, \alpha) = \inf \left\{ L \ge 0: \inf_{\Delta:\lvert \Delta/\sigma_0 \rvert \le b} \P_\Delta(\lvert \htau_{\rm PT} - \tau \rvert \le L (1+\gamma)^{-1/2} \sigma_0) \ge 1-\zeta\right\}.
\end{align}

However, $\hat{L}_{\rm PT}$ generally does not admit a closed-form expression. Moreover, direct computation of $\hat{L}_{\rm PT}$ based on \eqref{eq:L_PT} is computationally challenging since for each $L \ge 0$, the optimization problem \eqref{eq:L_PT} involves finding the infimum of $\P_\Delta(\lvert \htau_{\rm PT} - \tau \rvert \le L (1+\gamma)^{-1/2} \sigma_0)$ over all $\Delta$ satisfying $\lvert \Delta/\sigma_0 \rvert \le b$. Nevertheless, we show that computing $\hat{L}_{\rm PT}$ is tractable due to the following theorem:
\begin{theorem}\label{thm:L_PT}
    For any $L > 0$, $\P_\Delta(\lvert \htau_{\rm PT} - \tau \rvert \le L (1+\gamma)^{-1/2} \sigma_0)$ as a function of $\Delta$ is symmetric about $\Delta = 0$. Then $\hat{L}_{\rm PT} = \hat{L}_{\rm PT}(b, \zeta, \gamma, \alpha)$ in \eqref{eq:L_PT} is the solution to the following equation of $L$:
    \begin{align*}
      \min_{0 \le t \le b}\P_{\Delta/\sigma_0 = t}(\lvert \htau_{\rm PT} - \tau \rvert \le L (1+\gamma)^{-1/2} \sigma_0) = 1-\zeta,
    \end{align*}
    where
    \begin{align}\label{eq:P_PT}
      &\P_{\Delta/\sigma_0 = t}(\lvert \htau_{\rm PT} - \tau \rvert \le L (1+\gamma)^{-1/2} \sigma_0) \nonumber\\
      =& \Big[\Phi\Big(c_{\alpha/2} - \sqrt{\frac{\gamma}{1+\gamma}}t \Big) - \Phi\Big(-c_{\alpha/2} - \sqrt{\frac{\gamma}{1+\gamma}}t \Big)\Big]\Big[\Phi\Big(L - \frac{\gamma}{\sqrt{1+\gamma}}t \Big) - \Phi\Big(-L - \frac{\gamma}{\sqrt{1+\gamma}}t \Big)\Big]\nonumber\\
      &+ \int_{-\infty}^{-c_{\alpha/2} - \sqrt{\frac{\gamma}{1+\gamma}}t} \Big[\Phi\Big(L + \sqrt{\gamma} u \Big) - \Phi\Big(-L + \sqrt{\gamma} u\Big)\Big] \phi(u) \d u\nonumber\\
      &+ \int_{c_{\alpha/2} - \sqrt{\frac{\gamma}{1+\gamma}}t}^\infty \Big[\Phi\Big(L + \sqrt{\gamma} u\Big) - \Phi\Big(-L + \sqrt{\gamma} u\Big)\Big] \phi(u) \d u.
    \end{align}
    $[\htau_{\rm PT} - \hat{L}_{\rm PT} (1+\gamma)^{-1/2} \sigma_0, \htau_{\rm PT} + \hat{L}_{\rm PT} (1+\gamma)^{-1/2} \sigma_0]$ is the shortest length symmetric centered confidence interval based on $\htau_{\rm PT}$ for $\tau$ satisfying \eqref{eq:coverage}.
\end{theorem}

We have two main observations from Theorem~\ref{thm:L_PT}. First, for any $L>0$, the coverage probability $\P_\Delta(\lvert \htau_{\rm PT} - \tau \rvert \le L (1+\gamma)^{-1/2} \sigma_0)$ is symmetric about $\Delta=0$. As a result, instead of minimizing over all $\Delta$ satisfying $\lvert \Delta/\sigma_0 \rvert \le b$ as in \eqref{eq:L_PT}, we can restrict attention to the interval $0 \le \Delta/\sigma_0 \le b$. Second, Theorem~\ref{thm:L_PT} provides an explicit expression for the coverage probability $\P_{\Delta/\sigma_0 = t}(\lvert \htau_{\rm PT} - \tau \rvert \le L (1+\gamma)^{-1/2} \sigma_0)$ as a function of $t$ and $L$. The first term in \eqref{eq:P_PT} corresponds to the event that the pretest fails to reject the null hypothesis ($\Delta = 0$), in which case $\hat\tau_{\rm PT}$ reduces to the precision-weighted estimator. The second and third terms in \eqref{eq:P_PT} integrate over the regions where the pretest rejects the null hypothesis, in which case $\hat\tau_{\rm PT}$ reduces to the unbiased estimator.

Although the probability $\P_{\Delta/\sigma_0 = t}(\lvert \htau_{\rm PT} - \tau \rvert \le L (1+\gamma)^{-1/2} \sigma_0)$ admits a closed form for any $L > 0$, the minimum of $\P_{\Delta/\sigma_0 = t}(\lvert \htau_{\rm PT} - \tau \rvert \le L (1+\gamma)^{-1/2} \sigma_0)$ over $0 \le t \le b$ is not explicitly available in closed form, and thus must be computed numerically. Since the function $\P_{\Delta/\sigma_0 = t}(\lvert \htau_{\rm PT} - \tau \rvert \le L (1+\gamma)^{-1/2} \sigma_0)$ is monotonically increasing with respect to $L$ for $L \ge 0$, $L^*_{\rm PT}$ can be computed numerically using, for example, the bisection method, combined with numerical computation of the minimum of $\P_{\Delta/\sigma_0 = t}(\lvert \htau_{\rm PT} - \tau \rvert \le L (1+\gamma)^{-1/2} \sigma_0)$ with respect to $t$ for each candidate $L$. 

\subsection{Confidence interval based on the soft-thresholding estimator}

We now construct the confidence interval for the soft-thresholding estimator $\htau_{\rm ST}$. In the following lemma, we present the distribution of $\hat\tau_{\rm ST} - \tau$.

\begin{lemma}\label{lemma:clt_ST}
  Let $\cZ_1, \cZ_2$ be two independent standard normal random variables. Then $\hat\tau_{\rm ST} - \tau$ is distributed as
\begin{align*}
  &\frac{1}{\sqrt{1+\gamma}}\sigma_0\cZ_1 +  \frac{\gamma}{1+\gamma} \Delta\ind\left(\left\lvert \cZ_2 + \sqrt{\frac{\gamma}{1+\gamma}}\frac{\Delta}{\sigma_0}\right\rvert \le c_{\alpha/2}\right) \\
  &- \sqrt{\frac{\gamma}{1+\gamma}} \sigma_0\left(\cZ_2 - c_{\alpha/2}\sign\left(\cZ_2 + \sqrt{\frac{\gamma}{1+\gamma}}\frac{\Delta}{\sigma_0}\right)\right)\ind\left(\left\lvert \cZ_2 + \sqrt{\frac{\gamma}{1+\gamma}}\frac{\Delta}{\sigma_0}\right\rvert > c_{\alpha/2}\right).
\end{align*}
\end{lemma}

By Lemma~\ref{lemma:clt_ST}, the distribution of $\hat\tau_{\rm ST}$ is a mixture of a normal distribution and a truncated normal distribution. Suppose $\lvert \Delta/\sigma_0 \rvert \leq b$ for some $b > 0$. We seek the shortest length $\hat{L}_{\rm ST} = \hat{L}_{\rm ST}(b, \zeta, \sigma_0^2, \sigma_1^2, \alpha)$ such that the confidence interval $[\htau_{\rm ST}- \hat{L}_{\rm ST} (1+\gamma)^{-1/2} \sigma_0, \htau_{\rm ST}+ \hat{L}_{\rm ST} (1+\gamma)^{-1/2} \sigma_0]$ achieves correct coverage uniformly over all $\Delta$ with $\lvert \Delta/\sigma_0 \rvert \le b$. By Definition~\ref{def:CI}, we can formulate $\hat{L}_{\rm ST}$ as the optimization problem:
\begin{align}\label{eq:L}
  \hat{L}_{\rm ST} = \inf \left\{ L \ge 0: \inf_{\Delta:\lvert \Delta/\sigma_0 \rvert \le b} \P_\Delta(\lvert \htau_{\rm ST}- \tau \rvert \le L (1+\gamma)^{-1/2} \sigma_0) \ge 1-\zeta\right\}. 
\end{align}

However, $\hat{L}_{\rm ST}$ generally does not admit a closed-form expression. Moreover, direct computation of $\hat{L}_{\rm ST}$ based on \eqref{eq:L} is computationally challenging since for each $L \ge 0$, the optimization problem \eqref{eq:L} involves finding the infimum of $\P_\Delta(\lvert \htau_{\rm ST} - \tau \rvert \le L (1+\gamma)^{-1/2} \sigma_0)$ over all $\Delta$ satisfying $\lvert \Delta/\sigma_0 \rvert \le b$. Nevertheless, we show that $\hat{L}_{\rm ST}$ can be computed efficiently due to the following theorem:
\begin{theorem}\label{thm:L_ST}
  For any $L > 0$, $\P_\Delta(\lvert \htau_{\rm ST}- \tau \rvert \le L (1+\gamma)^{-1/2} \sigma_0)$ as a function of $\Delta$ is symmetric about $\Delta = 0$ and monotonically decreasing in $\lvert \Delta \rvert$. Then $\hat{L}_{\rm ST} = \hat{L}_{\rm ST}(b, \zeta, \gamma, \alpha)$ in \eqref{eq:L} is the solution to the following equation of $L$:
  \begin{align*}
    \P_{\Delta/\sigma_0 = b}(\lvert \htau_{\rm ST}- \tau \rvert \le L (1+\gamma)^{-1/2} \sigma_0) = 1-\zeta,
  \end{align*} 
  where
  \begin{align}\label{eq:P_ST}
    &\P_{\Delta/\sigma_0 = t}(\lvert \htau_{\rm ST}- \tau \rvert \le L (1+\gamma)^{-1/2} \sigma_0) \nonumber\\
    =& \Big[\Phi\Big(c_{\alpha/2} - \sqrt{\frac{\gamma}{1+\gamma}}t \Big) - \Phi\Big(-c_{\alpha/2} - \sqrt{\frac{\gamma}{1+\gamma}}t \Big)\Big]\Big[\Phi\Big(L - \frac{\gamma}{\sqrt{1+\gamma}}t \Big) - \Phi\Big(-L - \frac{\gamma}{\sqrt{1+\gamma}}t \Big)\Big]\nonumber\\
    &+ \int_{-\infty}^{-c_{\alpha/2} - \sqrt{\frac{\gamma}{1+\gamma}}t} \Big[\Phi\Big(L + \sqrt{\gamma} (u+c_{\alpha/2})\Big) - \Phi\Big(-L + \sqrt{\gamma} (u+c_{\alpha/2})\Big)\Big] \phi(u) \d u\nonumber\\
    &+ \int_{c_{\alpha/2} - \sqrt{\frac{\gamma}{1+\gamma}}t}^\infty \Big[\Phi\Big(L + \sqrt{\gamma} (u-c_{\alpha/2})\Big) - \Phi\Big(-L + \sqrt{\gamma} (u-c_{\alpha/2})\Big)\Big] \phi(u) \d u.
  \end{align}
  $[\htau_{\rm ST} - \hat{L}_{\rm ST} (1+\gamma)^{-1/2} \sigma_0, \htau_{\rm ST} + \hat{L}_{\rm ST} (1+\gamma)^{-1/2} \sigma_0]$ is the shortest length symmetric centered confidence interval based on $\htau_{\rm ST}$ for $\tau$ satisfying \eqref{eq:coverage}.
\end{theorem}

We have two main observations from Theorem~\ref{thm:L_ST}. First, for any $L>0$, the coverage probability $\P_\Delta(\lvert \htau_{\rm ST} - \tau \rvert \le L (1+\gamma)^{-1/2} \sigma_0)$ is symmetric about $\Delta=0$, unlike the pretest estimator, monotonically decreasing in $\lvert \Delta \rvert$. This monotonicity implies that the worst-case coverage over the bias $\Delta$ satisfying $\lvert \Delta/\sigma_0 \rvert \le b$ as in \eqref{eq:L} is always attained at the boundary $\Delta/\sigma_0 = b$. The monotonicity of the coverage probability $\P_\Delta(\lvert \htau_{\rm ST} - \tau \rvert \le L (1+\gamma)^{-1/2} \sigma_0)$ makes the computation of $\hat{L}_{\rm ST}$ more efficient than that of $\hat{L}_{\rm PT}$. Second, Theorem~\ref{thm:L_ST} provides an explicit expression for the coverage probability $\P_{\Delta/\sigma_0 = t}(\lvert \htau_{\rm ST} - \tau \rvert \le L (1+\gamma)^{-1/2} \sigma_0)$ as a function of $t$ and $L$. The first term in \eqref{eq:P_ST} corresponds to the event that the pretest fails to reject the null hypothesis ($\Delta = 0$), in which case $\hat\tau_{\rm ST}$ reduces to the precision-weighted estimator. The second and third terms in \eqref{eq:P_ST} integrate over the regions where the pretest rejects the null hypothesis, in which case $\hat\tau_{\rm ST}$ reduces to the unbiased estimator with a constant shift to ensure continuity at the pretest boundary.

Since $\P_{\Delta/\sigma_0 = b}(\lvert \htau_{\rm ST}- \tau \rvert \le L (1+\gamma)^{-1/2} \sigma_0)$ is monotonically increasing in $L$, $\hat{L}_{\rm ST}$ can be efficiently computed by, for example, the bisection method.

\subsection{Comparison of the point estimators and confidence intervals}

Compared with $\htau_{\rm PW}$, $\htau_{\rm ST}$ has bounded worst-case risk. The mean squared error of $\htau_{\rm ST}$ remains bounded regardless of the magnitude of the bias, whereas the mean squared error of $\htau_{\rm PW}$ grows without bound as the bias increases. Thus, $\htau_{\rm ST}$ offers a more balanced compromise between efficiency and robustness: it retains efficiency comparable to $\htau_{\rm PW}$ when the bias is small, yet its performance comparable to the unbiased estimator when the bias is large.

Compared with $\htau_{\rm PT}$, $\htau_{\rm ST}$ enjoys a desirable monotonicity property: for any $L > 0$, the probability $\P_{\Delta}(\lvert \htau_{\rm ST}- \tau \rvert \le L)$ decreases monotonically in $\lvert \Delta \rvert$ (Theorem~\ref{thm:L_ST}). This behavior parallels the result in \cite{bickel1983minimax}, which shows that the mean squared error of $\htau_{\rm ST}$ increases monotonically in $\lvert \Delta \rvert$. This monotonicity of $\htau_{\rm ST}$ allows us to compute the confidence interval based on $\htau_{\rm ST}$ more efficiently than $\htau_{\rm PT}$, as shown in Theorem~\ref{thm:L_ST}.

Figure~\ref{fig:length} compares the confidence intervals based on $\htau_{\rm ST}$, $\htau_{\rm PT}$, $\htau_{\rm PW}$, and the unbiased estimator $\htau_0$, as the bias bound $b$ varies. We fix $\sigma_0^2 = 1$ and $\htau_0 = 1$. Let $\zeta = 0.05$ and $\alpha = 0.05$. We set $\htau_1 = 2$ and examine two scenarios: $\gamma = 10$ and $\gamma = 100$. We observe that when the bias bound $b$ is small, the confidence intervals based on $\htau_{\rm ST}$, $\htau_{\rm PT}$, and $\htau_{\rm PW}$ are all shorter than that based on $\htau_0$. Importantly, when the bias bound $b$ is small, the confidence interval based on $\htau_{\rm ST}$ is shorter than $\htau_{\rm PT}$ and is comparable to $\htau_{\rm PW}$. Furthermore, the confidence interval length based on $\htau_{\rm ST}$ remains bounded, whereas that based on $\htau_{\rm PW}$ grows without bound as $b$ increases. Therefore, $\htau_{\rm ST}$ combines the efficiency gains of precision-weighted methods when bias is small with robustness to large biases, making it a superior choice for inference in practice.

\begin{figure}[htbp]
  \centering
  \includegraphics[width=\linewidth]{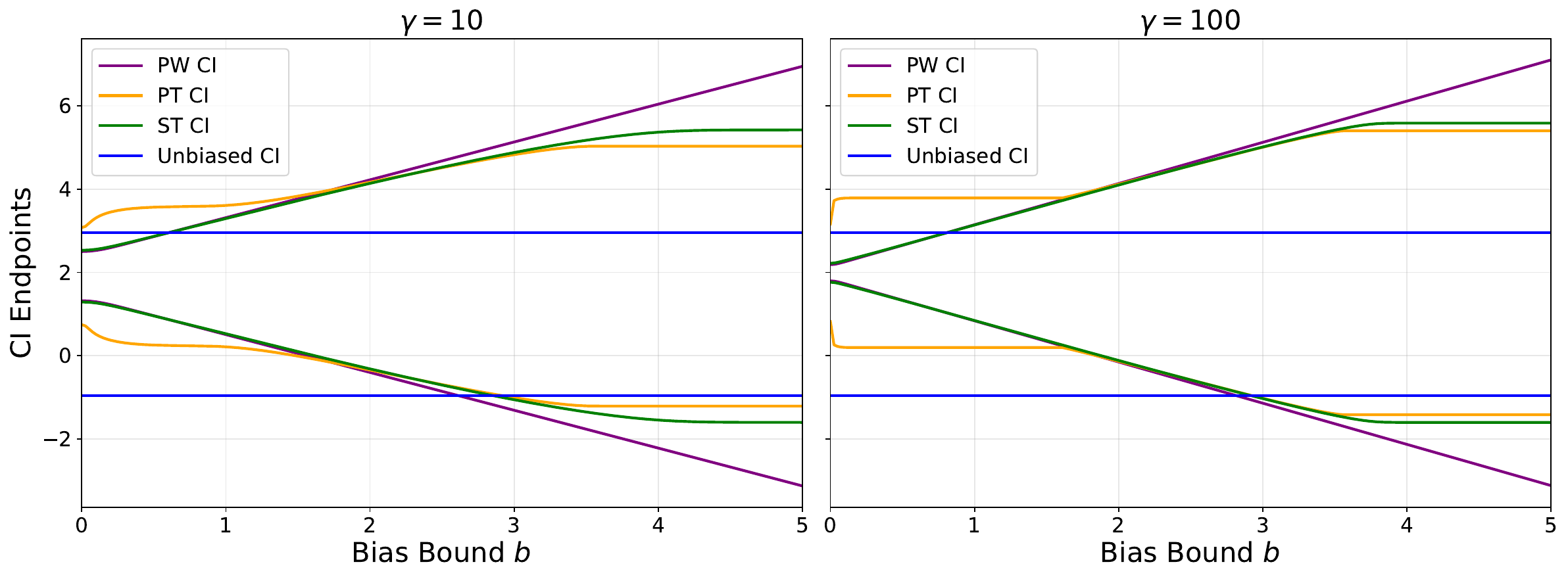}
  \caption{Confidence intervals against the maximum relative bias $\lvert \Delta/\sigma_0 \rvert \le b$}
  \label{fig:length}
\end{figure}

\begin{remark}[Computing the b-value]
  After plotting the estimators and their confidence intervals as in Figure~\ref{fig:length}, we are ready to read the b-value for each estimator given the significance level $\zeta$. Since the sequence of confidence intervals is different for different estimators, the b-values are also different. Alternatively, we can also compute the b-value directly. A naive way is to use bisection method by the monotonicity of the confidence interval, which involves computing the confidence interval for each possible bias level $b$, and see whether the confidence interval contains the null value. However, this procedure is computationally heavy since the computation of the confidence interval for a given $b$ also involves the bisection method when based on $\htau_{\rm PT}$ and $\htau_{\rm ST}$. In Appendix~\ref{sec:bvalue}, we propose a method to compute the b-value efficiently.
\end{remark}

\subsection{Generalization to the dependent case}\label{sec:dependent}

So far we assume independent unbiased and biased estimators. In this section, we generalize the above discussion to the case where $\htau_0$ and $\htau_1$ are jointly normal with known correlation $\rho$. We assume $\rho \sigma_1 \neq \sigma_0$, which trivially holds when $\sigma_0^2 > \sigma_1^2$, or equivalently $\gamma > 1$. The key is to construct a biased estimator which is independent of $\htau_0$ without losing information of $\htau_1$. To achieve that, we define the reparametrization
\begin{align}\label{eq:htau1_prime}
  \htau_1' = \frac{\htau_1 - (\rho \sigma_1/\sigma_0)\htau_0}{1 - \rho \sigma_1/\sigma_0}.
\end{align}
Then, $\htau_0$ and $\htau_1'$ are independent, with $\htau_1' \sim N(\tau + \Delta', \sigma_1'^2)$, where 
\[
  \Delta' = \frac{\Delta}{1 - \rho \sigma_1/\sigma_0}, \quad \sigma_1'^2 = \frac{(1-\rho^2)\sigma_1^2}{(1-\rho \sigma_1/\sigma_0)^2}.
\] 
Thus, the problem reduces to the independent case considered above.
% If $\rho \sigma_1 = \sigma_0$, $\htau_1'$ is not well defined. The reason is the precision-weighted estimator reduces to $\htau_0$ in this case, meaning we exclusively use $\htau_0$ regardless of the value of $\Delta$, or $\htau_1$ does not provide additional information beyond $\htau_0$.

In the reparametrization in \eqref{eq:htau1_prime}, we need to know how the transformation rescales the bias and how to interpret the relative bias. We can compute the confidence intervals and the b-value based on the reparametrization with independent unbiased estimator $\htau_0$ and transformed biased estimator $\htau_1'$. Let the confidence intervals and the b-value in the transformed problem be $\htau - \cI'(b, \zeta)$ and $b^{*’}$, respectively. Then the confidence intervals for the original problem are $\htau - \cI(b, \zeta) = \htau - \cI'(b/\lvert 1 - \rho \sigma_1/\sigma_0 \rvert, \zeta)$ and the b-value is $b^* = \lvert 1 - \rho \sigma_1/\sigma_0 \rvert b^{*'}$, respectively. We relegate the technical details to Appendix~\ref{sec:dependence}.

\section{Generalization to the multivariate case}\label{sec:multivariate}

In this section, we extend our framework to the multivariate case. Many applications involve vector-valued parameters rather than scalars. Consider two leading examples in causal inference. First, when a treatment has multiple levels, the parameter of interest is a vector of treatment effects across those levels. This setting also includes factorial designs, where researchers aim to estimate factorial effects jointly. Second, when the population is partitioned into subgroups, the focus is often on subgroup treatment effects, yielding a vector of conditional average treatment effects indexed by the subgroup variable \citep{schwartz2025harmonized}. Beyond causal inference, similar issues arise in regression settings. For instance, when combining OLS and IV estimators, the target parameter can be a vector of regression coefficients. Therefore, generalization to the multivariate case is essential for applications where researchers must make inferences about multiple parameters.

\subsection{Setup}

We consider the following multivariate setting:
\begin{assumption}\label{ass:normal_multi}
  Suppose we observe two independent random vectors: one unbiased estimator $\bm{\htau}_0 \sim N(\bm{\tau}, \bm{\Sigma}_0)$ and one biased estimator $\bm{\htau}_1 \sim N(\bm{\tau} + \bm{\Delta}, \bm{\Sigma}_1)$. Here $\bm{\tau} \in \mathbb{R}^d$ is the unknown parameter of interest. We assume that $\bm{\Sigma}_0$ and $\bm{\Sigma}_1$ are known but $\bm{\Delta}$ is unknown.
\end{assumption}

As in the univariate case, we present our analysis in the exact normality setting which asymptotically captures the essential structure of inference as in \cite{wald1943tests} and \cite{le1956asymptotic}. Our goal is to construct confidence regions for $\bm{\tau}$ at different bias levels. Although we assume independent unbiased and biased estimators, the extension to the dependent case is straightforward; see Section~\ref{sec:dependent}. Since $\bm{\htau}_0$ is unbiased for $\bm{\tau}$, given significance level $\zeta \in (0,1)$, a natural confidence region based solely on $\bm{\htau}_0$ is the ellipsoid:
\begin{align}
  \{\bm{\tau} \in \mathbb{R}^d : (\bm{\htau}_0 - \bm{\tau})^\top \bm{\Sigma}_0^{-1} (\bm{\htau}_0 - \bm{\tau}) \leq \chi^2_{d,1-\zeta}\},
\end{align}
where $\chi^2_{d,1-\zeta}$ is the $(1-\zeta)$ upper quantile of the chi-squared distribution with $d$ degrees of freedom.

Then we consider combining the two estimators $\bm{\htau}_0$ and $\bm{\htau}_1$, and we consider a generic combined estimator $\bm{\hat\tau} = \bm{\hat\tau}(\bm{\htau}_0, \bm{\htau}_1, \bm{\Sigma}_0, \bm{\Sigma}_1)$. We assume that $\lvert [\bm{\Sigma}^{-1/2} \bm{\Delta}]_j \rvert \le b_j$ for some $b_j \ge 0$ for all $j=1,2,\ldots,d$ and study how the confidence region changes with the maximum relative bias vector $\bm{b} = (b_1, b_2, \ldots, b_d)^\top$. Here $\bm{\Sigma} = \bm{\Sigma}(\bm{\Sigma}_0, \bm{\Sigma}_1) \in \bR^{d \times d}$ can be any fixed positive definite scaling matrix that depends on $\bm{\Sigma}_0$ and $\bm{\Sigma}_1$, which corresponds to $\sigma_0^2$ we used in the univariate case. For a given bias vector $\bm{b}$, different choices of $\bm{\Sigma}$ correspond to different regions for the bias $\Delta$. In this paper, we do not discuss how to choose the scaling matrix $\bm{\Sigma}$ optimally. For simplicity, one may take $\bm{\Sigma} = \bm{I}_d$, where $\bm{I}_d$ is the $d \times d$ identity matrix, or set $\bm{\Sigma} = \bm{\Sigma}_0$ as in the univariate case. Analogous to Definition~\ref{def:CI} for the univariate case, the confidence region is defined below.

\begin{definition}\label{def:CI_multi}
Given a significance level $\zeta \in (0,1)$ and the maximum relative bias vector $\bm{b}$ with $b_j \ge 0$ for all $j=1,2,\ldots,d$, we seek a region $\bm{\cI}(\bm{b}, \zeta) = \bm{\cI}(\bm{b}, \zeta, \bm{\htau}_0, \bm{\htau}_1, \bm{\Sigma}_0, \bm{\Sigma}_1)$ such that
\begin{align}\label{eq:coverage_multi}
  \inf_{\bm{\Delta}: \lvert \bm{\Sigma}^{-1/2} \bm{\Delta}\rvert \le \bm{b}} \P_{\bm{\Delta}}(\bm{\tau} \in \bm{\hat\tau} - \bm{\cI}(\bm{b}, \zeta)) 
  = \inf_{\bm{\Delta}: \lvert \bm{\Sigma}^{-1/2} \bm{\Delta}\rvert \le \bm{b}} \P_{\bm{\Delta}}(\bm{\hat\tau} - \bm{\tau} \in \bm{\cI}(\bm{b}, \zeta)) \ge 1 - \zeta.
\end{align}
The confidence region based on $\bm{\hat\tau}$ is then given by $\bm{\hat\tau}-\bm{\cI}(\bm{b},\zeta)$.
\end{definition}

Then we introduce the monotonicity conditions on the region $\bm{\cI}(\bm{b},\zeta)$ below, which generalizes Assumption~\ref{asp:monotonicity} to the multivariate case.

\begin{assumption}\label{asp:monotonicity_multi}
  We assume:
  \begin{enumerate}
    \item For fixed $(\bm{\htau}_0, \bm{\htau}_1, \bm{\Sigma}_0, \bm{\Sigma}_1)$ and $\zeta$, we have $\bm{\cI}(\bm{b}, \zeta) \subset \bm{\cI}(\bm{b'}, \zeta)$ whenever $\bm{b} \le \bm{b'}$, i.e., $b_j \le b'_j$ for all $j=1,2,\ldots,d$.
    \item For fixed $(\bm{\htau}_0, \bm{\htau}_1, \bm{\Sigma}_0, \bm{\Sigma}_1)$ and $\bm{b}$, we have $\bm{\cI}(\bm{b}, \zeta) \subset \bm{\cI}(\bm{b}, \zeta')$ whenever $\zeta \ge \zeta'$.
  \end{enumerate}
\end{assumption}

Assume $\bm{\cI}(\bm{b},\zeta)$ satisfies the monotonicity conditions of Assumption~\ref{asp:monotonicity_multi}. A common family of such regions is fixed-length centered ellipsoids: $\bm{\cI}(\bm{b}, \zeta) = \{\bm{h} \in \bR^d: \bm{h}^\top \bm{A}^{-1} \bm{h} \le c(\bm{b}, \zeta, \bm{\Sigma}_0, \bm{\Sigma}_1)\}$ with some constant $c(\bm{b}, \zeta, \bm{\Sigma}_0, \bm{\Sigma}_1) \ge 0$ and a positive definite matrix $\bm{A}$. In this case, the confidence region $\bm{\hat\tau}-\bm{\cI}(\bm{b},\zeta)$ in Definition~\ref{def:CI_multi} is given by $\{\bm{\tau} \in \bR^d: (\bm{\htau} - \bm{\tau})^\top \bm{A}^{-1} (\bm{\htau} - \bm{\tau}) \le c(\bm{b}, \zeta, \bm{\Sigma}_0, \bm{\Sigma}_1)\}$.

\begin{remark}
  Unlike the univariate case, the construction of confidence regions in the multivariate setting is more complicated. First, in one dimension, the only convex confidence set is an interval, whereas in higher dimensions, there exists a wide variety of admissible convex confidence regions. Although ellipsoidal regions enjoy optimality properties under specific conditions \citep{stein1962confidence, wald1949statistical}, the appropriate geometry depends on the family of contrasts and the norm used to measure uncertainty. Ellipsoidal regions based on the Mahalanobis distance \citep{hotelling1931generalization} arise naturally under the Gaussian model, but alternative geometries have been studied in the literature \citep{tukey1949comparing,scheffe1953method,vsidak1967rectangular}. Second, the optimal choice of the center of the confidence region is not straightforward. For $d \ge 3$, Stein’s phenomenon implies that recentering at shrinkage estimators such as the James–Stein estimator can yield confidence regions with smaller volume while maintaining nominal coverage \citep{stein1962confidence,hwang1982minimax,berger1985statistical}. In this section we follow one specific track: constructing ellipsoidal confidence regions under the Mahalanobis distance, centered at prespecified estimators, in parallel with our discussion for the univariate case in Section~\ref{sec:confidence}. We do not focus on the optimality among all possible shapes or centers, but rather focus on this formulation for its analytical tractability and interpretability within our framework.
\end{remark}

Then we define the multivariate b-value below.

\begin{definition}\label{def:b_multi}
  Define the b-value as the critical boundary $\bm{b}^*$ of testing $\bm{\tau} = \bm{0}$ versus $\bm{\tau} \neq \bm{0}$ as
  \begin{align}\label{eq:b_multi}
    \bm{b}^*(\zeta) = \bm{b}^*(\zeta, \bm{\htau}_0, \bm{\htau}_1, \bm{\Sigma}_0, \bm{\Sigma}_1) = \partial\left\{\bm{b} \ge \bm{0}: \bm{0} \in \bm{\hat\tau} - \bm{\cI}(\bm{b}, \zeta)\right\}.
  \end{align}
\end{definition}

Definition~\ref{def:b_multi} extends the univariate notion of the b-value (Definition~\ref{def:b}) to the multivariate setting. In higher dimensions $d > 1$, the set $\left\{\bm{b} \ge \bm{0}: \bm{0} \in \bm{\hat\tau} - \bm{\cI}(\bm{b}, \zeta)\right\}$ is a convex region in the first quadrant. The b-value is defined as the boundary of this region, $\partial\{\bm{b} \ge \bm{0}: \bm{0} \in \bm{\hat\tau} - \bm{\cI}(\bm{b}, \zeta)\}$, which is a $(d-1)$-dimensional surface, e.g., a curve when $d=2$, in the first quadrant. When $d=1$, this boundary reduces to the single left endpoint $\inf\left\{b \ge 0: 0 \in \hat\tau - \cI(b, \zeta)\right\}$, which is exactly Definition~\ref{def:b}. By the monotonicity conditions in Assumption~\ref{asp:monotonicity_multi}, the geometry of this boundary yields a natural decision rule: we reject the null hypothesis $\bm{\tau} = \bm{0}$ if the bias bound vector $\bm{b}$ lies to the left of the b-value surface, and fail to reject it if $\bm{b}$ lies on or to the right of the b-value surface. To compute and visualize the multivariate b-value, we can plot the estimators and their confidence regions, and then we are ready to read the b-value for each estimator given the significance level $\zeta$.

\subsection{Point estimation}\label{sec:point_multi}

First, we recall the precision-weighted estimator defined as
\[
  \bm{\htau}_{\rm PW} = (\bm{\Sigma}_0^{-1} + \bm{\Sigma}_1^{-1})^{-1}(\bm{\Sigma}_0^{-1} \bm{\htau}_0 + \bm{\Sigma}_1^{-1} \bm{\htau}_1) = \bm{\htau}_0 + (\bm{\Sigma}_0^{-1} + \bm{\Sigma}_1^{-1})^{-1} \bm{\Sigma}_1^{-1} (\bm{\htau}_1 - \bm{\htau}_0).
\]
The precision-weighted estimator is the maximum likelihood estimator of $\bm{\tau}$ if $\bm{\Delta}$ is known to be all zeros. However, its risk is large when $\lVert \bm{\Delta} \rVert_2$ is large. 
% To see that, since $\bm{\htau}_{\rm PW} \sim N(\bm{\tau} + (\bm{\Sigma}_0^{-1} + \bm{\Sigma}_1^{-1})^{-1}\bm{\Sigma}_1^{-1}\bm{\Delta}, (\bm{\Sigma}_0^{-1} + \bm{\Sigma}_1^{-1})^{-1})$, the bias of $\bm{\htau}_{\rm PW}$ in estimating $\bm{\tau}$ is $ (\bm{\Sigma}_0^{-1} + \bm{\Sigma}_1^{-1})^{-1}\bm{\Sigma}_1^{-1}\bm{\Delta}$, and its $L_2$ risk is
% \[
% \E_{\bm{\Delta}}[\lVert \bm{\htau}_{\rm PW} - \bm{\tau} \rVert_2^2] = \lVert (\bm{\Sigma}_0^{-1} + \bm{\Sigma}_1^{-1})^{-1}\bm{\Sigma}_1^{-1}\bm{\Delta} \rVert_2^2 + {\rm Tr}((\bm{\Sigma}_0^{-1} + \bm{\Sigma}_1^{-1})^{-1}),
% \]
% which diverges as $\lVert \bm{\Delta} \rVert_2 \to \infty$. 
Thus, we seek a combined estimator that performs nearly as well as $\bm{\htau}_{\rm PW}$ when the bias is small, and is robust to unknown bias $\bm{\Delta}$, ensuring that the maximum risk $\sup_{\bm{\Delta} \in \bR^d} \E_{\bm{\Delta}}[\lVert \bm{\htau}_{\rm PW} - \bm{\tau} \rVert_2^2]$ remains bounded. To develop estimators that perform well under zero bias yet remain robust to unknown bias, various estimators have been proposed in the empirical Bayes and shrinkage estimation literature \citep{j1981estimation,bickel1984parametric,green1991james,green2005improved,rosenman2023combining,rosenman2023empirical}. In this paper, we focus on the generic pretest estimator and the generic soft-thresholding estimator, as discussed below.

Second, we recall the pretest estimator, which involves incorporating a pretest for $\bm{\Delta} = \bm{0}$ versus $\bm{\Delta} \neq \bm{0}$. Under $\bm{\Delta} = \bm{0}$, given the independence between $\bm{\htau}_1$ and $\bm{\htau}_0$, their difference follows: $\bm{\htau}_1 - \bm{\htau}_0 \sim N(\bm{0}, \bm{\Sigma}_0 + \bm{\Sigma}_1)$. We consider the test statistic $\lVert (\bm{\Sigma}_0+ \bm{\Sigma}_1)^{-1/2}(\bm{\htau}_1 - \bm{\htau}_0) \rVert_2^2$ and the critical value $q \ge 0$. The critical value $q$ here plays a similar role as the significance level $\alpha$ in Section~\ref{sec:point}. Let $$\mathcal{A} = \{\lVert (\bm{\Sigma}_0+ \bm{\Sigma}_1)^{-1/2}(\bm{\htau}_1 - \bm{\htau}_0) \rVert_2^2 \le q\}$$ denote the event that the pretest fails to reject the null hypothesis ($\bm{\Delta} = \bm{0}$). If the pretest fails to reject the null hypothesis, i.e., $\lVert (\bm{\Sigma}_0+ \bm{\Sigma}_1)^{-1/2}(\bm{\htau}_1 - \bm{\htau}_0) \rVert_2^2 \le q$, the pretest estimator uses the precision-weighted estimator:
\[
  \bm{\htau}_0 + (\bm{\Sigma}_0^{-1} + \bm{\Sigma}_1^{-1})^{-1} \bm{\Sigma}_1^{-1} (\bm{\htau}_1 - \bm{\htau}_0).
\]
If the pretest rejects the null hypothesis, i.e., $\lVert (\bm{\Sigma}_0+ \bm{\Sigma}_1)^{-1/2}(\bm{\htau}_1 - \bm{\htau}_0) \rVert_2^2 > q$, the pretest estimator employs hard-thresholding by reverting to the unbiased estimator $\bm{\htau}_0$. Combining the two cases, the pretest estimator is:
\[
\bm{\htau}_{\rm PT} = \bm{\htau}_0 + (\bm{\Sigma}_0^{-1} + \bm{\Sigma}_1^{-1})^{-1} \bm{\Sigma}_1^{-1} (\bm{\htau}_1 - \bm{\htau}_0) \ind\left(\mathcal{A}\right).
\]

Third, we recall the soft-thresholding estimator, which ensures continuity at the pretest boundary ($\lVert (\bm{\Sigma}_0+ \bm{\Sigma}_1)^{-1/2}(\bm{\htau}_1 - \bm{\htau}_0) \rVert_2^2 = q$). If the pretest fails to reject the null hypothesis, the soft-thresholding estimator also uses the precision-weighted estimator. If the pretest rejects the null hypothesis, the soft-thresholding estimator employs soft-thresholding by:
\[
  \bm{\htau}_0 + h_q^*(\lVert (\bm{\Sigma}_0+ \bm{\Sigma}_1)^{-1/2}(\bm{\htau}_1 - \bm{\htau}_0) \rVert_2^2)(\bm{\Sigma}_0^{-1} + \bm{\Sigma}_1^{-1})^{-1}\bm{\Sigma}_1^{-1}(\bm{\htau}_1 - \bm{\htau}_0),
\]
where $h_q^*(\cdot): [q,\infty) \to [0,1]$ is a non-increasing function with $h_q^*(q)=1$. Combining the two cases, the soft-thresholding estimator is:
\[
  \bm{\hat\tau}_{\rm ST} = \bm{\htau}_0 + h_q(\lVert (\bm{\Sigma}_0+ \bm{\Sigma}_1)^{-1/2}(\bm{\htau}_1 - \bm{\htau}_0) \rVert_2^2)(\bm{\Sigma}_0^{-1} + \bm{\Sigma}_1^{-1})^{-1}\bm{\Sigma}_1^{-1}(\bm{\htau}_1 - \bm{\htau}_0),
\]
where $h_q(\cdot): [0,\infty) \to [0,1]$ is defined as $h_q(r) = \ind(0 \le r \le q) + h_q^*(r) \ind(r > q)$. Here both $q$ and $h_q^*(\cdot)$ can depend on $\bm{\Sigma}_0$ and $\bm{\Sigma}_1$.

The generic soft-thresholding estimator $\bm{\hat\tau}_{\rm ST}$ contains many classical estimators as special cases, such as the estimator studied in \citet[Theorem 3]{j1981estimation} and \citet[Section 4]{bickel1984parametric}, which generalizes the univariate soft-thresholding estimator $\hat\tau_{\rm ST}$ to the multivariate setting. We provide a detailed discussion of the relationship between the generic soft-thresholding estimator and the estimator in \citet{j1981estimation, bickel1984parametric} in Appendix~\ref{sec:efron--morris}. In this paper, we focus on general $\bm{\Sigma}_0$, $\bm{\Sigma}_1$, and $h_q^*(\cdot)$.

% In this paper, we do not impose structural assumptions on the covariance matrices $\bm{\Sigma}_0$ and $\bm{\Sigma}_1$ as \cite{j1981estimation} and \cite{bickel1984parametric}, and focus on general $\bm{\Sigma}_0$, $\bm{\Sigma}_1$, and $h_q^*(\cdot)$.

% The distribution of $\bm{\htau}$ is given in the following theorem.

% \begin{theorem}
%   Let $\cZ_1, \cZ_2$ be two independent random vectors from $N(\bm{0},\bm{I}_d)$. The distribution of $\bm{\htau}$ is
%   \begin{align*}
%   (\bm{\htau} - \bm{\tau}) \stackrel{\sf d}{=} (1+\gamma)^{-1/2}\sigma_0\cZ_1 +  (1+\gamma)^{-1} \gamma\bm{\Delta}\ind(\lVert \bm{\htau}_1 - \bm{\htau}_0 \rVert_2^2/\sigma_0^2 \le b) - \gamma^{1/2}(1+\gamma)^{-1/2} \sigma_0\cZ_2\ind(\lVert \bm{\htau}_1 - \bm{\htau}_0 \rVert_2^2/\sigma_0^2 > b)
%   \end{align*}
% \end{theorem}

\subsection{Confidence intervals}

First, we construct the confidence regions based on the precision-weighted estimator $\bm{\htau}_{\rm PW}$. Under Assumption~\ref{ass:normal_multi}, we have $\bm{\htau}_{\rm PW} \sim N(\bm{\tau} + (\bm{\Sigma}_0^{-1} + \bm{\Sigma}_1^{-1})^{-1} \bm{\Sigma}_1^{-1} \bm{\Delta}, (\bm{\Sigma}_0^{-1} + \bm{\Sigma}_1^{-1})^{-1} )$. By using the ellipsoidal confidence region under the Mahalanobis distance and the covariance matrix of $\bm{\htau}_{\rm PW}$, we consider the confidence region based on $\bm{\htau}_{\rm PW}$ with the following form:
\[
\{\bm{\tau}\in\mathbb{R}^d:(\bm{\htau}_{\rm PW}-\bm{\tau})^\top (\bm{\Sigma}_0^{-1} + \bm{\Sigma}_1^{-1})(\bm{\htau}_{\rm PW}-\bm{\tau})\le M\},
\]
for some $M \ge 0$. The following theorem explicitly characterizes the confidence region based on $\bm{\htau}_{\rm PW}$ as a function of bias bound $\bm{b}$.
\begin{theorem}\label{prop:CI_multi}
  Let $\hat{M}_{\rm PW} = \hat{M}_{\rm PW}(\bm{b}, \zeta, \bm{\Sigma}_0, \bm{\Sigma}_1, \bm{\Sigma}) \ge 0$ be the $(1-\zeta)$ upper quantile of the noncentral chi-squared distribution with $d$ degrees of freedom and noncentrality parameter $$\sup_{\bm{s} \in \{\pm 1\}^d }\left\lVert (\bm{\Sigma}_0^{-1} + \bm{\Sigma}_1^{-1})^{-1/2} \bm{\Sigma}_1^{-1} \bm{\Sigma}^{1/2} \bm{b} \odot \bm{s}\right\rVert_2^2.$$ The confidence region for $\bm{\tau}$ satisfying \eqref{eq:coverage_multi} is given by $\{\bm{\tau} \in \bR^d: (\bm{\htau}_{\rm PW} - \bm{\tau})^\top (\bm{\Sigma}_0^{-1} + \bm{\Sigma}_1^{-1}) (\bm{\htau}_{\rm PW} - \bm{\tau}) \le \hat{M}_{\rm PW}\}$.
\end{theorem}

Theorem~\ref{prop:CI_multi} extends the univariate result in Theorem~\ref{thm:L_PW} to the multivariate setting. In the univariate case, taking $\bm{\Sigma} = \sigma_0^2$ and $\bm{b} = b$ reduces the multivariate confidence region in Theorem~\ref{prop:CI_multi} exactly to the confidence interval in Theorem~\ref{thm:L_PW}. In the univariate case, the bias constraint $\{\Delta: \lvert \Delta/\sigma_0 \rvert \le b\}$ has only two boundary points, and the absolute bias of $\htau_{\rm PW}$ is the same at both endpoints. By contrast, the multivariate bias constraint $\{\bm{\Delta}: \lvert \bm{\Sigma}^{-1/2} \bm{\Delta}\rvert \le \bm{b}\}$ is a convex hyperrectangle with $2^d$ vertices, and the maximum $L_2$ norm of the bias of $\bm{\htau}_{\rm PW}$ may occur at any of these vertices. Theorem~\ref{prop:CI_multi} therefore characterizes the maximum possible bias of $\bm{\htau}_{\rm PW}$ over the entire bias constraint.

In the special case when $\bm{b}=0$, we have $\hat{M}_{\rm PW} = \chi^2_{d, 1-\zeta}$. Compared with the ellipsoid confidence region based on $\bm{\htau}_0$, the scaling matrix in the definition of ellipsoid confidence region reduces from $\bm{\Sigma}_0^{-1}$ to $\bm{\Sigma}_0^{-1} + \bm{\Sigma}_1^{-1} = (\bm{\Sigma}_1^{-1} \bm{\Sigma}_0 + \bm{I}_d)\bm{\Sigma}_0^{-1}$, analogous to the efficiency gain in the univariate case, where $\bm{\Sigma}_1^{-1} \bm{\Sigma}_0$ can be regarded as the variance ratio $\gamma$ in the univariate case. As $\bm{b}$ increases, $\hat{M}_{\rm PW}$ increases, and goes to infinity as $\bm{b} \to \infty$.

Second, we construct the confidence regions based on the generic pretest estimator $\bm{\hat\tau}_{\rm PT}$. Similar to the confidence regions based on the precision-weighted estimator, we consider the confidence regions based on $\bm{\hat\tau}_{\rm PT}$ with the following form:
\[
\{\bm{\tau}\in\mathbb{R}^d:(\bm{\htau}_{\rm PT}-\bm{\tau})^\top (\bm{\Sigma}_0^{-1} + \bm{\Sigma}_1^{-1})(\bm{\htau}_{\rm PT}-\bm{\tau})\le M\},
\]
for some $M \ge 0$. Here, we focus on the ellipsoid confidence region defined using the same scaling matrix, $\bm{\Sigma}_0^{-1} + \bm{\Sigma}_1^{-1}$, as that used in the precision-weighted estimator. This choice allows for a direct comparison with the confidence region based on the precision-weighted estimator through the upper bound $M$ in the confidence region. To ensure the confidence regions satisfy Definition~\ref{def:CI_multi}, we need to find the minimal value $\hat{M}_{\rm PT} = \hat{M}_{\rm PT}(\bm{b}, \zeta, \bm{\Sigma}_0, \bm{\Sigma}_1, \bm{\Sigma})$ such that the confidence region
\[
\{\bm{\tau}\in\mathbb{R}^d:(\bm{\htau}_{\rm PT}-\bm{\tau})^\top (\bm{\Sigma}_0^{-1} + \bm{\Sigma}_1^{-1})(\bm{\htau}_{\rm PT}-\bm{\tau})\le \hat{M}_{\rm PT}\}
\]
achieves correct coverage for all $\bm{\Delta}$ satisfying $\lvert [\bm{\Sigma}^{-1/2} \bm{\Delta}]_j \rvert \le b_j$ for all $j=1,2,\ldots,d$. We can formulate $\hat{M}_{\rm PT}$ as the optimization problem
\begin{align*}
  \hat{M}_{\rm PT} = \inf \left\{ M \ge 0: \inf_{\bm{\Delta}: \lvert \bm{\Sigma}^{-1/2} \bm{\Delta}\rvert \le \bm{b}} \P_{\bm{\Delta}}((\bm{\htau}_{\rm PT} - \bm{\tau})^\top (\bm{\Sigma}_0^{-1} + \bm{\Sigma}_1^{-1}) (\bm{\htau}_{\rm PT} - \bm{\tau}) \le M) \ge 1-\zeta \right\}.
\end{align*}

We show how to compute $\hat{M}_{\rm PT}$ in the following theorem.
\begin{theorem}\label{prop:CI_multi_PT}
  $\hat{M}_{\rm PT}$ is the solution to the equation of $M$:
  \begin{align*}
    \inf_{\bm{\Delta}: \lvert \bm{\Sigma}^{-1/2} \bm{\Delta}\rvert \le \bm{b}} \P_{\bm{\Delta}}((\bm{\htau}_{\rm PT} - \bm{\tau})^\top (\bm{\Sigma}_0^{-1} + \bm{\Sigma}_1^{-1}) (\bm{\htau}_{\rm PT} - \bm{\tau}) \le M) = 1-\zeta,
  \end{align*} 
  with an explicit form given by:
  \begin{align}\label{eq:P_multi_PT}
    &\P_{\bm{\Sigma}^{-1/2}\bm{\Delta} = \bm{t}}\left((\bm{\htau}_{\rm PT}-\bm{\tau})^\top (\bm{\Sigma}_0^{-1} + \bm{\Sigma}_1^{-1})(\bm{\htau}_{\rm PT}-\bm{\tau})\le M\right) \nonumber\\
    =& \Psi_d\left(M;\left\lVert (\bm{\Sigma}_0^{-1} + \bm{\Sigma}_1^{-1})^{-1/2} \bm{\Sigma}_1^{-1} \bm{\Sigma}^{1/2} \bm{t}\right\rVert_2^2\right) \Psi_d\left(q; \left\lVert (\bm{\Sigma}_0 + \bm{\Sigma}_1)^{-1/2} \bm{\Sigma}^{1/2} \bm{t}\right\rVert_2^2\right) \nonumber\\
    &+ \int_{\lVert \bm{u} \rVert_2^2 > q} \Psi_d\left(M;\left\lVert (\bm{\Sigma}_0^{-1} + \bm{\Sigma}_1^{-1})^{-1/2} \bm{\Sigma}_1^{-1} [\bm{\Sigma}^{1/2} \bm{t} - (\bm{\Sigma}_0+ \bm{\Sigma}_1)^{1/2} \bm{u}]\right\rVert_2^2\right) \phi_{(\bm{\Sigma}_0 + \bm{\Sigma}_1)^{-1/2} \bm{\Sigma}^{1/2} \bm{t}, \bm{I}_d}\left(\bm{u}\right) \d \bm{u}.
  \end{align}
\end{theorem}

Theorem~\ref{prop:CI_multi_PT} extends the univariate result in Theorem~\ref{thm:L_PT} to the multivariate setting. Theorem~\ref{prop:CI_multi_PT} provides an explicit expression for the coverage probability $\P_{\bm{\Sigma}^{-1/2}\bm{\Delta} = \bm{t}}((\bm{\htau}_{\rm PT}-\bm{\tau})^\top (\bm{\Sigma}_0^{-1} + \bm{\Sigma}_1^{-1})(\bm{\htau}_{\rm PT}-\bm{\tau})\le M)$ as a function of $M$ and $\bm{t}$. The first term in \eqref{eq:P_multi_PT} corresponds to the event that the pretest fails to reject the null hypothesis ($\bm{\Delta} = \bm{0}$), in which case $\bm{\htau}_{\rm PT}$ reduces to the precision-weighted estimator. The second term in \eqref{eq:P_multi_PT} integrates over the regions where the pretest rejects the null hypothesis, in which case $\bm{\htau}_{\rm PT}$ reduces to the unbiased estimator.

Third, we construct the confidence regions based on the generic soft-thresholding estimator $\bm{\hat\tau}_{\rm ST}$. Similar to the confidence regions based on the precision-weighted estimator, we consider the confidence regions based on $\bm{\hat\tau}_{\rm ST}$ with the following form:
\[
\{\bm{\tau}\in\mathbb{R}^d:(\bm{\htau}_{\rm ST}-\bm{\tau})^\top (\bm{\Sigma}_0^{-1} + \bm{\Sigma}_1^{-1})(\bm{\htau}_{\rm ST}-\bm{\tau})\le M\},
\]
for some $M \ge 0$. To ensure the confidence regions satisfy Definition~\ref{def:CI_multi}, we need to find the minimal value $\hat{M}_{\rm ST} = \hat{M}_{\rm ST}(\bm{b}, \zeta, \bm{\Sigma}_0, \bm{\Sigma}_1, \bm{\Sigma})$ such that the confidence region
\[
\{\bm{\tau}\in\mathbb{R}^d:(\bm{\htau}_{\rm ST}-\bm{\tau})^\top (\bm{\Sigma}_0^{-1} + \bm{\Sigma}_1^{-1})(\bm{\htau}_{\rm ST}-\bm{\tau})\le \hat{M}_{\rm ST}\}
\]
achieves correct coverage for all $\bm{\Delta}$ satisfying $\lvert [\bm{\Sigma}^{-1/2} \bm{\Delta}]_j \rvert \le b_j$ for all $j=1,2,\ldots,d$. We can formulate $\hat{M}_{\rm ST}$ as the optimization problem
\begin{align*}
  \hat{M}_{\rm ST} = \inf \left\{ M \ge 0: \inf_{\bm{\Delta}: \lvert \bm{\Sigma}^{-1/2} \bm{\Delta}\rvert \le \bm{b}} \P_{\bm{\Delta}}((\bm{\htau}_{\rm ST} - \bm{\tau})^\top (\bm{\Sigma}_0^{-1} + \bm{\Sigma}_1^{-1}) (\bm{\htau}_{\rm ST} - \bm{\tau}) \le M) \ge 1-\zeta \right\}.
\end{align*}

We show that $\hat{M}_{\rm ST}$ can be computed efficiently in the following theorem.

\begin{theorem}\label{prop:CI_multi_ST}
  $\hat{M}_{\rm ST}$ is the solution to the equation of $M$:
  \begin{align*}
    \inf_{\bm{\Delta}: \bm{\Sigma}^{-1/2} \bm{\Delta} = \bm{b} \odot \bm{s}, \bm{s} \in \{-1,1\}^d} \P_{\bm{\Delta}}((\bm{\htau}_{\rm ST} - \bm{\tau})^\top (\bm{\Sigma}_0^{-1} + \bm{\Sigma}_1^{-1}) (\bm{\htau}_{\rm ST} - \bm{\tau}) \le M) = 1-\zeta,
  \end{align*} 
  with an explicit form given by:
  \begin{align}\label{eq:P_multi_ST}
    &\P_{\bm{\Sigma}^{-1/2}\bm{\Delta} = \bm{t}}\left((\bm{\htau}_{\rm ST}-\bm{\tau})^\top (\bm{\Sigma}_0^{-1} + \bm{\Sigma}_1^{-1})(\bm{\htau}_{\rm ST}-\bm{\tau})\le M\right) \nonumber\\
    =& \Psi_d\left(M;\left\lVert (\bm{\Sigma}_0^{-1} + \bm{\Sigma}_1^{-1})^{-1/2} \bm{\Sigma}_1^{-1} \bm{\Sigma}^{1/2} \bm{t}\right\rVert_2^2\right) \Psi_d\left(q; \left\lVert (\bm{\Sigma}_0 + \bm{\Sigma}_1)^{-1/2} \bm{\Sigma}^{1/2} \bm{t}\right\rVert_2^2\right) \nonumber\\
    &+ \int_{\lVert \bm{u} \rVert_2^2 > q} \Psi_d\left(M;\left\lVert (\bm{\Sigma}_0^{-1} + \bm{\Sigma}_1^{-1})^{-1/2} \bm{\Sigma}_1^{-1} [\bm{\Sigma}^{1/2} \bm{t} - (1-h_q^*(\lVert \bm{u} \rVert_2^2))(\bm{\Sigma}_0+ \bm{\Sigma}_1)^{1/2} \bm{u}]\right\rVert_2^2\right) \nonumber\\
    &\phi_{(\bm{\Sigma}_0 + \bm{\Sigma}_1)^{-1/2} \bm{\Sigma}^{1/2} \bm{t}, \bm{I}_d}\left(\bm{u}\right) \d \bm{u}.
  \end{align}
\end{theorem}

Theorem~\ref{prop:CI_multi_ST} extends the univariate result in Theorem~\ref{thm:L_ST} to the multivariate setting. Theorem~\ref{prop:CI_multi_ST} provides an explicit expression for the coverage probability $\P_{\bm{\Sigma}^{-1/2}\bm{\Delta} = \bm{t}}((\bm{\htau}_{\rm ST}-\bm{\tau})^\top (\bm{\Sigma}_0^{-1} + \bm{\Sigma}_1^{-1})(\bm{\htau}_{\rm ST}-\bm{\tau})\le M)$ as a function of $M$ and $\bm{t}$. The first term in \eqref{eq:P_multi_ST} corresponds to the event that the pretest fails to reject the null hypothesis ($\bm{\Delta} = \bm{0}$), in which case $\bm{\htau}_{\rm ST}$ reduces to the precision-weighted estimator. The second term in \eqref{eq:P_multi_ST} integrates over the regions where the pretest rejects the null hypothesis, in which case $\bm{\htau}_{\rm ST}$ reduces to the unbiased estimator with a nonconstant shift to ensure continuity at the pretest boundary. As in the univariate case, the monotonicity of the coverage probability makes the computation of $\hat{M}_{\rm ST}$ more efficient than that of $\hat{M}_{\rm PT}$. For any $M>0$, the coverage probability is minimized at the boundary of the bias region, namely at the vertices $\bm{\Delta}$ satisfying $\bm{\Sigma}^{-1/2} \bm{\Delta} = \bm{b} \odot \bm{s}$ for some $\bm{s} \in \{-1,1\}^d$.

Using Theorems~\ref{prop:CI_multi}--\ref{prop:CI_multi_ST}, we can construct confidence regions based on $\bm{\hat\tau}_{\rm PW}$, $\bm{\hat\tau}_{\rm PT}$, and $\bm{\hat\tau}_{\rm ST}$ for various bias bounds $\bm{b}$ and significance levels $\zeta$.

\section{Generalization to multiple estimators}\label{sec:data_fusion}

In this section, we extend our framework to the multiple estimators setting. In empirical research, analysts often face the challenge of integrating evidence from datasets of heterogeneous quality. A leading example in causal inference is to combine an RCT with multiple observational studies. While the RCT provides unbiased but often noisy estimates, observational studies can offer much larger sample sizes but are subject to hidden bias due to unmeasured confounding. Beyond causal inference, similar problems occur when synthesizing results from multiple registries, surveys, or administrative databases, where no single source is sufficient on its own. We present our theory below.

\subsection{Setup}

We consider the following multiple estimators setting with one unbiased estimator and $K$ potentially biased estimators. We can generalize to the case where there are multiple unbiased estimators estimating the same parameter. In that case, we can first combine all the unbiased estimators into a single unbiased estimator with smaller variance using the precision-weighted estimator. Then we can apply our current framework to combine this aggregated unbiased estimator with multiple biased estimators.
\begin{assumption}\label{ass:normal_multi_data}
  Suppose we observe $K+1$ independent random variables: one unbiased estimator $\htau_0 \sim N(\tau,\sigma_0^2)$ and $K$ potentially biased estimators $\htau_j \sim N(\tau + \Delta_j,\sigma_j^2)$ for $j = 1,\ldots,K$. Here $\tau \in \bR$ is the unknown parameter of interest. We assume that $\sigma_0^2, \sigma_1^2, \ldots, \sigma_K^2$ are known, whereas $\Delta_1, \ldots, \Delta_K$ are unknown. 
\end{assumption}
Let $\gamma_j = \sigma_0^2/\sigma_j^2$ be the variance ratio between the unbiased estimator and the $j$-th biased estimator. Let $\bm{\gamma} = (\gamma_1,\ldots,\gamma_K)^\top$ be the vector of variance ratios.

We consider a generic combined estimator $\hat\tau = \hat\tau(\htau_0, \htau_1, \ldots, \htau_K, \sigma_0^2, \sigma_1^2, \ldots, \sigma_K^2)$. We assume that $\lvert \Delta_j/\sigma_0 \rvert \le b_j$ for some $b_j > 0$ for all $j = 1,\ldots,K$ and  study how the confidence interval changes with the maximum relative bias vector $\bm{b} = (b_1,\ldots,b_K)$. Let $\bm{\Delta} = (\Delta_1,\ldots,\Delta_K)$ be the vector of unknown biases. Analogous to Definition~\ref{def:CI} for the univariate case, the confidence interval is defined below.

\begin{definition}\label{def:CI_multi_data}
  Given a significance level $\zeta \in (0,1)$ and the maximum relative bias vector $\bm{b}$ with $b_j \ge 0$ for all $j=1,\ldots,K$, we want to construct an interval $\cI(\bm{b}, \zeta) = \cI(\bm{b}, \zeta, \htau_0, \htau_1, \ldots, \htau_K, \sigma_0^2, \sigma_1^2, \ldots, \sigma_K^2)$ such that
  \begin{align}\label{eq:coverage_multi_data}
    \inf_{\bm{\Delta}: \lvert \bm{\Delta}/\sigma_0 \rvert \le \bm{b}} \P_{\bm{\Delta}}(\tau \in \hat\tau - \cI(\bm{b}, \zeta)) 
    = \inf_{\bm{\Delta}: \lvert \bm{\Delta}/\sigma_0 \rvert \le \bm{b}} \P_{\bm{\Delta}}(\hat\tau - \tau \in \cI(\bm{b}, \zeta)) \ge 1 - \zeta,
  \end{align}
The confidence interval for $\tau$ based on $\hat\tau$ is then given by $\hat\tau - \cI(\bm{b}, \zeta)$.
\end{definition}

Then we introduce the monotonicity conditions on the interval $\cI(\bm{b}, \zeta)$ below, which generalizes Assumption~\ref{asp:monotonicity} to the multiple estimators setting.

\begin{assumption}\label{asp:monotonicity_multi_data}
  We assume:
  \begin{enumerate}
    \item For fixed $(\htau_0, \htau_1, \ldots, \htau_K, \sigma_0^2, \sigma_1^2, \ldots, \sigma_K^2)$ and $\zeta$, we have $\cI(\bm{b}, \zeta) \subset \cI(\bm{b'}, \zeta)$ whenever $\bm{b} \le \bm{b'}$, i.e., $b_j \le b'_j$ for all $j=1,\ldots,K$.
    \item For fixed $(\htau_0, \htau_1, \ldots, \htau_K, \sigma_0^2, \sigma_1^2, \ldots, \sigma_K^2)$ and $\bm{b}$, we have $\cI(\bm{b}, \zeta) \subset \cI(\bm{b}, \zeta')$ whenever $\zeta \ge \zeta'$.
  \end{enumerate}
\end{assumption}

Assume $\cI(\bm{b},\zeta)$ satisfies the monotonicity conditions of Assumption~\ref{asp:monotonicity_multi_data}. A common family of such intervals is fixed-length centered intervals: $\cI(\bm{b}, \zeta) = [-c(\bm{b}, \zeta, \sigma_0^2, \sigma_1^2, \ldots, \sigma_K^2), c(\bm{b}, \zeta, \sigma_0^2, \sigma_1^2, \ldots, \sigma_K^2)]$ with some constant $c(\bm{b}, \zeta, \sigma_0^2, \sigma_1^2, \ldots, \sigma_K^2) \ge 0$. In this case, the confidence interval $\hat\tau - \cI(\bm{b}, \zeta)$ in Definition~\ref{def:CI_multi_data} is given by $[\hat\tau - c(\bm{b}, \zeta, \sigma_0^2, \sigma_1^2, \ldots, \sigma_K^2), \hat\tau + c(\bm{b}, \zeta, \sigma_0^2, \sigma_1^2, \ldots, \sigma_K^2)]$.

Then we define the b-value with multiple estimators below.
\begin{definition}\label{def:b_multi_data}
  Define the b-value as the critical boundary $\bm{b}^*$ of testing $\tau = 0$ versus $\tau \neq 0$ as
  \begin{align}\label{eq:b_multi_data}
    \bm{b}^*(\zeta) = \bm{b}^*(\zeta, \htau_0, \htau_1, \ldots, \htau_K, \sigma_0^2, \sigma_1^2, \ldots, \sigma_K^2) = \partial\left\{\bm{b} \ge \bm{0}: 0 \in \hat\tau - \bm{\cI}(\bm{b}, \zeta)\right\}.
  \end{align}
\end{definition}

Definition~\ref{def:b_multi_data} extends the univariate notion of the b-value (Definition~\ref{def:b}) to the multiple estimators setting. As in the multivariate b-value defined in Definition~\ref{def:b_multi}, in higher dimensions $K > 1$, the set $\left\{\bm{b} \ge \bm{0}: \bm{0} \in \bm{\hat\tau} - \bm{\cI}(\bm{b}, \zeta)\right\}$ is a convex region in the first quadrant, and the b-value is the boundary of this region, which is a $(K-1)$-dimensional surface. See Definition~\ref{def:b_multi} and the subsequent discussion for more details.

\subsection{Point estimation}

First, we recall the precision-weighted estimator defined as
\[
  \htau_{\rm PW} = \frac{\sigma_0^{-2}}{\sigma_0^{-2} + \sum_{\ell=1}^K \sigma_\ell^{-2}} \htau_0 + \sum_{j=1}^K \frac{\sigma_j^{-2}}{\sigma_0^{-2} + \sum_{\ell=1}^K \sigma_\ell^{-2}} \htau_j = \htau_0 + \sum_{j=1}^K \frac{\gamma_j}{1+\lVert \bm{\gamma} \rVert_1} (\htau_j - \htau_0).
\]
If $\Delta_j=0$ for all $j=1,\ldots,K$, the precision-weighted estimator is the maximum likelihood estimator and best linear unbiased estimator of $\tau$. Moreover, by the classical Wald–Le Cam asymptotic decision theory \citep{wald1949statistical, le1956asymptotic}, it is asymptotically admissible and minimax under the $L_2$ risk, i.e., $\E[(\hat\tau - \tau)^2]$, among all regular estimators under standard regularity conditions.

Second, we introduce a pretest estimator. We incorporate the pretest for $\Delta_j = 0$ separately using $\htau_j$ and the unbiased estimator $\htau_0$, and combine those $\htau_j$'s that fail to reject the null. Here we assume the pretests share the same significance level $\alpha$ for simplicity. The generalization to different significance levels is straightforward. Since $\htau_j - \htau_0 \sim N(\Delta_j, \sigma_0^2 + \sigma_j^2)$ and $\sigma_0^2 + \sigma_j^2 = (1+\gamma_j^{-1}) \sigma_0^2$, let $$\mathcal{A}_j = \{\lvert \htau_j - \htau_0 \rvert \le (1+\gamma_j^{-1})^{1/2}\sigma_0 c_{\alpha/2}\}$$ denote the event that the pretest fails to reject the null hypothesis ($\Delta_j=0$). The pretest estimator is:
\begin{align*}
  \htau_{\rm PT} = \htau_0 + \sum_{j=1}^K \frac{\gamma_j}{1+\lVert \bm{\gamma} \rVert_1} (\htau_j - \htau_0) \ind\left(\mathcal{A}_j\right).
\end{align*}

Third, we introduce a soft-thresholding estimator, which ensures continuity at the pretest boundary. Instead of dropping a biased estimator entirely when the pretest rejects, soft-thresholding shifts the unbiased estimator by making the estimator continuous at the pretest boundary. The soft-thresholding estimator is:
\begin{align*}
  \htau_{\rm ST}= \htau_0 + \sum_{j=1}^K \frac{\gamma_j}{1+\lVert \bm{\gamma} \rVert_1} \left[(\htau_j - \htau_0) \ind\left(\mathcal{A}_j\right) + (1+\gamma_j^{-1})^{1/2}\sigma_0 c_{\alpha/2} {\rm sign}(\htau_j - \htau_0) \ind\left(\mathcal{A}_j^\textup{c}\right) \right].
\end{align*}

By the definitions of the pretest estimator and the soft-thresholding estimator, when the pretests fail to reject for all $j=1,\ldots,K$, both the pretest estimator and the soft-thresholding estimator reduce to the precision-weighted estimator. When $K=1$, i.e., there is only one biased estimator, both pretest estimator and soft-thresholding estimator reduce to the estimators introduced in Section~\ref{sec:point}. Similar to the discussion in Section~\ref{sec:point}, the choice of $\alpha$ is a tuning parameter.

The pretest estimator and the soft-thresholding estimator considered above are not the only way to combine $\hat\tau_0$ and $\hat\tau_1,\ldots,\hat\tau_K$. For example, one could apply shrinkage or soft-thresholding directly between the precision-weighted estimator $\hat\tau_{\rm PW}$ and the unbiased estimator $\hat\tau_0$, rather than at the level of individual biased components. We focus on thresholding at the component level because it admits a transparent interpretation in terms of testing and controlling each bias component $\Delta_j$ separately, which aligns naturally with our sensitivity analysis perspective. Our framework can accommodate other shrinkage schemes in principle, but we do not pursue them here.

\subsection{Confidence intervals}

First, we construct the confidence intervals based on the precision-weighted estimator $\htau_{\rm PW}$. Under Assumption~\ref{ass:normal_multi_data}, we have $\htau_{\rm PW} \sim N(\tau + (1+\lVert \bm{\gamma} \rVert_1)^{-1} \langle \bm{\gamma}, \bm{\Delta} \rangle, (1+\lVert \bm{\gamma} \rVert_1)^{-1} \sigma_0^2)$. The following theorem provides the confidence intervals based on $\htau_{\rm PW}$ as a function of bias bound $\bm{b}$.

\begin{theorem}\label{prop:L_PW_multi}
  Let $\hat{L}_{\rm PW} = \hat{L}_{\rm PW}(\bm{b}, \zeta, \bm{\gamma}) \ge 0$ denote the solution to the equation of $L$:
  \[
    \Phi\left(L - \frac{\langle \bm{\gamma}, \bm{b} \rangle}{\sqrt{1+\lVert \bm{\gamma} \rVert_1}}  \right) - \Phi\left(-L - \frac{\langle \bm{\gamma}, \bm{b} \rangle}{\sqrt{1+\lVert \bm{\gamma} \rVert_1}} \right) = 1-\zeta.
  \]
  The $\hat{L}_{\rm PW}$ always exists and is unique. The shortest length symmetric centered confidence interval based on $\htau_{\rm PW}$ for $\tau$ satisfying \eqref{eq:coverage_multi_data} is given by $[\htau_{\rm PW} - \hat{L}_{\rm PW} (1+\lVert \bm{\gamma} \rVert_1)^{-1/2} \sigma_0, \htau_{\rm PW} + \hat{L}_{\rm PW} (1+\lVert \bm{\gamma} \rVert_1)^{-1/2} \sigma_0]$.
\end{theorem}

Theorem~\ref{prop:L_PW_multi} extends the single estimator result in Theorem~\ref{thm:L_PW} to the multiple estimators setting. In the single estimator case, the bias constraint $\{\Delta: \lvert \Delta/\sigma_0 \rvert \le b\}$ has only two boundary points, and the absolute bias of $\htau_{\rm PW}$ is the same at both endpoints. By contrast, the multiple estimators bias constraint $\{\bm{\Delta}: \lvert \bm{\Delta}/\sigma_0 \rvert \le \bm{b}\}$ is a hyperrectangle with $2^K$ vertices, and the maximum absolute bias of $\htau_{\rm PW}$ occurs only at $\bm{\Delta}/\sigma_0 = \pm \bm{b}$. Theorem~\ref{prop:L_PW_multi} therefore characterizes the maximum possible bias of $\htau_{\rm PW}$ over the entire bias constraint.

Second, we construct the confidence intervals based on the pretest estimator $\htau_{\rm PT}$. We seek the shortest length $\hat{L}_{\rm PT} = \hat{L}_{\rm PT}(\bm{b}, \zeta, \bm{\gamma}, \alpha)$ such that the confidence interval $[\htau_{\rm PT}- \hat{L}_{\rm PT} (1+\lVert \bm{\gamma} \rVert_1)^{-1/2} \sigma_0, \htau_{\rm PT}+ \hat{L}_{\rm PT} (1+\lVert \bm{\gamma} \rVert_1)^{-1/2} \sigma_0]$ achieves correct coverage uniformly over all $\bm{\Delta}$ with $\lvert \Delta_j/\sigma_0 \rvert \le b_j$ for all $j=1,\ldots,K$. We can formulate $\hat{L}_{\rm PT}$ as the optimization problem:
\begin{align}\label{eq:L_PT_multi}
  \hat{L}_{\rm PT} = \inf \left\{ L \ge 0: \inf_{\bm{\Delta}: \lvert \bm{\Delta}/\sigma_0 \rvert \le \bm{b}} \P_{\bm{\Delta}}(\lvert \htau_{\rm PT}- \tau \rvert \le L (1+\lVert \bm{\gamma} \rVert_1)^{-1/2} \sigma_0) \ge 1-\zeta \right\}. 
\end{align}

We show how to compute $\hat{L}_{\rm PT}$ in the following theorem.

\begin{theorem}\label{thm:L_PT_multi}
  $\hat{L}_{\rm PT} = \hat{L}_{\rm PT}(\bm{b}, \zeta, \bm{\gamma}, \alpha)$ in \eqref{eq:L_PT_multi} is the solution to the equation of $L$:
  \begin{align*}
    \inf_{\bm{\Delta}: \lvert \bm{\Delta}/\sigma_0 \rvert \le \bm{b}} \P_{\bm{\Delta}}(\lvert \htau_{\rm PT}- \tau \rvert \le L (1+\lVert \bm{\gamma} \rVert_1)^{-1/2} \sigma_0) = 1-\zeta,
  \end{align*} 
  with an explicit form given by:
  \begin{align*}
    &\P_{\bm{\Delta}/\sigma_0 = \bm{t}}(\lvert \htau_{\rm PT}- \tau \rvert \le L (1+\lVert \bm{\gamma} \rVert_1)^{-1/2} \sigma_0) \\
    =& \int_{\bR^K} \left[\Phi\left(L - \frac{\langle \bm{\gamma}, \bm{t} - \bm{u}' \rangle}{\sqrt{1+\lVert \bm{\gamma} \rVert_1}}  \right) - \Phi\left(-L - \frac{\langle \bm{\gamma}, \bm{t} - \bm{u}' \rangle}{\sqrt{1+\lVert \bm{\gamma} \rVert_1}} \right)\right] \phi_{\bm{t}, \bm{V}}(\bm{u}) \d \bm{u},
  \end{align*}
  where for any $\bm{u} \in \bR^K$, $\bm{u}' \in \bR^K$ is defined as
  \[
    u_j' = u_j \ind(\lvert u_j \rvert > (1+\gamma_j^{-1})^{1/2} c_{\alpha/2}) \text{ for all } j=1,\ldots,K,
  \]
  and $\bm{V}_{ij} = 1 + \gamma_i^{-1} \ind (i=j)$ for $i,j=1,\ldots,K$. $[\htau_{\rm PT}- \hat{L}_{\rm PT} (1+\lVert \bm{\gamma} \rVert_1)^{-1/2} \sigma_0, \htau_{\rm PT}+ \hat{L}_{\rm PT} (1+\lVert \bm{\gamma} \rVert_1)^{-1/2} \sigma_0]$ is the shortest length symmetric centered confidence interval based on $\htau_{\rm PT}$ for $\tau$ satisfying \eqref{eq:coverage_multi_data}.
\end{theorem}

Theorem~\ref{thm:L_PT_multi} extends the single estimator result in Theorem~\ref{thm:L_PT} to the multiple estimators setting. Theorem~\ref{thm:L_PT_multi} provides an explicit expression for the coverage probability $\P_{\bm{\Delta}/\sigma_0 = \bm{t}}(\lvert \htau_{\rm PT}- \tau \rvert \le L (1+\lVert \bm{\gamma} \rVert_1)^{-1/2} \sigma_0)$ as a function of $\bm{t}$ and $L$.

Third, we construct the confidence intervals based on the soft-thresholding estimator $\htau_{\rm ST}$. We seek the shortest length $\hat{L}_{\rm ST} = \hat{L}_{\rm ST}(\bm{b}, \zeta, \bm{\gamma}, \alpha)$ such that the confidence interval $[\htau_{\rm ST}- \hat{L}_{\rm ST} (1+\lVert \bm{\gamma} \rVert_1)^{-1/2} \sigma_0, \htau_{\rm ST}+ \hat{L}_{\rm ST} (1+\lVert \bm{\gamma} \rVert_1)^{-1/2} \sigma_0]$ achieves correct coverage uniformly over all $\bm{\Delta}$ with $\lvert \Delta_j/\sigma_0 \rvert \le b_j$ for all $j=1,\ldots,K$. We can formulate $\hat{L}_{\rm ST}$ as the optimization problem:
\begin{align}\label{eq:L_multi}
  \hat{L}_{\rm ST} = \inf \left\{ L \ge 0: \inf_{\bm{\Delta}: \lvert \bm{\Delta}/\sigma_0 \rvert \le \bm{b}} \P_{\bm{\Delta}}(\lvert \htau_{\rm ST}- \tau \rvert \le L (1+\lVert \bm{\gamma} \rVert_1)^{-1/2} \sigma_0) \ge 1-\zeta \right\}. 
\end{align}

We show that $\hat{L}_{\rm ST}$ can be computed efficiently in the following theorem.

\begin{theorem}\label{prop:L_ST_multi}
  For any $L > 0$, $\P_{\bm{\Delta}}(\lvert \htau_{\rm ST}- \tau \rvert \le L (1+\lVert \bm{\gamma} \rVert_1)^{-1/2} \sigma_0)$ is symmetric about $\bm{\Delta} = \bm{0}$ and monotonically decreasing in $\lvert \Delta_j \rvert$ for all $j=1,\ldots,K$. Then $\hat{L}_{\rm ST} = \hat{L}_{\rm ST}(\bm{b}, \zeta, \bm{\gamma}, \alpha)$ in \eqref{eq:L_multi} is the solution to the equation of $L$:
  \begin{align*}
    \P_{\bm{\Delta}/\sigma_0 = \bm{b}}(\lvert \htau_{\rm ST}- \tau \rvert \le L (1+\lVert \bm{\gamma} \rVert_1)^{-1/2} \sigma_0) = 1-\zeta,
  \end{align*}
  with an explicit form given by:
  \begin{align*}
    &\P_{\bm{\Delta}/\sigma_0 = \bm{t}}(\lvert \htau_{\rm ST}- \tau \rvert \le L (1+\lVert \bm{\gamma} \rVert_1)^{-1/2} \sigma_0) \\
    =& \int_{\bR^K} \left[\Phi\left(L - \frac{\langle \bm{\gamma}, \bm{t} - \bm{u}' \rangle}{\sqrt{1+\lVert \bm{\gamma} \rVert_1}}  \right) - \Phi\left(-L - \frac{\langle \bm{\gamma}, \bm{t} - \bm{u}' \rangle}{\sqrt{1+\lVert \bm{\gamma} \rVert_1}} \right)\right] \phi_{\bm{t}, \bm{V}}(\bm{u}) \d \bm{u},
  \end{align*}
  where for any $\bm{u} \in \bR^K$, $\bm{u}' \in \bR^K$ is defined as
  \[
    u_j' = [u_j  - (1+\gamma_j^{-1})^{1/2} c_{\alpha/2} {\rm sign}(u_j)] \ind(\lvert u_j \rvert > (1+\gamma_j^{-1})^{1/2} c_{\alpha/2}) \text{ for all } j=1,\ldots,K,
  \]
  and $\bm{V}_{ij} = 1 + \gamma_i^{-1} \ind (i=j)$ for $i,j=1,\ldots,K$. $[\htau_{\rm ST}- \hat{L}_{\rm ST} (1+\lVert \bm{\gamma} \rVert_1)^{-1/2} \sigma_0, \htau_{\rm ST}+ \hat{L}_{\rm ST} (1+\lVert \bm{\gamma} \rVert_1)^{-1/2} \sigma_0]$ is the shortest length symmetric centered confidence interval based on $\htau_{\rm ST}$ for $\tau$ satisfying \eqref{eq:coverage_multi_data}.
\end{theorem}

Theorem~\ref{prop:L_ST_multi} extends the single estimator result in Theorem~\ref{thm:L_ST} to the multiple estimators setting. Theorem~\ref{prop:L_ST_multi} provides an explicit expression for the coverage probability $\P_{\bm{\Delta}/\sigma_0 = \bm{t}}(\lvert \htau_{\rm ST}- \tau \rvert \le L (1+\lVert \bm{\gamma} \rVert_1)^{-1/2} \sigma_0)$ as a function of $\bm{t}$ and $L$. As in the single estimator case, the monotonicity makes the computation of $\hat{L}_{\rm ST}$ more efficient than that of $\hat{L}_{\rm PT}$. For any $L>0$, the coverage probability is minimized at the vertices $\bm{\Delta}$ satisfying $\bm{\Delta}/\sigma_0 = \pm \bm{b}$.

\section{Empirical studies}\label{sec:empirical}

In this section, we use the example from \cite{angrist1991does} to demonstrate how our framework works in practice, where the authors studied the effect of years of schooling on earnings. \cite{angrist1991does} used quarter of birth as an instrument to obtain the IV estimate. They also reported the OLS estimate, which would be biased in the presence of endogeneity.

% In this section, we use two empirical examples to demonstrate how our framework works in practice. The first example is from \cite{angrist1991does}, where the authors studied the effect of years of schooling on earnings. \cite{angrist1991does} used quarter of birth as an instrument to obtain the IV estimate. They also reported the OLS estimate, which would be biased in the presence of endogeneity. The second example is from \cite{lalonde1986evaluating}, where the authors compared experimental and observational estimates of the causal effect of job training using the National Supported Work demonstration. We consider two estimators that both use the experimental treatment group but differ in how the control group is constructed. The first estimator (RCT) uses the experimental control group and therefore yields an unbiased estimate of the treatment effect. The second estimator (OC) uses observational controls from \cite{lalonde1986evaluating}’s ``CPS-1'' sample, as reconstructed by \cite{dehejia1999causal}, which is subject to selection bias. Thus, the RCT estimator is unbiased, while the OC estimator is biased but more precise since the CPS-1 sample is larger than the experimental sample. We analyze the \cite{angrist1991does} example in the remainder of this section, and relegate the analysis of the \cite{lalonde1986evaluating} example to Appendix~\ref{appendix:empirical}.

We consider the OLS estimator and the IV estimator. The OLS estimator of the return to schooling is potentially biased because schooling decisions are endogenous: unobserved factors may influence both schooling and earnings. By contrast, the IV estimator uses quarter of birth as a source of exogenous variation, yielding a consistent estimate of the causal effect under standard IV assumptions. We take the IV estimator as the unbiased estimator $\htau_0$ and the OLS estimator as the biased estimator $\htau_1$. Using these two estimators, we construct the three combined estimators: the precision-weighted estimator $\htau_{\rm PW}$, the pretest estimator $\htau_{\rm PT}$, and the soft-thresholding estimator $\htau_{\rm ST}$.

Figure~\ref{fig:ols_iv} visualizes the confidence intervals of the IV, OLS, and combined estimators against the maximum relative bias $\lvert \Delta/\sigma_0 \rvert \le b$.

\begin{figure}[htbp]
  \centering
  \includegraphics[width=\linewidth]{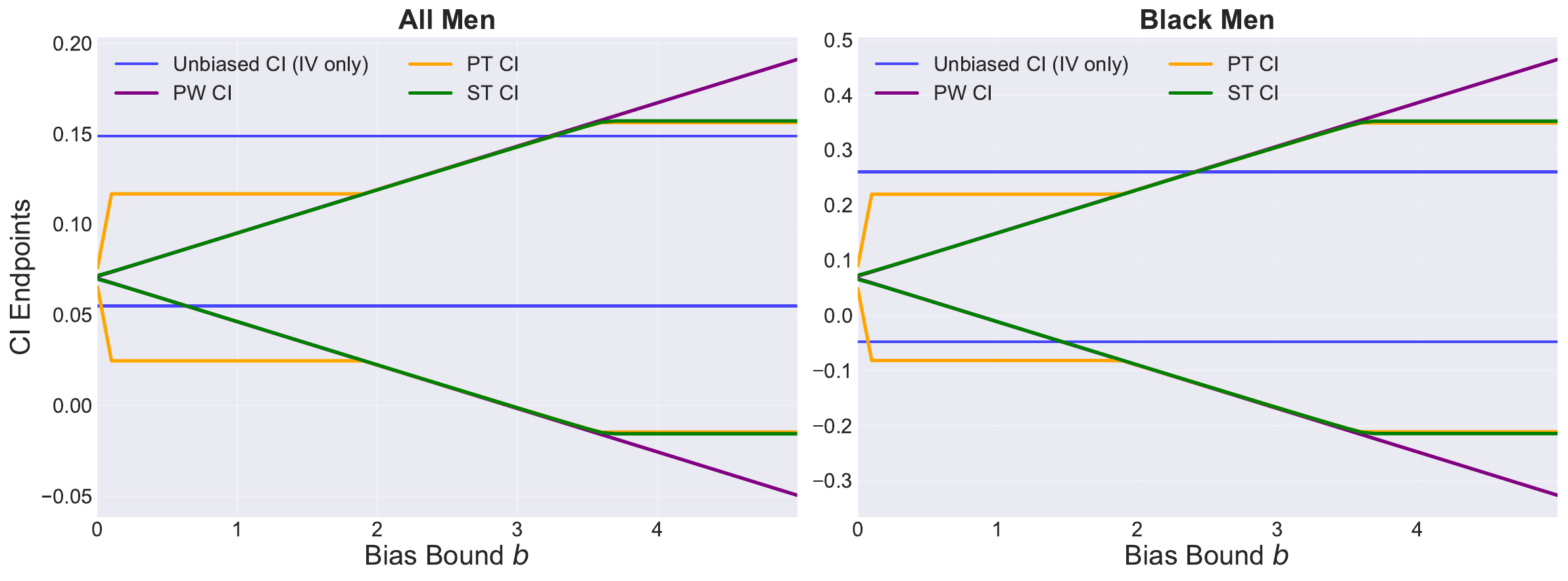}
  \caption{Confidence intervals against the maximum relative bias $\lvert \Delta/\sigma_0 \rvert \le b$}
  \label{fig:ols_iv}
\end{figure}

% Table~\ref{tab:iv_ols} reports the estimates and standard errors of the IV and OLS estimators for men (black men and white men) born between 1930-39, replicating the results in \citet[Panel B, Table III]{angrist1991does}. Following Section II.C in \cite{angrist1991does}, Table~\ref{tab:iv_ols} further reports the estimates and standard errors of the IV and OLS estimators for black men born between 1930-39. For each sample, Table~\ref{tab:iv_ols} also reports, for each combined estimator, the point estimate and the half-length of the 95\% ($\alpha=0.05$) confidence interval under two scenarios: no bias ($b=0$) and infinite bias ($b=\infty$). For comparison, the half-lengths of the confidence intervals based solely on OLS and IV are also provided. For IV, the half-length is always $c_{\alpha/2} \sigma_0$, while for OLS, the half-length is $c_{\alpha/2} \sigma_1$ when $b=0$ and $\infty$ when $b=\infty$.

In both the full sample (black and white men) and the subsample (black men), we observe that the standard error of the OLS estimator is much smaller than that of the IV estimator. Moreover, although we compute the OLS estimator and the IV estimator using the same sample, their correlation is very low. Therefore, the new biased estimator constructed following Section~\ref{sec:dependent} is nearly identical to the original biased estimator (OLS). The combined estimators behave as expected. When the bias bound is small, $\htau_{\rm PW}$, $\htau_{\rm PT}$, and $\htau_{\rm ST}$ all yield confidence intervals substantially shorter than that of the unbiased estimator (IV), reflecting efficiency gains from incorporating the more precise OLS estimator. Notably, the confidence interval of $\htau_{\rm ST}$ is nearly identical to that of $\htau_{\rm PW}$ and much shorter than that of $\htau_{\rm PT}$ when the bias bound is small. Even when the bias bound becomes large, the confidence intervals of $\htau_{\rm ST}$ and $\htau_{\rm PT}$ remain nearly the same in length. These behaviors arise because the OLS estimator is far more precise than the IV estimator.

However, there are differences between the full sample (black men and white men) and the subsample (black men). In the full sample, the sample size is very large, and the IV estimator is sufficiently precise that the null hypothesis of zero returns to schooling can be rejected using IV alone. In contrast, for the subsample of black men, the sample size is much smaller, leading to a much larger standard error for the IV estimator. As a result, the IV-based confidence interval is wide enough that the null hypothesis cannot be rejected at the 5\% significance level. This contrast illustrates a common empirical challenge: when the unbiased estimator is noisy, inference based solely on it can be inconclusive, even when the overall dataset is large. In such settings, the combined estimators, particularly the precision-weighted estimator $\hat\tau_{\rm PW}$ and the soft-thresholding estimator $\hat\tau_{\rm ST}$, yield much tighter confidence intervals under small biases. Under small biases, these estimators provide sufficient evidence to reject the null hypothesis of zero returns to schooling, while still allowing researchers to explicitly assess sensitivity to potential bias.

% \begin{table}[htbp]
%   \centering
%   \begin{threeparttable}
%   \caption{Estimates of the return to schooling (men born between 1930--39)}
%   \label{tab:iv_ols}
%   \begin{tabular}{llcccc}
%   \toprule
%   \textbf{Sample} & \textbf{Estimator} & \textbf{Estimate} & \textbf{SE} & \textbf{CI Half-Length ($b=0$)} & \textbf{CI Half-Length ($b=\infty$)} \\
%   \midrule
%   \multirow{5}{*}{All men} 
%     & IV ($\htau_0$)  & 0.1019 & 0.0239 & 0.0469 & 0.0469 \\
%     & OLS ($\htau_1$) & 0.0709 & 0.0003 & 0.0007 & $\infty$ \\
%     & PW ($\htau_{\rm PW}$) & 0.0709 & -- & 0.0007 & $\infty$ \\
%     & PT ($\htau_{\rm PT}$) & 0.0709 & -- & 0.0053 & 0.0855 \\
%     & ST ($\htau_{\rm ST}$) & 0.0709 & -- & 0.0010 & 0.0863 \\
%   \midrule
%   \multirow{5}{*}{Black men} 
%     & IV ($\htau_0$)  & 0.1066 & 0.0787 & 0.1542 & 0.1542 \\
%     & OLS ($\htau_1$) & 0.0694 & 0.0013 & 0.0025 & $\infty$ \\
%     & PW ($\htau_{\rm PW}$) & 0.0694 & -- & 0.0025 & $\infty$ \\
%     & PT ($\htau_{\rm PT}$) & 0.0694 & -- & 0.0204 & 0.2804 \\
%     & ST ($\htau_{\rm ST}$) & 0.0694 & -- & 0.0036 & 0.2836 \\
%   \bottomrule
%   \end{tabular}
%   \footnotesize
%   \begin{tablenotes}
%     \item Notes: The full sample (black and white men) size is $n=329{,}509$ and the subsample (black men) size is $n=26{,}913$. 
%     The correlation between OLS and IV is $\rho=0.0141$ for the full sample and $\rho=0.0165$ for the subsample.
%   \end{tablenotes}
% \end{threeparttable}
% \end{table}

\section{Discussion}\label{sec:discussion}

This paper develops a strategy to combine unbiased and biased estimators from
a sensitivity analysis perspective. In particular, we construct a sequence of confidence intervals indexed by the
magnitude of bias and propose the notion of the b-value to quantify the maximum relative bias so that combining estimators yields an insignificant result. 

\section*{Acknowledgment}

We thank Avi Feller and Liyang Sun for helpful comments. Lin was partially supported by the Two Sigma PhD Fellowship. Ding was supported by the U.S. National Science Foundation (1945136, 2514234).

{%\small
\bibliographystyle{apalike}
\bibliography{AMS}
}

\newpage

\appendix

\clearpage
\setcounter{page}{1}
\pagenumbering{arabic}
\renewcommand{\thepage}{S\arabic{page}}

\section*{Supplementary materials for “Introducing the b-value: combining unbiased and biased estimators from a sensitivity analysis perspective”}

Appendix~\ref{sec:discussions} contains several additional discussions that complement the main text. Section~\ref{sec:bvalue} discusses how to compute the b-value efficiently. Section~\ref{sec:one_sided} extends our framework to one-sided confidence bounds for $\tau$. Section~\ref{sec:bias_dependent} discusses the advantages of using prespecified point estimators instead of bias-dependent point estimators. Section~\ref{sec:dependence} provides further details on the generalization to dependent case, complementing the discussion in Section~\ref{sec:dependent}. Section~\ref{sec:efron--morris} discusses further details on the relationship between the generic soft-thresholding estimator and the estimator in in \cite{j1981estimation, bickel1984parametric}, complementing the discussion in Section~\ref{sec:point_multi}.

% Appendix~\ref{sec:discussions} contains several additional discussions that complement the main text. Section~\ref{sec:bvalue} discusses how to compute the b-value efficiently. Section~\ref{sec:pvalue} discusses the definition and computation of the p-value for testing $\tau = 0$ versus $\tau \neq 0$ using the generic combined estimator $\hat\tau$ and its corresponding interval $\hat\tau - \cI(b, \zeta)$, and discusses its relationship with the b-value. Section~\ref{sec:one_sided} extends our framework to one-sided confidence bounds for $\tau$. Section~\ref{sec:bias_dependent} discusses the advantages of using prespecified point estimators instead of bias-dependent point estimators. Section~\ref{sec:any_bias_level} discusses how to construct confidence intervals for $\tau$ which hold for any bias level. Section~\ref{sec:dependence} provides further details on the generalization to dependent case, complementing the discussion in Section~\ref{sec:dependent}. Section~\ref{sec:efron--morris} discusses further details on the relationship between the generic soft-thresholding estimator and the Efron--Morris estimator, complementing the discussion in Section~\ref{sec:point_multi}. Section~\ref{appendix:empirical} provides further details on the empirical studies, complementing the discussion in Section~\ref{sec:empirical}.

Appendix~\ref{sec:proof_main} contains the proofs of the results in the main paper, and Appendix~\ref{sec:proof_additional} contains the proofs of the results in the appendix.

\section{More discussions}\label{sec:discussions}

\subsection{Computing the b-value efficiently}\label{sec:bvalue}

In this section, we discuss how to compute the b-value efficiently.

Consider a general combined estimator $\hat\tau$ with symmetric fixed-length centered confidence interval $[\hat\tau - c(b,\zeta, \sigma_0^2, \sigma_1^2), \hat\tau + c(b,\zeta, \sigma_0^2, \sigma_1^2)]$. By Definition~\ref{def:b}, the b-value $b^*(\zeta)  = b^*(\zeta, \hat\tau, \sigma_0^2, \sigma_1^2) \ge 0$ is the smallest bias level at which the null value $0$ is contained in the confidence interval. Equivalently, it is the solution to the equation of $b$:
\begin{align*}
  b^*(\zeta) =& \inf\left\{b \ge 0: 0 \in [\hat\tau - c(b,\zeta, \sigma_0^2, \sigma_1^2), \hat\tau + c(b,\zeta, \sigma_0^2, \sigma_1^2)]\right\} \\
  =& \inf\left\{b \ge 0: c(b,\zeta, \sigma_0^2, \sigma_1^2) \ge \lvert \hat\tau \rvert\right\}.
\end{align*}
Therefore, the b-value $b^*(\zeta)$ is $0$ if $c(0,\zeta, \sigma_0^2, \sigma_1^2) > \lvert \hat\tau \rvert$, or is $\infty$ if $c(\infty,\zeta, \sigma_0^2, \sigma_1^2) < \lvert \hat\tau \rvert$. Otherwise, if we further assume $c(b,\zeta, \sigma_0^2, \sigma_1^2)$ is strictly increasing in $b$ for given $\zeta, \sigma_0^2, \sigma_1^2$, then the b-value $b^*(\zeta)$ is the unique solution to the equation of $b$: $c(b, \zeta, \sigma_0^2, \sigma_1^2) = \lvert \hat\tau \rvert$. In this case, we can compute $b^*(\zeta)$ efficiently using standard one-dimensional root-finding methods such as the bisection algorithm. In many settings, the function $c(b,\zeta,\sigma_0^2,\sigma_1^2)$ does not admit a closed-form expression. Instead, we solve an equation of $L$: $g(L;b,\zeta, \sigma_0^2, \sigma_1^2) = 0$ to obtain $c(b,\zeta,\sigma_0^2,\sigma_1^2)$ for some function $g$ depending on $b, \zeta, \sigma_0^2, \sigma_1^2$. In such cases, we can compute the b-value by solving $g(\lvert \hat\tau \rvert;b,\zeta, \sigma_0^2, \sigma_1^2) = 0$ as an equation of $b$. This reduces the computation of the b-value to a single one-dimensional root-finding problem.

The above discussion provides a general method for computing the b-value for any combined estimator with a symmetric fixed-length centered confidence interval. Now we apply this method to compute the b-value for the precision-weighted estimator $\htau_{\rm PW}$, the pretest estimator $\htau_{\rm PT}$, and the soft-thresholding estimator $\htau_{\rm ST}$.

First, the following theorem provides the b-value $b^*_{\rm PW}(\zeta)$ for the precision-weighted estimator $\htau_{\rm PW}$ with confidence interval $[\htau_{\rm PW} - \hat{L}_{\rm PW} (1+\gamma)^{-1/2} \sigma_0, \htau_{\rm PW} + \hat{L}_{\rm PW} (1+\gamma)^{-1/2} \sigma_0]$.

\begin{theorem}\label{thm:b_PW}
  The b-value $b^*_{\rm PW}(\zeta)$ is $0$ if
  \[
    \Phi\left(\frac{\lvert \htau_{\rm PW} \rvert}{(1+\gamma)^{-1/2} \sigma_0}\right) - \Phi\left(-\frac{\lvert \htau_{\rm PW} \rvert}{(1+\gamma)^{-1/2} \sigma_0}\right) < 1-\zeta.
  \]
  Otherwise, the b-value $b^*_{\rm PW}(\zeta)$ can be equivalently written as the solution to the equation of $b$:
  \[
  \Phi\left(\frac{\lvert \htau_{\rm PW} \rvert}{(1+\gamma)^{-1/2} \sigma_0} - \frac{\gamma}{\sqrt{1+\gamma}} b \right) - \Phi\left(-\frac{\lvert \htau_{\rm PW} \rvert}{(1+\gamma)^{-1/2} \sigma_0}  - \frac{\gamma}{\sqrt{1+\gamma}} b \right) = 1-\zeta.
\]
\end{theorem}

Second, the following theorem provides the b-value $b^*_{\rm PT}(\zeta)$ for the pretest estimator $\htau_{\rm PT}$ with confidence interval $[\htau_{\rm PT} - \hat{L}_{\rm PT} (1+\gamma)^{-1/2} \sigma_0, \htau_{\rm PT} + \hat{L}_{\rm PT} (1+\gamma)^{-1/2} \sigma_0]$.

\begin{theorem}\label{thm:b_PT}
  The b-value $b^*_{\rm PT}(\zeta)$ is $0$ if
  \[
    \P_{\Delta/\sigma_0 = 0}(\lvert \tilde\tau_{\rm PT} - \tau \rvert \le \lvert \htau_{\rm PT} \rvert \given \htau_{\rm PT}) < 1-\zeta,
  \]
  and is $\infty$ if
  \[
    \min_{t \ge 0}\P_{\Delta/\sigma_0 = t}(\lvert \tilde\tau_{\rm PT} - \tau \rvert \le \lvert \htau_{\rm PT} \rvert \given \htau_{\rm PT}) > 1-\zeta.
  \]
  Otherwise, the b-value $b^*_{\rm PT}(\zeta)$ can be equivalently written as the solution to the equation of $b$:
  \[
    \min_{0 \le t \le b} \P_{\Delta/\sigma_0 = t}(\lvert \tilde\tau_{\rm PT} - \tau \rvert \le \lvert \htau_{\rm PT} \rvert \given \htau_{\rm PT}) = 1-\zeta,
  \]
  where $\tilde\tau_{\rm PT}$ is independent and identically distributed as $\htau_{\rm PT}$, with
  \begin{align*}
    &\P_{\Delta/\sigma_0 = t}(\lvert \tilde\tau_{\rm PT} - \tau \rvert \le \lvert \htau_{\rm PT} \rvert \given \htau_{\rm PT}) \\
    =& \Big[\Phi\Big(c_{\alpha/2} - \sqrt{\frac{\gamma}{1+\gamma}}t \Big) - \Phi\Big(-c_{\alpha/2} - \sqrt{\frac{\gamma}{1+\gamma}}t \Big)\Big]\\
    &\Big[\Phi\Big(\frac{\lvert  \htau_{\rm PT}\rvert}{(1+\gamma)^{-1/2} \sigma_0} - \frac{\gamma}{\sqrt{1+\gamma}} t \Big) - \Phi\Big(-\frac{\lvert  \htau_{\rm PT}\rvert}{(1+\gamma)^{-1/2} \sigma_0}- \frac{\gamma}{\sqrt{1+\gamma}} t\Big)\Big]\\
    &+ \int_{-\infty}^{-c_{\alpha/2} - \sqrt{\frac{\gamma}{1+\gamma}}t} \Big[\Phi\Big(\frac{\lvert  \htau_{\rm PT}\rvert}{(1+\gamma)^{-1/2} \sigma_0} + \sqrt{\gamma} u\Big) - \Phi\Big(-\frac{\lvert  \htau_{\rm PT}\rvert}{(1+\gamma)^{-1/2} \sigma_0} +  \sqrt{\gamma} u\Big)\Big] \phi(u) \d u\\
    &+ \int_{c_{\alpha/2} - \sqrt{\frac{\gamma}{1+\gamma}}t}^\infty \Big[\Phi\Big(\frac{\lvert  \htau_{\rm PT}\rvert}{(1+\gamma)^{-1/2} \sigma_0} + \sqrt{\gamma} u\Big) - \Phi\Big(-\frac{\lvert  \htau_{\rm PT}\rvert}{(1+\gamma)^{-1/2} \sigma_0} + \sqrt{\gamma} u\Big)\Big] \phi(u) \d u.
  \end{align*}
\end{theorem}

Third, the following theorem provides the b-value $b^*_{\rm ST}(\zeta)$ for the soft-thresholding estimator $\htau_{\rm ST}$ with confidence interval $[\htau_{\rm ST}- \hat{L}_{\rm ST} (1+\gamma)^{-1/2} \sigma_0, \htau_{\rm ST}+ \hat{L}_{\rm ST} (1+\gamma)^{-1/2} \sigma_0]$.

\begin{theorem}\label{thm:b_ST}
  The b-value $b^*_{\rm ST}(\zeta)$ is $0$ if
  \[
    \P_{\Delta/\sigma_0 = 0}(\lvert \tilde\tau_{\rm ST} - \tau \rvert \le \lvert \htau_{\rm ST} \rvert \given \htau_{\rm ST}) < 1-\zeta,
  \]
  and is $\infty$ if
  \[
    \P_{\Delta/\sigma_0 = \infty}(\lvert \tilde\tau_{\rm ST} - \tau \rvert \le \lvert \htau_{\rm ST} \rvert \given \htau_{\rm ST}) > 1-\zeta.
  \]
  Otherwise, the b-value $b^*_{\rm ST}(\zeta)$ can be equivalently written as the solution to the equation of $b$:
  \[
    \P_{\Delta/\sigma_0 = b}(\lvert \tilde\tau_{\rm ST} - \tau \rvert \le \lvert \hat\tau_{\rm ST}\rvert \given \hat\tau_{\rm ST}) = 1-\zeta,
  \]
  where $\tilde\tau_{\rm ST}$ is independent and identically distributed as $\hat\tau_{\rm ST}$, with
  \begin{align*}
    &\P_{\Delta/\sigma_0 = t}(\lvert \tilde\tau_{\rm ST} - \tau \rvert \le \lvert \htau_{\rm ST} \rvert \given \htau_{\rm ST}) \\
    =& \Big[\Phi\Big(c_{\alpha/2} - \sqrt{\frac{\gamma}{1+\gamma}}t \Big) - \Phi\Big(-c_{\alpha/2} - \sqrt{\frac{\gamma}{1+\gamma}}t \Big)\Big]\\
    &\Big[\Phi\Big(\frac{\lvert  \hat\tau_{\rm ST}\rvert}{(1+\gamma)^{-1/2} \sigma_0} - \frac{\gamma}{\sqrt{1+\gamma}} t \Big) - \Phi\Big(-\frac{\lvert  \hat\tau_{\rm ST}\rvert}{(1+\gamma)^{-1/2} \sigma_0}- \frac{\gamma}{\sqrt{1+\gamma}} t\Big)\Big]\\
    &+ \int_{-\infty}^{-c_{\alpha/2} - \sqrt{\frac{\gamma}{1+\gamma}}t} \Big[\Phi\Big(\frac{\lvert  \hat\tau_{\rm ST}\rvert}{(1+\gamma)^{-1/2} \sigma_0} + \sqrt{\gamma} (u+c_{\alpha/2})\Big) - \Phi\Big(-\frac{\lvert  \hat\tau_{\rm ST}\rvert}{(1+\gamma)^{-1/2} \sigma_0} +  \sqrt{\gamma} (u+c_{\alpha/2})\Big)\Big] \phi(u) \d u\\
    &+ \int_{c_{\alpha/2} - \sqrt{\frac{\gamma}{1+\gamma}}t}^\infty \Big[\Phi\Big(\frac{\lvert  \hat\tau_{\rm ST}\rvert}{(1+\gamma)^{-1/2} \sigma_0} + \sqrt{\gamma} (u-c_{\alpha/2})\Big) - \Phi\Big(-\frac{\lvert  \hat\tau_{\rm ST}\rvert}{(1+\gamma)^{-1/2} \sigma_0} + \sqrt{\gamma} (u-c_{\alpha/2})\Big)\Big] \phi(u) \d u.
  \end{align*}
\end{theorem}

\subsection{One-sided confidence bounds}\label{sec:one_sided}

In this section, we discuss how to construct one-sided confidence bounds within our framework. We only focus on the lower confidence bound since the upper confidence bound can be constructed in an analogous manner. By using the lower confidence bound, we can conduct one-sided hypothesis testing, such as testing $\tau = 0$ versus $\tau > 0$ or testing $\tau \le 0$ versus $\tau > 0$. The computation of the b-value in the one-sided setting follows the same logic as in the two-sided case, so here we focus on the construction of the one-sided confidence bounds.

First, the following theorem provides the sequence of lower confidence bounds based on $\htau_{\rm PW}$, analogous to Theorem~\ref{thm:L_PW}.

\begin{theorem}\label{thm:oneside_PW}
  Let $\hat{L}_{\rm PW}' = \hat{L}_{\rm PW}'(b, \zeta, \gamma) = c_\zeta + \frac{\gamma}{\sqrt{1+\gamma}} b$. The shortest length lower confidence bound based on $\htau_{\rm PW}$ for $\tau$ satisfying \eqref{eq:coverage} is given by $[\htau_{\rm PW} - \hat{L}_{\rm PW}' (1+\gamma)^{-1/2} \sigma_0, \infty)$.
\end{theorem}

Second, the following theorem provides the sequence of lower confidence bounds based on $\htau_{\rm PT}$, analogous to Theorem~\ref{thm:L_PT}.
\begin{theorem}\label{thm:oneside_PT}
    $\hat{L}_{\rm PT}' = \hat{L}_{\rm PT}'(b, \zeta, \gamma, \alpha)$ is the solution to the following equation of $L$:
    \begin{align*}
      \min_{0 \le t \le b}\P_{\Delta/\sigma_0 = t}(\htau_{\rm PT} - \tau \le L (1+\gamma)^{-1/2} \sigma_0) = 1-\zeta,
    \end{align*}
    where
    \begin{align*}
      &\P_{\Delta/\sigma_0 = t}(\htau_{\rm PT} - \tau \le L (1+\gamma)^{-1/2} \sigma_0) \\
      =& \Big[\Phi\Big(c_{\alpha/2} - \sqrt{\frac{\gamma}{1+\gamma}}t \Big) - \Phi\Big(-c_{\alpha/2} - \sqrt{\frac{\gamma}{1+\gamma}}t \Big)\Big]\Phi\Big(L - \frac{\gamma}{\sqrt{1+\gamma}}t \Big)\\
      &+ \int_{-\infty}^{-c_{\alpha/2} - \sqrt{\frac{\gamma}{1+\gamma}}t} \Phi\Big(L + \sqrt{\gamma} u \Big)\phi(u) \d u + \int_{c_{\alpha/2} - \sqrt{\frac{\gamma}{1+\gamma}}t}^\infty \Phi\Big(L + \sqrt{\gamma} u\Big)\phi(u) \d u.
    \end{align*}
  The shortest length lower confidence bound based on $\htau_{\rm PT}$ for $\tau$ satisfying \eqref{eq:coverage} is given by $[\htau_{\rm PT} - \hat{L}_{\rm PT}' (1+\gamma)^{-1/2} \sigma_0, \infty)$.
\end{theorem}

Third, the following theorem provides the sequence of lower confidence bounds based on $\htau_{\rm ST}$, analogous to Theorem~\ref{thm:L_ST}.

\begin{theorem}\label{thm:oneside_ST}
  For any $L > 0$, $\P_\Delta(\htau_{\rm ST}- \tau \le L (1+\gamma)^{-1/2} \sigma_0)$ is monotonically decreasing in $\Delta$. Then $\hat{L}_{\rm ST}' = \hat{L}_{\rm ST}'(b, \zeta, \gamma, \alpha)$ is the solution to the equation of $L$:
  \begin{align*}
    \P_{\Delta/\sigma_0 = b}(\htau_{\rm ST}- \tau \le L (1+\gamma)^{-1/2} \sigma_0) = 1-\zeta,
  \end{align*} 
  where
  \begin{align*}
    &\P_{\Delta/\sigma_0 = t}(\htau_{\rm ST}- \tau \le L (1+\gamma)^{-1/2} \sigma_0) \\
    =& \Big[\Phi\Big(c_{\alpha/2} - \sqrt{\frac{\gamma}{1+\gamma}}t \Big) - \Phi\Big(-c_{\alpha/2} - \sqrt{\frac{\gamma}{1+\gamma}}t \Big)\Big]\Phi\Big(L - \frac{\gamma}{\sqrt{1+\gamma}}t \Big)\\
    &+ \int_{-\infty}^{-c_{\alpha/2} - \sqrt{\frac{\gamma}{1+\gamma}}t} \Phi\Big(L + \sqrt{\gamma} (u+c_{\alpha/2})\Big)\phi(u) \d u + \int_{c_{\alpha/2} - \sqrt{\frac{\gamma}{1+\gamma}}t}^\infty \Phi\Big(L + \sqrt{\gamma} (u-c_{\alpha/2})\Big)\phi(u) \d u.
  \end{align*}
  The shortest length lower confidence bound based on $\htau_{\rm ST}$ for $\tau$ satisfying \eqref{eq:coverage} is given by $[\htau_{\rm ST} - \hat{L}_{\rm ST}' (1+\gamma)^{-1/2} \sigma_0, \infty)$.
\end{theorem}

\subsection{Bias-dependent point estimator}\label{sec:bias_dependent}

A natural question arises: since the combined estimator $\hat\tau$ itself does not depend on the unknown true bias $\Delta$, could we instead construct confidence intervals using a bias-dependent point estimator, i.e., one that explicitly depends on $\Delta$? More generally, suppose we have an estimator whose expectation is $\tau + g(\Delta)$, where $g(\cdot)$ is a known function of the bias. There are two natural approaches to construct a confidence interval for $\tau$ in this setting. The first approach is to estimate the bias and subtract it from the point estimator. Specifically, we can construct an adaptive estimator of the form $\hat\tau - \widehat{g(\Delta)}$, where $\widehat{g(\Delta)}$ is an estimate of $g(\Delta)$. The second approach is to construct a point estimator $\hat\tau - g(\Delta)$ and its corresponding confidence interval for each possible value of $\Delta$ within the bias bound, and then take the union of all such intervals. However, the union of these confidence intervals is conservative since the biased estimator cannot simultaneously exhibit multiple bias values within the bound.
% Instead, we construct the confidence interval based on the bias-independent estimator $\hat\tau$, i.e., one that does not depend on $\Delta$, but make the confidence interval achieve correct coverage uniformly over all possible values of $\Delta$ within the bias bound $b$ by using $\cI(b, \zeta)$. We show the efficiency of this approach in Sections~\ref{sec:simple} and \ref{sec:multivariate}, i.e., yielding a shorter confidence interval. 
See also \cite{armstrong2020simple} for a related argument: they show that using critical values to account for potential bias is more efficient than subtracting an estimate of the bias from the point estimator.

We now illustrate why constructing confidence intervals using a prespecified point estimator is more efficient than using a bias-dependent one.
% We now illustrate why constructing confidence intervals using a bias-independent point estimator is more efficient than using a bias-dependent one under the context of the precision-weighted estimator.

For the first approach, a natural estimate of the bias is $\hat\Delta = \hat\tau_1 - \hat\tau_0$. We take the precision-weighted estimator as an example. Subtracting this estimated bias from the precision-weighted estimator gives $\htau_0 + \frac{\gamma}{1+\gamma} (\htau_1 - \htau_0) - \frac{\gamma}{1+\gamma} \hat\Delta = \htau_0$. Thus, after bias correction, the estimator collapses to the unbiased estimator $\hat\tau_0$. In other words, this approach discards the more precise estimator. The issue is that, when constructing the confidence interval, we need to account for the randomness in $\hat\Delta$, which makes the confidence interval wider. As a result, this approach provides no efficiency gain over the unbiased estimator.

For the second approach, suppose we know the true bias $\Delta$. Then the bias-corrected combined estimator is $\htau_{\Delta} = \htau - g(\Delta)$. Suppose the shortest length symmetric centered confidence interval based on $\htau_{\Delta}$ for $\tau$ is given by $[\htau_{\Delta} - \hat{L}_{\zeta}, \htau_{\Delta} + \hat{L}_{\zeta}]$, where $\hat{L}_{\zeta}$ does not depends on $\Delta$. When we only have $\lvert \Delta \rvert \le b$, we take the union of these confidence intervals over all possible $\Delta$ with $\lvert \Delta \rvert \le b$. The resulting confidence interval is
\begin{align}\label{eq:CI_union}
  \bigcup_{\Delta:\lvert \Delta \rvert \le b} [\htau_{\Delta} - \hat{L}_{\zeta}, \htau_{\Delta} + \hat{L}_{\zeta}] = \left[\htau - \left(\max_{\Delta:\lvert \Delta \rvert \le b} g(\Delta) + \hat{L}_{\zeta}\right), \htau_{\Delta} + \left(\max_{\Delta:\lvert \Delta \rvert \le b} g(\Delta) + \hat{L}_{\zeta}\right)\right].
\end{align}
By Definition~\ref{def:CI}, the shortest length symmetric centered confidence interval for $\tau$ based on $\htau$ is given by $[\htau- \hat{L}_{\zeta}(b), \htau + \hat{L}_{\zeta}(b)]$, where
\begin{align}\label{eq:L_zeta_b}
  \hat{L}_{\zeta}(b) = \argmin_{L \ge 0} \left\{ \inf_{\Delta: \lvert \Delta \rvert \le b} \P_\Delta(\lvert \hat\tau - \tau \rvert \le L) \ge 1 - \zeta \right\}.
\end{align}
Comparing the confidence interval based on bias-dependent point estimator \eqref{eq:CI_union} with the confidence interval based on prespecified point estimator \eqref{eq:L_zeta_b}, we note that
\[
  \left(\max_{\Delta:\lvert \Delta \rvert \le b} g(\Delta) + \hat{L}_{\zeta}\right) \in \left\{ \inf_{\Delta: \lvert \Delta \rvert \le b} \P_\Delta(\lvert \hat\tau - \tau \rvert \le L) \ge 1 - \zeta \right\}.
\]
Indeed, for any $\Delta'$ with $\lvert \Delta' \rvert \le b$,
\[
  \P_{\Delta'}\left(\lvert \hat\tau - \tau \rvert \le \max_{\Delta:\lvert \Delta \rvert \le b} g(\Delta) + \hat{L}_{\zeta}\right) \ge  \P_{\Delta'}\left(\lvert \hat\tau - \tau \rvert \le g(\Delta') + \hat{L}_{\zeta} \right) \ge 1 - \zeta.
\]
Therefore
\[
  \hat{L}_{\zeta}(b) \le \max_{\Delta:\lvert \Delta \rvert \le b} g(\Delta) + \hat{L}_{\zeta},
\]
with equality only when $b=0$. Hence, the confidence interval based on prespecified point estimator is strictly narrower when the bias bound is nonzero. This comparison highlights a key principle of our framework: instead of accounting for the bias in the point estimation through subtracting the bias, it is more efficient to account for the bias in the confidence interval construction through critical values.

\subsection{More details for the dependence case}\label{sec:dependence}

In this section, we provide more details on the generalization to the dependence case in Section~\ref{sec:dependent} of the main paper. In the reparametrization in \eqref{eq:htau1_prime}, one may be concerned about how scaling the bias works and how to interpret the relative bias. We compute the confidence intervals and the b-value by the following three steps.

\begin{enumerate}[label=(\alph*)]
  \item Reparametrization: We apply the transformation in \eqref{eq:htau1_prime} to obtain the transformed biased estimator $\htau_1'$. Under this transformation, $\hat\tau_0$ remains an unbiased estimator of $\tau$, and $\hat\tau_1'$ becomes a biased estimator that is independent of $\hat\tau_0$. We then construct the combined estimator $\hat\tau$ using these two independent components.
  \item Compute the confidence intervals and the b-value in the transformed problem: We next work in the transformed problem. Let the relative bias in the transformed parametrization be $\Delta'/\sigma_0$ with bound $b'$. Using Definition~\ref{def:CI}, we construct the confidence interval $\htau - \cI(b', \zeta)$ such that
  \[
    \inf_{\Delta': \lvert \Delta'/\sigma_0 \rvert \le b'} \P_{\Delta'}(\tau \in \htau - \cI(b', \zeta)) \ge 1 - \zeta,
  \]
  and compute the corresponding b-value $b^{*’}$ in the transformed problem.
  \item Transform the confidence intervals and the b-value back to the original problem: Finally, we transform the results back to the original problem. Under the transformation in \eqref{eq:htau1_prime}, the relative biases in the two parametrizations are related by
  \[
    \lvert \Delta'/\sigma_0 \rvert \le b' \iff \lvert \Delta/\sigma_0 \rvert \le \lvert 1 - \rho \sigma_1/\sigma_0 \rvert b'.
  \]
  Thus, a bias bound $b’$ in the transformed parametrization corresponds to a bias bound $b = \lvert 1 - \rho \sigma_1/\sigma_0 \rvert b'$ in the original parametrization. Using this relationship, the confidence interval $\htau - \cI(b, \zeta)$ satisfies Definition~\ref{def:CI} in the original parametrization:
  \[
    \inf_{\Delta: \lvert \Delta/\sigma_0 \rvert \le b} \P_{\Delta}(\tau \in \htau - \cI(b, \zeta)) \ge 1 - \zeta,
  \]
  and the b-value in the original parametrization is $b^* = \lvert 1 - \rho \sigma_1/\sigma_0 \rvert b^{*'}$.
\end{enumerate}

\subsection{Relationship between the generic soft-thresholding estimator and the estimator in \citet{j1981estimation, bickel1984parametric}}\label{sec:efron--morris}

In this section, we discuss the relationship between the generic soft-thresholding estimator and the estimator in \citet{j1981estimation, bickel1984parametric} in Section~\ref{sec:point_multi} of the main paper. We consider the special setting studied in \cite{j1981estimation} and \cite{bickel1984parametric}, where the covariance matrices are proportional and isotropic. Specifically, assume $\bm{\Sigma}_0 = \sigma_0^2 \bm{I}_d$ and $\bm{\Sigma}_0=\gamma\bm{\Sigma}_1$ for some $\gamma>0$. The $\gamma$ here generalizes the variance ratio $\sigma_0^2/\sigma_1^2$ in the univariate setting. Let $\mathsf C \ge 0$ be a user-chosen constant and let $\rho_{\mathsf C}(r)$ denote the ratio of Bessel functions, whose explicit form is given in Lemma 3 of \cite{j1981estimation}. The constant $\mathsf C$ here plays the same role as the critical value $c_{\alpha/2}$ in the univariate setting. The estimator in \cite{j1981estimation} and \cite{bickel1984parametric} is then given by choosing the shrinkage function $h_q^*(r) = \rho_{\mathsf C/\sigma_0^2}(r)$ where the threshold $q = q(\mathsf C)$ is defined as the solution to
\[
h_q^*(q) = 1 \iff \rho_{\mathsf C/\sigma_0^2}(q) = 1.
\]

Suppose we choose the function $h_q^*(\cdot)$ as the estimator in \cite{j1981estimation} and \cite{bickel1984parametric}, i.e., $h_q^*(r) = \rho_{\mathsf C/\sigma_0^2}(r)$ with $q$ given above. In the univariate case $d=1$, the estimator $\bm{\hat\tau}_{\rm ST}$ reduces to the univariate soft-thresholding estimator $\hat\tau_{\rm ST}$. In this case, the user-chosen constant $\mathsf C$ in $\bm{\hat\tau}_{\rm ST}$ corresponds one-to-one to the significance level $\alpha$, or equivalently to the critical value $c_{\alpha/2}$ in $\hat\tau_{\rm ST}$. Consider the special case $\mathsf C=0$. When $d \le 2$, the estimator $\bm{\hat\tau}_{\rm ST}$ reduces to the unbiased estimator $\bm{\hat\tau}_0$. However, when $d \ge 3$, the behavior of $\bm{\hat\tau}_{\rm ST}$ is different. When $\mathsf C=0$ and $d \ge 3$, we have $\rho_{\mathsf C/\sigma_0^2}(r) = \rho_0(r) = 2(d-2)/r$, and therefore the threshold $q$ is given by $q = 2(d-2)$. If $\lVert (\bm{\Sigma}_0+ \bm{\Sigma}_1)^{-1/2}(\bm{\htau}_1 - \bm{\htau}_0) \rVert_2^2 \le q$, or equivalently, $\lVert \bm{\htau}_1 - \bm{\htau}_0 \rVert_2^2 \le 2(d-2)(\sigma_0^2 + \sigma_1^2)$, the estimator $\bm{\hat\tau}_{\rm ST}$ reduces to the precision-weighted estimator $\bm{\htau}_0 + \frac{\gamma}{1+\gamma}(\bm{\htau}_1 - \bm{\htau}_0)$. If instead $\lVert (\bm{\Sigma}_0+ \bm{\Sigma}_1)^{-1/2}(\bm{\htau}_1 - \bm{\htau}_0) \rVert_2^2 > q$, or equivalently, $\lVert \bm{\htau}_1 - \bm{\htau}_0 \rVert_2^2 > 2(d-2)(\sigma_0^2 + \sigma_1^2)$, the estimator $\bm{\hat\tau}_{\rm ST}$ takes the James--Stein shrinkage form \citep{green1991james}:
\[
  \bm{\htau}_0 + h_q^*(\lVert (\bm{\Sigma}_0+ \bm{\Sigma}_1)^{-1/2}(\bm{\htau}_1 - \bm{\htau}_0) \rVert_2^2)(\bm{\Sigma}_0^{-1} + \bm{\Sigma}_1^{-1})^{-1}\bm{\Sigma}_1^{-1}(\bm{\htau}_1 - \bm{\htau}_0) = \bm{\htau}_0 + \frac{2(d-2)\sigma_0^2}{\lVert \bm{\htau}_1 - \bm{\htau}_0 \rVert_2^2}(\bm{\htau}_1 - \bm{\htau}_0).
\]

\section{Proofs of the main results}\label{sec:proof_main}

\subsection{Proof of Theorem~\ref{thm:L_PW}}

\begin{proof}[Proof of Theorem~\ref{thm:L_PW}]
  By the definition of $\htau_{\rm PW}$, under Assumption~\ref{ass:normal}, we have 
  \[
    \htau_{\rm PW} \sim N\left(\tau + \frac{\gamma}{1+\gamma} \Delta, \frac{1}{1+\gamma}\sigma_0^2\right).
  \]
  Equivalently, after standardization,
  \begin{align*}
    (1+\gamma)^{1/2} \sigma_0^{-1} (\htau_{\rm PW} - \tau) \sim N\left(\frac{\gamma}{\sqrt{1+\gamma}} \frac{\Delta}{\sigma_0}, 1\right).
  \end{align*}
  
  We now evaluate the worst-case coverage probability of the confidence interval
  \[
    [\htau_{\rm PW} - L (1+\gamma)^{-1/2} \sigma_0, \htau_{\rm PW} + L (1+\gamma)^{-1/2} \sigma_0].
  \]
  We have
  \begin{align*}
    & \inf_{\Delta: \lvert \Delta/\sigma_0 \rvert \le b} \P_\Delta(\tau \in [\htau_{\rm PW} - L (1+\gamma)^{-1/2} \sigma_0, \htau_{\rm PW} + L (1+\gamma)^{-1/2} \sigma_0])\\
    =& \inf_{\Delta: \lvert \Delta/\sigma_0 \rvert \le b} \P_\Delta(\lvert \htau_{\rm PW} - \tau \rvert \le L (1+\gamma)^{-1/2} \sigma_0)\\
    =& \inf_{\Delta: \lvert \Delta/\sigma_0 \rvert \le b} \P_\Delta(\lvert (1+\gamma)^{1/2} \sigma_0^{-1} (\htau_{\rm PW} - \tau) \rvert \le L)\\
    =& \inf_{\Delta: \lvert \Delta/\sigma_0 \rvert \le b} \P_\Delta\left(\left\lvert N\left(\frac{\gamma}{\sqrt{1+\gamma}} \frac{\Delta}{\sigma_0}, 1\right) \right\rvert \le L\right)\\
    =& \P\left( \left\lvert N\left(\sup_{\Delta: \lvert \Delta/\sigma_0 \rvert \le b} \frac{\gamma}{\sqrt{1+\gamma}} \frac{\Delta}{\sigma_0}, 1\right) \right\rvert \le L \right)\\
    =& \P\left( \left\lvert N\left(\frac{\gamma}{\sqrt{1+\gamma}} b, 1\right) \right\rvert \le L \right)\\
    =& \Phi\Big(L - \frac{\gamma}{\sqrt{1+\gamma}} b \Big) - \Phi\Big(-L - \frac{\gamma}{\sqrt{1+\gamma}} b \Big).
  \end{align*}
  Therefore, the proof is complete since $\hat{L}_{\rm PW}$ is the solution to
  \[
    \Phi\Big(L - \frac{\gamma}{\sqrt{1+\gamma}} b \Big) - \Phi\Big(-L - \frac{\gamma}{\sqrt{1+\gamma}} b \Big) = 1-\zeta.
  \] 
\end{proof}

\subsection{Proof of Theorem~\ref{thm:L_PT}}

\begin{proof}[Proof of Theorem~\ref{thm:L_PT}]
Recall that the pretest estimator $\htau_{\rm PT}$ is defined as
\[
  \htau_{\rm PT} = \htau_0 + \frac{\gamma}{1+\gamma} (\htau_1 - \htau_0) \ind\left(\lvert \htau_1 - \htau_0 \rvert \le \sigma c_{\alpha/2}\right).
\]
To characterize the distribution of $\htau_{\rm PT} - \tau$, note that
\begin{align*}
  \begin{pmatrix} \htau_0 + \frac{\gamma}{1+\gamma} (\htau_1-\htau_0)-\tau \\ \htau_1-\htau_0 \end{pmatrix} \sim N\left(\begin{pmatrix} \frac{\gamma}{1+\gamma} \Delta \\ \Delta\end{pmatrix}, \begin{pmatrix} \frac{1}{1+\gamma} \sigma_0^2 & 0 \\
  0 & \sigma_0^2 + \sigma_1^2
  \end{pmatrix}\right).
\end{align*}

We first consider the event that the pretest fails to reject, $\lvert \htau_1 - \htau_0 \rvert \le \sigma c_{\alpha/2}$. Conditional on this event, $\htau_{\rm PT} = \htau_0 + \frac{\gamma}{1+\gamma} (\htau_1 - \htau_0)$, and therefore
\begin{align*}
  (\htau_{\rm PT} - \tau) \given \{\lvert \htau_1 - \htau_0 \rvert \le \sigma c_{\alpha/2}\} \sim N\left(\frac{\gamma}{1+\gamma} \Delta, \frac{1}{1+\gamma} \sigma_0^2\right).
\end{align*}
Equivalently, after standardization,
\begin{align*}
  (1+\gamma)^{1/2} \sigma_0^{-1} (\htau_{\rm PT} - \tau) \given \{\lvert \htau_1 - \htau_0 \rvert \le \sigma c_{\alpha/2}\} \sim N\left(\frac{\gamma}{\sqrt{1+\gamma}} \frac{\Delta}{\sigma_0}, 1\right).
\end{align*}

We next consider the event that the pretest rejects, $\lvert \htau_1 - \htau_0 \rvert > \sigma c_{\alpha/2}$. Conditional on $\sigma^{-1}(\htau_1 - \htau_0) = u$, we have
\[
  \frac{\gamma}{1+\gamma} (\htau_1-\htau_0) = \sqrt{\frac{\gamma}{1+\gamma}} \sigma_0 u.
\]
Since $\htau_{\rm PT} = \htau_0$ conditional on this event, we have
\begin{align*}
  &(\htau_{\rm PT} - \tau) \given \{\sigma^{-1}(\htau_1 - \htau_0) = u, \lvert \htau_1 - \htau_0 \rvert > \sigma c_{\alpha/2}\} \\
  =& (\htau_0 - \tau) \given \{\sigma^{-1}(\htau_1 - \htau_0) = u, \lvert \htau_1 - \htau_0 \rvert > \sigma c_{\alpha/2}\}\\
  =& (\htau_0 + \frac{\gamma}{1+\gamma} (\htau_1-\htau_0)-\tau) \given \{\sigma^{-1}(\htau_1 - \htau_0) = u, \lvert \htau_1 - \htau_0 \rvert > \sigma c_{\alpha/2}\} - \sqrt{\frac{\gamma}{1+\gamma}} \sigma_0 u\\
  \sim& N\left(\frac{\gamma}{1+\gamma}\Delta - \sqrt{\frac{\gamma}{1+\gamma}} \sigma_0 u, \frac{1}{1+\gamma} \sigma_0^2\right),
\end{align*}
which implies
\begin{align*}
  (1+\gamma)^{1/2} \sigma_0^{-1} (\htau_{\rm PT} - \tau) \given \{\sigma^{-1}(\htau_1 - \htau_0) = u, \lvert \htau_1 - \htau_0 \rvert > \sigma c_{\alpha/2}\} \sim N\left(\frac{\gamma}{\sqrt{1+\gamma}} \frac{\Delta}{\sigma_0} - \sqrt{\gamma} u, 1\right).
\end{align*}

We now evaluate the worst-case coverage probability of the confidence interval
\[
  [\htau_{\rm PT} - L (1+\gamma)^{-1/2} \sigma_0, \htau_{\rm PT} + L (1+\gamma)^{-1/2} \sigma_0].
\]
We have
\begin{align*}
  & \inf_{\Delta: \lvert \Delta/\sigma_0 \rvert \le b} \P_\Delta(\tau \in [\htau_{\rm PT} - L (1+\gamma)^{-1/2} \sigma_0, \htau_{\rm PT} + L (1+\gamma)^{-1/2} \sigma_0])\\
  =& \inf_{\Delta: \lvert \Delta/\sigma_0 \rvert \le b} \P_\Delta(\lvert \htau_{\rm PT} - \tau \rvert \le L (1+\gamma)^{-1/2} \sigma_0).
\end{align*}
Fixing $\Delta/\sigma_0=t$, we decompose the probability according to whether the pretest accepts or rejects:
\begin{align*}
  &\P_{\Delta/\sigma_0 = t}(\lvert \htau_{\rm PT} - \tau \rvert \le L (1+\gamma)^{-1/2} \sigma_0)\\
  =&\P_{\Delta/\sigma_0 = t}(\lvert (1+\gamma)^{1/2} \sigma_0^{-1} (\htau_{\rm PT} - \tau) \rvert \le L)\\
  =&\P_{\Delta/\sigma_0 = t}(\lvert (1+\gamma)^{1/2} \sigma_0^{-1} (\htau_{\rm PT} - \tau) \rvert \le L, \lvert \htau_1 - \htau_0 \rvert \le \sigma c_{\alpha/2}) \\
  &+ \P_{\Delta/\sigma_0 = t}(\lvert (1+\gamma)^{1/2} \sigma_0^{-1} (\htau_{\rm PT} - \tau) \rvert \le L, \lvert \htau_1 - \htau_0 \rvert > \sigma c_{\alpha/2})\\ 
  =& \P\left( \left\lvert N\left(\frac{\gamma}{\sqrt{1+\gamma}} t, 1\right) \right\rvert \le L \right) \P\left( \left\lvert N\left(\sqrt{\frac{\gamma}{1+\gamma}} t, 1\right) \right\rvert \le c_{\alpha/2} \right)\\
  &+ \E_{U \sim N(\sqrt{\frac{\gamma}{1+\gamma}} t, 1)}\left[\P\left( \left\lvert N\left(\frac{\gamma}{\sqrt{1+\gamma}} t - \sqrt{\gamma} U, 1\right) \right\rvert \le L \right) \ind\left(\lvert U \rvert > c_{\alpha/2}\right)\right]\\
  =& \Big[\Phi\Big(c_{\alpha/2} - \sqrt{\frac{\gamma}{1+\gamma}}t \Big) - \Phi\Big(-c_{\alpha/2} - \sqrt{\frac{\gamma}{1+\gamma}}t \Big)\Big]\Big[\Phi\Big(L - \frac{\gamma}{\sqrt{1+\gamma}}t \Big) - \Phi\Big(-L - \frac{\gamma}{\sqrt{1+\gamma}}t \Big)\Big]\\
  &+ \int_{-\infty}^{-c_{\alpha/2} - \sqrt{\frac{\gamma}{1+\gamma}}t} \Big[\Phi\Big(L + \sqrt{\gamma} u \Big) - \Phi\Big(-L + \sqrt{\gamma} u\Big)\Big] \phi(u) \d u\\
  &+ \int_{c_{\alpha/2} - \sqrt{\frac{\gamma}{1+\gamma}}t}^\infty \Big[\Phi\Big(L + \sqrt{\gamma} u\Big) - \Phi\Big(-L + \sqrt{\gamma} u\Big)\Big] \phi(u) \d u,
\end{align*}
which yields exactly the expression in \eqref{eq:P_PT}.

Finally, since the coverage probability is symmetric in $\Delta$, the worst case over
$\lvert \Delta/\sigma_0\rvert\le b$ is attained for some $t\in[0,b]$. Choosing $\hat L_{\rm PT}$ to satisfy
\[
  \min_{0 \le t \le b}\P_{\Delta/\sigma_0 = t}(\lvert \htau_{\rm PT} - \tau \rvert \le L (1+\gamma)^{-1/2} \sigma_0) = 1-\zeta
\]
completes the proof.
  
\end{proof}

\subsection{Proof of Theorem~\ref{thm:L_ST}}

\begin{proof}[Proof of Theorem~\ref{thm:L_ST}]
Recall that the soft-thresholding estimator $\htau_{\rm ST}$ is defined as 
\[
  \htau_{\rm ST}= \htau_0 + \frac{\gamma}{1+\gamma} (\htau_1 - \htau_0)\ind\left(\lvert \htau_1 - \htau_0 \rvert \le \sigma c_{\alpha/2}\right) 
  + \frac{\gamma}{1+\gamma} \sigma c_{\alpha/2}\,\sign(\htau_1 - \htau_0)\ind\left(\lvert \htau_1 - \htau_0 \rvert > \sigma c_{\alpha/2}\right).
\]
As before, note that
\begin{align*}
  \begin{pmatrix} \htau_0 + \frac{\gamma}{1+\gamma} (\htau_1-\htau_0)-\tau \\ \htau_1-\htau_0 \end{pmatrix} \sim N\left(\begin{pmatrix} \frac{\gamma}{1+\gamma} \Delta \\ \Delta\end{pmatrix}, \begin{pmatrix} \frac{1}{1+\gamma} \sigma_0^2 & 0 \\
  0 & \sigma_0^2 + \sigma_1^2
  \end{pmatrix}\right).
\end{align*}

We first consider the event that the pretest accepts, $\lvert \htau_1 - \htau_0 \rvert \le \sigma c_{\alpha/2}$. Conditional on this event, $\htau_{\rm ST} = \htau_0 + \frac{\gamma}{1+\gamma} (\htau_1 - \htau_0)$, and therefore
\begin{align*}
  (\htau_{\rm ST} - \tau) \given \{\lvert \htau_1 - \htau_0 \rvert \le \sigma c_{\alpha/2}\} \sim N\left(\frac{\gamma}{1+\gamma} \Delta, \frac{1}{1+\gamma} \sigma_0^2\right).
\end{align*}
Equivalently, after standardization,
\begin{align*}
  (1+\gamma)^{1/2} \sigma_0^{-1} (\htau_{\rm ST} - \tau) \given \{\lvert \htau_1 - \htau_0 \rvert \le \sigma c_{\alpha/2}\} \sim N\left(\frac{\gamma}{\sqrt{1+\gamma}} \frac{\Delta}{\sigma_0}, 1\right).
\end{align*}

We next consider the event that the pretest rejects, $\lvert \htau_1 - \htau_0 \rvert > \sigma c_{\alpha/2}$. Conditional on $\sigma^{-1}(\htau_1 - \htau_0) = u$, we have
\[
  \frac{\gamma}{1+\gamma} (\htau_1-\htau_0) = \sqrt{\frac{\gamma}{1+\gamma}} \sigma_0 u.
\]
Since $\htau_{\rm ST} = \htau_0 + \frac{\gamma}{1+\gamma} \sigma c_{\alpha/2}\,\sign(\htau_1 - \htau_0)$ conditional on this event, we have
\begin{align*}
  &(\htau_{\rm ST} - \tau) \given \{\sigma^{-1}(\htau_1 - \htau_0) = u, \lvert \htau_1 - \htau_0 \rvert > \sigma c_{\alpha/2}\} \\
  =& \htau_0 - \tau + \frac{\gamma}{1+\gamma} \sigma c_{\alpha/2}\,\sign(\htau_1 - \htau_0)\given \{\sigma^{-1}(\htau_1 - \htau_0) = u, \lvert \htau_1 - \htau_0 \rvert > \sigma c_{\alpha/2}\}\\
  =& \htau_0 + \frac{\gamma}{1+\gamma} (\htau_1-\htau_0)-\tau \given \{\sigma^{-1}(\htau_1 - \htau_0) = u, \lvert \htau_1 - \htau_0 \rvert > \sigma c_{\alpha/2}\} - \sqrt{\frac{\gamma}{1+\gamma}} \sigma_0 [u - c_{\alpha/2}\,\sign(u)]\\
  \sim& N\left(\frac{\gamma}{1+\gamma}\Delta - \sqrt{\frac{\gamma}{1+\gamma}} \sigma_0 [u - c_{\alpha/2}\,\sign(u)], \frac{1}{1+\gamma} \sigma_0^2\right),
\end{align*}
which implies
\begin{align*}
  &(1+\gamma)^{1/2} \sigma_0^{-1} (\htau_{\rm ST} - \tau) \given \{\sigma^{-1}(\htau_1 - \htau_0) = u, \lvert \htau_1 - \htau_0 \rvert > \sigma c_{\alpha/2}\} \sim N\left(\frac{\gamma}{\sqrt{1+\gamma}} \frac{\Delta}{\sigma_0} - \sqrt{\gamma} [u - c_{\alpha/2}\,\sign(u)], 1\right).
\end{align*}
  
We now evaluate the worst-case coverage probability of the confidence interval
\[
  [\htau_{\rm ST} - L (1+\gamma)^{-1/2} \sigma_0, \htau_{\rm ST} + L (1+\gamma)^{-1/2} \sigma_0].
\]
We have
\begin{align*}
  & \inf_{\Delta: \lvert \Delta/\sigma_0 \rvert \le b} \P_\Delta(\tau \in [\htau_{\rm ST} - L (1+\gamma)^{-1/2} \sigma_0, \htau_{\rm ST} + L (1+\gamma)^{-1/2} \sigma_0])\\
  =& \inf_{\Delta: \lvert \Delta/\sigma_0 \rvert \le b} \P_\Delta(\lvert \htau_{\rm ST} - \tau \rvert \le L (1+\gamma)^{-1/2} \sigma_0).
\end{align*}
Fixing $\Delta/\sigma_0=t$, we decompose the probability according to whether the pretest accepts or rejects:
\begin{align*}
  &\P_{\Delta/\sigma_0 = t}(\lvert \htau_{\rm ST} - \tau \rvert \le L (1+\gamma)^{-1/2} \sigma_0)\\
  =&\P_{\Delta/\sigma_0 = t}(\lvert (1+\gamma)^{1/2} \sigma_0^{-1} (\htau_{\rm ST} - \tau) \rvert \le L)\\
  =&\P_{\Delta/\sigma_0 = t}(\lvert (1+\gamma)^{1/2} \sigma_0^{-1} (\htau_{\rm ST} - \tau) \rvert \le L, \lvert \htau_1 - \htau_0 \rvert \le \sigma c_{\alpha/2}) \\
  &+ \P_{\Delta/\sigma_0 = t}(\lvert (1+\gamma)^{1/2} \sigma_0^{-1} (\htau_{\rm ST} - \tau) \rvert \le L, \lvert \htau_1 - \htau_0 \rvert > \sigma c_{\alpha/2})\\ 
  =& \P\left( \left\lvert N\left(\frac{\gamma}{\sqrt{1+\gamma}} t, 1\right) \right\rvert \le L \right) \P\left( \left\lvert N\left(\sqrt{\frac{\gamma}{1+\gamma}} t, 1\right) \right\rvert \le c_{\alpha/2} \right)\\
  &+ \E_{U \sim N(\sqrt{\frac{\gamma}{1+\gamma}} t, 1)}\left[\P\left( \left\lvert N\left(\frac{\gamma}{\sqrt{1+\gamma}} t - \sqrt{\gamma} [U - c_{\alpha/2}\,\sign(U)], 1\right) \right\rvert \le L \right) \ind\left(\lvert U \rvert > c_{\alpha/2}\right)\right]\\
  =& \Big[\Phi\Big(c_{\alpha/2} - \sqrt{\frac{\gamma}{1+\gamma}}t \Big) - \Phi\Big(-c_{\alpha/2} - \sqrt{\frac{\gamma}{1+\gamma}}t \Big)\Big]\Big[\Phi\Big(L - \frac{\gamma}{\sqrt{1+\gamma}}t \Big) - \Phi\Big(-L - \frac{\gamma}{\sqrt{1+\gamma}}t \Big)\Big]\\
  &+ \int_{-\infty}^{-c_{\alpha/2} - \sqrt{\frac{\gamma}{1+\gamma}}t} \Big[\Phi\Big(L + \sqrt{\gamma} (u+c_{\alpha/2})\Big) - \Phi\Big(-L + \sqrt{\gamma} (u+c_{\alpha/2})\Big)\Big] \phi(u) \d u\\
  &+ \int_{c_{\alpha/2} - \sqrt{\frac{\gamma}{1+\gamma}}t}^\infty \Big[\Phi\Big(L + \sqrt{\gamma} (u-c_{\alpha/2})\Big) - \Phi\Big(-L + \sqrt{\gamma} (u-c_{\alpha/2})\Big)\Big] \phi(u) \d u.
\end{align*}
which yields exactly the expression in \eqref{eq:P_ST}.

To establish the monotonicity of the coverage probability in $\lvert \Delta \rvert$, fix $L>0$ and write $t=\Delta/\sigma_0$. 
Since $\P_{\Delta}(\lvert \hat\tau_{\rm ST}-\tau\rvert\le L(1+\gamma)^{-1/2}\sigma_0)$
is symmetric about $\Delta = 0$, it suffices to consider $t\ge0$. By Lemma~\ref{lemma:clt_ST}, we have
\[
(1+\gamma)^{1/2}\sigma_0^{-1}(\hat\tau_{\rm ST}-\tau)
\stackrel{d}{=}
Z + \frac{\gamma}{\sqrt{1+\gamma}}\,t
- \sqrt{\gamma}\,S(U),
\]
where $Z\sim N(0,1)$ is independent of
$U\sim N(\sqrt{\frac{\gamma}{1+\gamma}}\,t,1)$ and
$S(u)=\sign(u)(|u|-c_{\alpha/2})_+$.
Consequently, we can write the coverage probability as
\[
  \P_{\Delta/\sigma_0 = t}(\lvert \htau_{\rm ST} - \tau \rvert \le L (1+\gamma)^{-1/2} \sigma_0) = \E\left[
\Phi\left(L-\mu_t(U)\right)
-\Phi\left(-L-\mu_t(U)\right)
\right],
\]
with $\mu_t(U)=\frac{\gamma}{\sqrt{1+\gamma}}\,t-\sqrt{\gamma}\,S(U)$. The function $\mu\mapsto\Phi(L-\mu)-\Phi(-L-\mu)$ is even and strictly
decreasing in $\lvert \mu \rvert$. Moreover, increasing $t$ shifts the distribution of $U$ away from zero,
which increases $\lvert \mu_t(U) \rvert$ in the sense of stochastic order.
Therefore, the coverage probability
is nonincreasing in $t\ge0$, or equivalently, monotonically decreasing in $\lvert \Delta\rvert$.

Finally, since the coverage probability is symmetric in $\Delta$ and monotonically decreasing in $\lvert \Delta \rvert$, the worst case over
$\lvert \Delta/\sigma_0\rvert\le b$ is attained at $\Delta/\sigma_0 = b$. Choosing $\hat L_{\rm ST}$ to satisfy
\[
  \P_{\Delta/\sigma_0 = b}(\lvert \htau_{\rm ST}- \tau \rvert \le L (1+\gamma)^{-1/2} \sigma_0) = 1-\zeta
\]
completes the proof.

\end{proof}

\subsection{Proof of Theorem~\ref{prop:CI_multi}}

\begin{proof}[Proof of Theorem~\ref{prop:CI_multi}]
By the definition of $\bm{\htau}_{\rm PW}$, under Assumption~\ref{ass:normal_multi}, we have 
\[
  \bm{\htau}_{\rm PW} \sim N(\bm{\tau} + (\bm{\Sigma}_0^{-1} + \bm{\Sigma}_1^{-1})^{-1} \bm{\Sigma}_1^{-1} \bm{\Delta}, (\bm{\Sigma}_0^{-1} + \bm{\Sigma}_1^{-1})^{-1} ).
\]
Equivalently, after standardization,
\begin{align*}
  (\bm{\Sigma}_0^{-1} + \bm{\Sigma}_1^{-1})^{1/2} (\bm{\htau}_{\rm PW} - \bm{\tau}) \sim N((\bm{\Sigma}_0^{-1} + \bm{\Sigma}_1^{-1})^{-1/2} \bm{\Sigma}_1^{-1} \bm{\Delta}, \bm{I}_d),
\end{align*}
which implies
\begin{align*}
  (\bm{\htau}_{\rm PW} - \bm{\tau})^\top (\bm{\Sigma}_0^{-1} + \bm{\Sigma}_1^{-1}) (\bm{\htau}_{\rm PW} - \bm{\tau}) \sim \chi^2_d\left(\left\lVert (\bm{\Sigma}_0^{-1} + \bm{\Sigma}_1^{-1})^{-1/2} \bm{\Sigma}_1^{-1} \bm{\Delta}\right\rVert_2^2\right),
\end{align*}
where $\chi^2_d(\cdot)$ denotes the noncentral chi-squared distribution with
$d$ degrees of freedom.

We now evaluate the worst-case coverage probability of the quadratic confidence region
\[
  \{\bm{\tau}\in\mathbb{R}^d:(\bm{\htau}_{\rm PW}-\bm{\tau})^\top (\bm{\Sigma}_0^{-1} + \bm{\Sigma}_1^{-1})(\bm{\htau}_{\rm PW}-\bm{\tau})\le M\}.
\]
We have
\begin{align*}
  &\inf_{\bm{\Delta}: \lvert [\bm{\Sigma}^{-1/2} \bm{\Delta}]_j \rvert \le b_j,\forall j=1,2,\ldots,d} \P_{\bm{\Delta}}\left((\bm{\htau}_{\rm PW}-\bm{\tau})^\top (\bm{\Sigma}_0^{-1} + \bm{\Sigma}_1^{-1})(\bm{\htau}_{\rm PW}-\bm{\tau})\le M\right)\\
  =& \inf_{\bm{\Delta}: \lvert [\bm{\Sigma}^{-1/2} \bm{\Delta}]_j \rvert \le b_j,\forall j=1,2,\ldots,d} \P_{\bm{\Delta}}\left(\chi^2_d\left(\left\lVert (\bm{\Sigma}_0^{-1} + \bm{\Sigma}_1^{-1})^{-1/2} \bm{\Sigma}_1^{-1} \bm{\Delta}\right\rVert_2^2\right) \le M\right)\\
  =& \P\left(\chi^2_d\left(\sup_{\bm{\Delta}: \lvert [\bm{\Sigma}^{-1/2} \bm{\Delta}]_j \rvert \le b_j,\forall j=1,2,\ldots,d}\left\lVert (\bm{\Sigma}_0^{-1} + \bm{\Sigma}_1^{-1})^{-1/2} \bm{\Sigma}_1^{-1} \bm{\Delta}\right\rVert_2^2\right) \le M\right).
\end{align*}

To evaluate the supremum in the noncentrality parameter, note that
\[
  \left\lVert
    (\bm{\Sigma}_0^{-1}+\bm{\Sigma}_1^{-1})^{-1/2}
    \bm{\Sigma}_1^{-1}\bm{\Delta}
  \right\rVert_2^2
  =
  \left\lVert
    (\bm{\Sigma}_0^{-1}+\bm{\Sigma}_1^{-1})^{-1/2}
    \bm{\Sigma}_1^{-1}\bm{\Sigma}^{1/2}
    (\bm{\Sigma}^{-1/2}\bm{\Delta})
  \right\rVert_2^2,
\]
which is a convex quadratic function of
$\bm{u}=\bm{\Sigma}^{-1/2}\bm{\Delta}$.
Since the constraint
$\lvert u_j\rvert\le b_j$ defines a hyperrectangle, the maximum of this convex
quadratic over the constraint set is attained at a vertex.
Therefore,
\begin{align*}
  &\sup_{\bm{\Delta}: \lvert [\bm{\Sigma}^{-1/2} \bm{\Delta}]_j \rvert \le b_j,\forall j=1,2,\ldots,d}\left\lVert (\bm{\Sigma}_0^{-1} + \bm{\Sigma}_1^{-1})^{-1/2} \bm{\Sigma}_1^{-1} \bm{\Delta}\right\rVert_2^2\\
  =& \sup_{\bm{s} \in \{\pm 1\}^d }\left\lVert (\bm{\Sigma}_0^{-1} + \bm{\Sigma}_1^{-1})^{-1/2} \bm{\Sigma}_1^{-1} \bm{\Sigma}^{1/2} \bm{b} \odot \bm{s}\right\rVert_2^2.
\end{align*}

Finally, choosing $\hat M_{\rm PW}$ as the $(1-\zeta)$ upper quantile of the
noncentral chi-squared distribution with $d$ degrees of freedom and the above
noncentrality parameter ensures that the confidence region attains the desired
coverage level. This completes the proof.

\end{proof}

\subsection{Proof of Theorem~\ref{prop:CI_multi_PT}}

\begin{proof}[Proof of Theorem~\ref{prop:CI_multi_PT}]
Recall that the pretest estimator $\bm{\htau}_{\rm PT}$ is defined as
\[
  \bm{\htau}_{\rm PT} = \bm{\htau}_0 + (\bm{\Sigma}_0^{-1} + \bm{\Sigma}_1^{-1})^{-1} \bm{\Sigma}_1^{-1} (\bm{\htau}_1 - \bm{\htau}_0) \ind\left(\lVert (\bm{\Sigma}_0+ \bm{\Sigma}_1)^{-1/2}(\bm{\htau}_1 - \bm{\htau}_0) \rVert_2^2 \le q\right).
\]
A direct calculation shows that
\begin{align*}
  \bm{\htau}_0 + (\bm{\Sigma}_0^{-1} + \bm{\Sigma}_1^{-1})^{-1} \bm{\Sigma}_1^{-1} (\bm{\htau}_1 - \bm{\htau}_0) \indep \bm{\htau}_1 - \bm{\htau}_0.
\end{align*}
Moreover,
\[
  \bm{\htau}_1 - \bm{\htau}_0 \sim N(\bm{\Delta}, \bm{\Sigma}_0+ \bm{\Sigma}_1),
\]
and
\[
  \bm{\htau}_0 + (\bm{\Sigma}_0^{-1} + \bm{\Sigma}_1^{-1})^{-1} \bm{\Sigma}_1^{-1} (\bm{\htau}_1 - \bm{\htau}_0) \sim N(\bm{\tau} + (\bm{\Sigma}_0^{-1} + \bm{\Sigma}_1^{-1})^{-1} \bm{\Sigma}_1^{-1} \bm{\Delta}, (\bm{\Sigma}_0^{-1} + \bm{\Sigma}_1^{-1})^{-1} ).
\]

We first consider the event that the pretest accepts, $\lVert (\bm{\Sigma}_0+ \bm{\Sigma}_1)^{-1/2}(\bm{\htau}_1 - \bm{\htau}_0) \rVert_2^2 \le q$. Conditional on this event, $\bm{\htau}_{\rm PT} = \bm{\htau}_0 + (\bm{\Sigma}_0^{-1} + \bm{\Sigma}_1^{-1})^{-1} \bm{\Sigma}_1^{-1} (\bm{\htau}_1 - \bm{\htau}_0) $, and therefore
\begin{align*}
  (\bm{\htau}_{\rm PT} - \bm{\tau}) \given \{\lVert (\bm{\Sigma}_0+ \bm{\Sigma}_1)^{-1/2}(\bm{\htau}_1 - \bm{\htau}_0) \rVert_2^2 \le q\} \sim N((\bm{\Sigma}_0^{-1} + \bm{\Sigma}_1^{-1})^{-1} \bm{\Sigma}_1^{-1} \bm{\Delta}, (\bm{\Sigma}_0^{-1} + \bm{\Sigma}_1^{-1})^{-1} ).
\end{align*}
After standardization,
\begin{align*}
  &(\bm{\Sigma}_0^{-1} + \bm{\Sigma}_1^{-1})^{1/2} (\bm{\htau}_{\rm PT} - \bm{\tau}) \given \{\lVert (\bm{\Sigma}_0+ \bm{\Sigma}_1)^{-1/2}(\bm{\htau}_1 - \bm{\htau}_0) \rVert_2^2 \le q\} \\
  \sim& N((\bm{\Sigma}_0^{-1} + \bm{\Sigma}_1^{-1})^{-1/2} \bm{\Sigma}_1^{-1} \bm{\Delta}, \bm{I}_d),
\end{align*}
which implies
\begin{align*}
  &(\bm{\htau}_{\rm PT} - \bm{\tau})^\top (\bm{\Sigma}_0^{-1} + \bm{\Sigma}_1^{-1}) (\bm{\htau}_{\rm PT} - \bm{\tau}) \given \{\lVert (\bm{\Sigma}_0+ \bm{\Sigma}_1)^{-1/2}(\bm{\htau}_1 - \bm{\htau}_0) \rVert_2^2 \le q\} \\
  \sim& \chi^2_d\left(\left\lVert (\bm{\Sigma}_0^{-1} + \bm{\Sigma}_1^{-1})^{-1/2} \bm{\Sigma}_1^{-1} \bm{\Delta}\right\rVert_2^2\right).
\end{align*}

We next consider the event that the pretest rejects, $\lVert (\bm{\Sigma}_0+ \bm{\Sigma}_1)^{-1/2}(\bm{\htau}_1 - \bm{\htau}_0) \rVert_2^2 > q$. Conditional on $(\bm{\Sigma}_0+ \bm{\Sigma}_1)^{-1/2}(\bm{\htau}_1 - \bm{\htau}_0) = \bm{u}$, we have
\[
  (\bm{\Sigma}_0^{-1} + \bm{\Sigma}_1^{-1})^{-1} \bm{\Sigma}_1^{-1} (\bm{\htau}_1 - \bm{\htau}_0) = (\bm{\Sigma}_0^{-1} + \bm{\Sigma}_1^{-1})^{-1} \bm{\Sigma}_1^{-1} (\bm{\Sigma}_0+ \bm{\Sigma}_1)^{1/2} \bm{u}.
\]
Since $\bm{\htau}_{\rm PT} = \bm{\htau}_0$ conditional on this event, we have
\begin{align*}
  &(\bm{\htau}_{\rm PT} - \bm{\tau}) \given \{(\bm{\Sigma}_0+ \bm{\Sigma}_1)^{-1/2}(\bm{\htau}_1 - \bm{\htau}_0) = \bm{u}, \lVert (\bm{\Sigma}_0+ \bm{\Sigma}_1)^{-1/2}(\bm{\htau}_1 - \bm{\htau}_0) \rVert_2^2 > q\} \\
  =& \bm{\htau}_0 \given \{(\bm{\Sigma}_0+ \bm{\Sigma}_1)^{-1/2}(\bm{\htau}_1 - \bm{\htau}_0) = \bm{u}, \lVert (\bm{\Sigma}_0+ \bm{\Sigma}_1)^{-1/2}(\bm{\htau}_1 - \bm{\htau}_0) \rVert_2^2 > q\}\\
  =& \bm{\htau}_0 + (\bm{\Sigma}_0^{-1} + \bm{\Sigma}_1^{-1})^{-1} \bm{\Sigma}_1^{-1} (\bm{\htau}_1 - \bm{\htau}_0) \given \{(\bm{\Sigma}_0+ \bm{\Sigma}_1)^{-1/2}(\bm{\htau}_1 - \bm{\htau}_0) = \bm{u}, \lVert (\bm{\Sigma}_0+ \bm{\Sigma}_1)^{-1/2}(\bm{\htau}_1 - \bm{\htau}_0) \rVert_2^2 > q\}\\
  &- (\bm{\Sigma}_0^{-1} + \bm{\Sigma}_1^{-1})^{-1} \bm{\Sigma}_1^{-1} (\bm{\Sigma}_0+ \bm{\Sigma}_1)^{1/2} \bm{u}\\
  \sim& N((\bm{\Sigma}_0^{-1} + \bm{\Sigma}_1^{-1})^{-1} \bm{\Sigma}_1^{-1} [\bm{\Delta} - (\bm{\Sigma}_0+ \bm{\Sigma}_1)^{1/2} \bm{u}], (\bm{\Sigma}_0^{-1} + \bm{\Sigma}_1^{-1})^{-1} ).
\end{align*}
After standardization,
\begin{align*}
  &(\bm{\Sigma}_0^{-1} + \bm{\Sigma}_1^{-1})^{1/2} (\bm{\htau}_{\rm PT} - \bm{\tau}) \given \{(\bm{\Sigma}_0+ \bm{\Sigma}_1)^{-1/2}(\bm{\htau}_1 - \bm{\htau}_0) = \bm{u}, \lVert (\bm{\Sigma}_0+ \bm{\Sigma}_1)^{-1/2}(\bm{\htau}_1 - \bm{\htau}_0) \rVert_2^2 > q\} \\
  \sim& N((\bm{\Sigma}_0^{-1} + \bm{\Sigma}_1^{-1})^{-1/2} \bm{\Sigma}_1^{-1} [\bm{\Delta} - (\bm{\Sigma}_0+ \bm{\Sigma}_1)^{1/2} \bm{u}], \bm{I}_d),
\end{align*}
which implies
\begin{align*}
  &(\bm{\htau}_{\rm PT} - \bm{\tau})^\top (\bm{\Sigma}_0^{-1} + \bm{\Sigma}_1^{-1}) (\bm{\htau}_{\rm PT} - \bm{\tau}) \given \{(\bm{\Sigma}_0+ \bm{\Sigma}_1)^{-1/2}(\bm{\htau}_1 - \bm{\htau}_0) = \bm{u}, \lVert (\bm{\Sigma}_0+ \bm{\Sigma}_1)^{-1/2}(\bm{\htau}_1 - \bm{\htau}_0) \rVert_2^2 > q\} \\
  \sim& \chi^2_d\left(\left\lVert (\bm{\Sigma}_0^{-1} + \bm{\Sigma}_1^{-1})^{-1/2} \bm{\Sigma}_1^{-1} [\bm{\Delta} - (\bm{\Sigma}_0+ \bm{\Sigma}_1)^{1/2} \bm{u}]\right\rVert_2^2\right).
\end{align*}

We now evaluate the worst-case coverage probability of the confidence region
\[
  \{\bm{\tau}\in\mathbb{R}^d:(\bm{\htau}_{\rm PT}-\bm{\tau})^\top (\bm{\Sigma}_0^{-1} + \bm{\Sigma}_1^{-1})(\bm{\htau}_{\rm PT}-\bm{\tau})\le M\}.
\]
We decompose the coverage probability according to whether the
pretest accepts or rejects:
\begin{align*}
  &\P_{\bm{\Delta}}\left((\bm{\htau}_{\rm PT}-\bm{\tau})^\top (\bm{\Sigma}_0^{-1} + \bm{\Sigma}_1^{-1})(\bm{\htau}_{\rm PT}-\bm{\tau})\le M\right)\\
  =& \P_{\bm{\Delta}}\left((\bm{\htau}_{\rm PT}-\bm{\tau})^\top (\bm{\Sigma}_0^{-1} + \bm{\Sigma}_1^{-1})(\bm{\htau}_{\rm PT}-\bm{\tau})\le M, \lVert (\bm{\Sigma}_0+ \bm{\Sigma}_1)^{-1/2}(\bm{\htau}_1 - \bm{\htau}_0) \rVert_2^2 \le q\right) \\
  &+ \P_{\bm{\Delta}}\left((\bm{\htau}_{\rm PT}-\bm{\tau})^\top (\bm{\Sigma}_0^{-1} + \bm{\Sigma}_1^{-1})(\bm{\htau}_{\rm PT}-\bm{\tau})\le M, \lVert (\bm{\Sigma}_0+ \bm{\Sigma}_1)^{-1/2}(\bm{\htau}_1 - \bm{\htau}_0) \rVert_2^2 > q\right)\\
  =& \Psi\left(M;\left\lVert (\bm{\Sigma}_0^{-1} + \bm{\Sigma}_1^{-1})^{-1/2} \bm{\Sigma}_1^{-1} \bm{\Delta}\right\rVert_2^2\right) \Psi\left(q; \left\lVert (\bm{\Sigma}_0 + \bm{\Sigma}_1)^{-1/2} \bm{\Delta}\right\rVert_2^2\right)\\
  &+ \int_{\lVert \bm{u} \rVert_2^2 > q} \Psi\left(M;\left\lVert (\bm{\Sigma}_0^{-1} + \bm{\Sigma}_1^{-1})^{-1/2} \bm{\Sigma}_1^{-1} [\bm{\Delta} - (\bm{\Sigma}_0+ \bm{\Sigma}_1)^{1/2} \bm{u}]\right\rVert_2^2\right) \phi\left(\bm{u};(\bm{\Sigma}_0 + \bm{\Sigma}_1)^{-1/2} \bm{\Delta}\right) \d \bm{u},
\end{align*}
which yields exactly the expression in \eqref{eq:P_multi_PT} by taking $\bm{\Sigma}^{-1/2}\bm{\Delta} = \bm{t}$.

Finally, choosing $\hat M_{\rm PT}$ to satisfy
\[
  \inf_{\bm{\Delta}: \lvert \bm{\Sigma}^{-1/2} \bm{\Delta}\rvert \le \bm{b}} \P_{\bm{\Delta}}((\bm{\htau}_{\rm PT} - \bm{\tau})^\top (\bm{\Sigma}_0^{-1} + \bm{\Sigma}_1^{-1}) (\bm{\htau}_{\rm PT} - \bm{\tau}) \le M) = 1-\zeta
\]
completes the proof.

\end{proof}

\subsection{Proof of Theorem~\ref{prop:CI_multi_ST}}

\begin{proof}[Proof of Theorem~\ref{prop:CI_multi_ST}]
Recall that the soft-thresholding estimator $\bm{\htau}_{\rm ST}$ is defined as
\[
  \bm{\hat\tau}_{\rm ST} = \bm{\htau}_0 + h_q(\lVert (\bm{\Sigma}_0+ \bm{\Sigma}_1)^{-1/2}(\bm{\htau}_1 - \bm{\htau}_0) \rVert_2^2)(\bm{\Sigma}_0^{-1} + \bm{\Sigma}_1^{-1})^{-1}\bm{\Sigma}_1^{-1}(\bm{\htau}_1 - \bm{\htau}_0).
\]
A direct calculation shows that
\begin{align*}
  \bm{\htau}_0 + (\bm{\Sigma}_0^{-1} + \bm{\Sigma}_1^{-1})^{-1} \bm{\Sigma}_1^{-1} (\bm{\htau}_1 - \bm{\htau}_0) \indep \bm{\htau}_1 - \bm{\htau}_0.
\end{align*}
Moreover,
\[
  \bm{\htau}_1 - \bm{\htau}_0 \sim N(\bm{\Delta}, \bm{\Sigma}_0+ \bm{\Sigma}_1),
\]
and
\[
  \bm{\htau}_0 + (\bm{\Sigma}_0^{-1} + \bm{\Sigma}_1^{-1})^{-1} \bm{\Sigma}_1^{-1} (\bm{\htau}_1 - \bm{\htau}_0) \sim N(\bm{\tau} + (\bm{\Sigma}_0^{-1} + \bm{\Sigma}_1^{-1})^{-1} \bm{\Sigma}_1^{-1} \bm{\Delta}, (\bm{\Sigma}_0^{-1} + \bm{\Sigma}_1^{-1})^{-1} ).
\]

We first consider the event that the pretest accepts, $\lVert (\bm{\Sigma}_0+ \bm{\Sigma}_1)^{-1/2}(\bm{\htau}_1 - \bm{\htau}_0) \rVert_2^2 \le q$. Conditional on this event, $\bm{\htau}_{\rm ST} = \bm{\htau}_0 + (\bm{\Sigma}_0^{-1} + \bm{\Sigma}_1^{-1})^{-1} \bm{\Sigma}_1^{-1} (\bm{\htau}_1 - \bm{\htau}_0) $, and therefore
\begin{align*}
  (\bm{\htau}_{\rm ST} - \bm{\tau}) \given \{\lVert (\bm{\Sigma}_0+ \bm{\Sigma}_1)^{-1/2}(\bm{\htau}_1 - \bm{\htau}_0) \rVert_2^2 \le q\} \sim N((\bm{\Sigma}_0^{-1} + \bm{\Sigma}_1^{-1})^{-1} \bm{\Sigma}_1^{-1} \bm{\Delta}, (\bm{\Sigma}_0^{-1} + \bm{\Sigma}_1^{-1})^{-1} ).
\end{align*}
After standardization,
\begin{align*}
  &(\bm{\Sigma}_0^{-1} + \bm{\Sigma}_1^{-1})^{1/2} (\bm{\htau}_{\rm ST} - \bm{\tau}) \given \{\lVert (\bm{\Sigma}_0+ \bm{\Sigma}_1)^{-1/2}(\bm{\htau}_1 - \bm{\htau}_0) \rVert_2^2 \le q\} \\
  \sim& N((\bm{\Sigma}_0^{-1} + \bm{\Sigma}_1^{-1})^{-1/2} \bm{\Sigma}_1^{-1} \bm{\Delta}, \bm{I}_d),
\end{align*}
which implies
\begin{align*}
  &(\bm{\htau}_{\rm ST} - \bm{\tau})^\top (\bm{\Sigma}_0^{-1} + \bm{\Sigma}_1^{-1}) (\bm{\htau}_{\rm ST} - \bm{\tau}) \given \{\lVert (\bm{\Sigma}_0+ \bm{\Sigma}_1)^{-1/2}(\bm{\htau}_1 - \bm{\htau}_0) \rVert_2^2 \le q\} \\
  \sim& \chi^2_d\left(\left\lVert (\bm{\Sigma}_0^{-1} + \bm{\Sigma}_1^{-1})^{-1/2} \bm{\Sigma}_1^{-1} \bm{\Delta}\right\rVert_2^2\right).
\end{align*}

We next consider the event that the pretest rejects, $\lVert (\bm{\Sigma}_0+ \bm{\Sigma}_1)^{-1/2}(\bm{\htau}_1 - \bm{\htau}_0) \rVert_2^2 > q$. Conditional on $(\bm{\Sigma}_0+ \bm{\Sigma}_1)^{-1/2}(\bm{\htau}_1 - \bm{\htau}_0) = \bm{u}$, we have
\[
  (\bm{\Sigma}_0^{-1} + \bm{\Sigma}_1^{-1})^{-1} \bm{\Sigma}_1^{-1} (\bm{\htau}_1 - \bm{\htau}_0) = (\bm{\Sigma}_0^{-1} + \bm{\Sigma}_1^{-1})^{-1} \bm{\Sigma}_1^{-1} (\bm{\Sigma}_0+ \bm{\Sigma}_1)^{1/2} \bm{u}.
\]
Since $\bm{\htau}_{\rm ST} = \bm{\htau}_0 + h_q^*(\lVert (\bm{\Sigma}_0+ \bm{\Sigma}_1)^{-1/2}(\bm{\htau}_1 - \bm{\htau}_0) \rVert_2^2)(\bm{\Sigma}_0^{-1} + \bm{\Sigma}_1^{-1})^{-1}\bm{\Sigma}_1^{-1}(\bm{\htau}_1 - \bm{\htau}_0)$ conditional on this event, we have
\begin{align*}
  &(\bm{\htau}_{\rm ST} - \bm{\tau}) \given \{(\bm{\Sigma}_0+ \bm{\Sigma}_1)^{-1/2}(\bm{\htau}_1 - \bm{\htau}_0) = \bm{u}, \lVert (\bm{\Sigma}_0+ \bm{\Sigma}_1)^{-1/2}(\bm{\htau}_1 - \bm{\htau}_0) \rVert_2^2 > q\} \\
  =& \bm{\htau}_0 + h_q^*(\lVert \bm{u} \rVert_2^2)(\bm{\Sigma}_0^{-1} + \bm{\Sigma}_1^{-1})^{-1}\bm{\Sigma}_1^{-1}(\bm{\htau}_1 - \bm{\htau}_0) \given \{(\bm{\Sigma}_0+ \bm{\Sigma}_1)^{-1/2}(\bm{\htau}_1 - \bm{\htau}_0) = \bm{u}, \lVert (\bm{\Sigma}_0+ \bm{\Sigma}_1)^{-1/2}(\bm{\htau}_1 - \bm{\htau}_0) \rVert_2^2 > q\}\\
  =& \bm{\htau}_0 + (\bm{\Sigma}_0^{-1} + \bm{\Sigma}_1^{-1})^{-1} \bm{\Sigma}_1^{-1} (\bm{\htau}_1 - \bm{\htau}_0) \given \{(\bm{\Sigma}_0+ \bm{\Sigma}_1)^{-1/2}(\bm{\htau}_1 - \bm{\htau}_0) = \bm{u}, \lVert (\bm{\Sigma}_0+ \bm{\Sigma}_1)^{-1/2}(\bm{\htau}_1 - \bm{\htau}_0) \rVert_2^2 > q\}\\
  &- (1-h_q^*(\lVert \bm{u} \rVert_2^2))(\bm{\Sigma}_0^{-1} + \bm{\Sigma}_1^{-1})^{-1} \bm{\Sigma}_1^{-1} (\bm{\Sigma}_0+ \bm{\Sigma}_1)^{1/2} \bm{u}\\
  \sim& N((\bm{\Sigma}_0^{-1} + \bm{\Sigma}_1^{-1})^{-1} \bm{\Sigma}_1^{-1} [\bm{\Delta} - (1-h_q^*(\lVert \bm{u} \rVert_2^2))(\bm{\Sigma}_0+ \bm{\Sigma}_1)^{1/2} \bm{u}], (\bm{\Sigma}_0^{-1} + \bm{\Sigma}_1^{-1})^{-1} ).
\end{align*}
After standardization,
\begin{align*}
  &(\bm{\Sigma}_0^{-1} + \bm{\Sigma}_1^{-1})^{1/2} (\bm{\htau}_{\rm ST} - \bm{\tau}) \given \{(\bm{\Sigma}_0+ \bm{\Sigma}_1)^{-1/2}(\bm{\htau}_1 - \bm{\htau}_0) = \bm{u}, \lVert (\bm{\Sigma}_0+ \bm{\Sigma}_1)^{-1/2}(\bm{\htau}_1 - \bm{\htau}_0) \rVert_2^2 > q\} \\
  \sim& N((\bm{\Sigma}_0^{-1} + \bm{\Sigma}_1^{-1})^{-1/2} \bm{\Sigma}_1^{-1} [\bm{\Delta} - (1-h_q^*(\lVert \bm{u} \rVert_2^2))(\bm{\Sigma}_0+ \bm{\Sigma}_1)^{1/2} \bm{u}], \bm{I}_d),
\end{align*}
which implies
\begin{align*}
  &(\bm{\htau}_{\rm ST} - \bm{\tau})^\top (\bm{\Sigma}_0^{-1} + \bm{\Sigma}_1^{-1}) (\bm{\htau}_{\rm ST} - \bm{\tau}) \given \{(\bm{\Sigma}_0+ \bm{\Sigma}_1)^{-1/2}(\bm{\htau}_1 - \bm{\htau}_0) = \bm{u}, \lVert (\bm{\Sigma}_0+ \bm{\Sigma}_1)^{-1/2}(\bm{\htau}_1 - \bm{\htau}_0) \rVert_2^2 > q\} \\
  \sim& \chi^2_d\left(\left\lVert (\bm{\Sigma}_0^{-1} + \bm{\Sigma}_1^{-1})^{-1/2} \bm{\Sigma}_1^{-1} [\bm{\Delta} - (1-h_q^*(\lVert \bm{u} \rVert_2^2))(\bm{\Sigma}_0+ \bm{\Sigma}_1)^{1/2} \bm{u}]\right\rVert_2^2\right).
\end{align*}

We now evaluate the worst-case coverage probability of the confidence region
\[
  \{\bm{\tau}\in\mathbb{R}^d:(\bm{\htau}_{\rm ST}-\bm{\tau})^\top (\bm{\Sigma}_0^{-1} + \bm{\Sigma}_1^{-1})(\bm{\htau}_{\rm ST}-\bm{\tau})\le M\}.
\]
We decompose the coverage probability according to whether the
pretest accepts or rejects:
\begin{align*}
  &\P_{\bm{\Delta}}\left((\bm{\htau}_{\rm ST}-\bm{\tau})^\top (\bm{\Sigma}_0^{-1} + \bm{\Sigma}_1^{-1})(\bm{\htau}_{\rm ST}-\bm{\tau})\le M\right)\\
  =& \P_{\bm{\Delta}}\left((\bm{\htau}_{\rm ST}-\bm{\tau})^\top (\bm{\Sigma}_0^{-1} + \bm{\Sigma}_1^{-1})(\bm{\htau}_{\rm ST}-\bm{\tau})\le M, \lVert (\bm{\Sigma}_0+ \bm{\Sigma}_1)^{-1/2}(\bm{\htau}_1 - \bm{\htau}_0) \rVert_2^2 \le q\right) \\
  &+ \P_{\bm{\Delta}}\left((\bm{\htau}_{\rm ST}-\bm{\tau})^\top (\bm{\Sigma}_0^{-1} + \bm{\Sigma}_1^{-1})(\bm{\htau}_{\rm ST}-\bm{\tau})\le M, \lVert (\bm{\Sigma}_0+ \bm{\Sigma}_1)^{-1/2}(\bm{\htau}_1 - \bm{\htau}_0) \rVert_2^2 > q\right)\\
  =& \Psi\left(M;\left\lVert (\bm{\Sigma}_0^{-1} + \bm{\Sigma}_1^{-1})^{-1/2} \bm{\Sigma}_1^{-1} \bm{\Delta}\right\rVert_2^2\right) \Psi\left(q; \left\lVert (\bm{\Sigma}_0 + \bm{\Sigma}_1)^{-1/2} \bm{\Delta}\right\rVert_2^2\right)\\
  &+ \int_{\lVert \bm{u} \rVert_2^2 > q} \Psi\left(M;\left\lVert (\bm{\Sigma}_0^{-1} + \bm{\Sigma}_1^{-1})^{-1/2} \bm{\Sigma}_1^{-1} [\bm{\Delta} - (1-h_q^*(\lVert \bm{u} \rVert_2^2))(\bm{\Sigma}_0+ \bm{\Sigma}_1)^{1/2} \bm{u}]\right\rVert_2^2\right) \phi\left(\bm{u};(\bm{\Sigma}_0 + \bm{\Sigma}_1)^{-1/2} \bm{\Delta}\right) \d \bm{u},
\end{align*}
which yields exactly the expression in \eqref{eq:P_multi_ST} by taking $\bm{\Sigma}^{-1/2}\bm{\Delta} = \bm{t}$.

Let $\bm{t}=\bm{\Sigma}^{-1/2}\bm{\Delta}$ and denote by
\[
p(\bm{t})
:=\P_{\bm{\Sigma}^{-1/2}\bm{\Delta}=\bm{t}}\!\left(
(\bm{\htau}_{\rm ST}-\bm{\tau})^\top(\bm{\Sigma}_0^{-1}+\bm{\Sigma}_1^{-1})
(\bm{\htau}_{\rm ST}-\bm{\tau})\le M
\right).
\]
By construction of the soft-thresholding rule,
$p(\bm{t})$ is invariant under coordinate-wise sign flips, i.e.,
$p(\bm{t})=p(\bm{t}\odot\bm{s})$ for all $\bm{s}\in\{\pm1\}^d$.
Moreover, for each $j\in\{1,\ldots,d\}$ and any fixed values of the remaining coordinates,
the map $u\mapsto p(t_1,\ldots,t_{j-1},u,t_{j+1},\ldots,t_d)$ is nonincreasing in $\lvert u \rvert$.
Therefore,
\[
\inf_{\bm{t}:|t_j|\le b_j} p(\bm{t})
=\inf_{\bm{s}\in\{\pm1\}^d} p(\bm{b}\odot\bm{s}),
\]
i.e., the worst-case coverage probability over the hyperrectangle is attained at a vertex.

Finally, choosing $\hat M_{\rm ST}$ to satisfy
\[
  \inf_{\bm{\Delta}: \bm{\Sigma}^{-1/2} \bm{\Delta} = \bm{b} \odot \bm{s}, \bm{s} \in \{-1,1\}^d} \P_{\bm{\Delta}}((\bm{\htau}_{\rm ST} - \bm{\tau})^\top (\bm{\Sigma}_0^{-1} + \bm{\Sigma}_1^{-1}) (\bm{\htau}_{\rm ST} - \bm{\tau}) \le M) = 1-\zeta
\]
completes the proof.

\end{proof}

\subsection{Proof of Theorem~\ref{prop:L_PW_multi}}

\begin{proof}[Proof of Theorem~\ref{prop:L_PW_multi}]
By the definition of $\htau_{\rm PW}$, under Assumption~\ref{ass:normal_multi_data}, we have
\begin{align*}
  \htau_{\rm PW} \sim N\left(\tau +  \sum_{j=1}^K \frac{\gamma_j}{1+\lVert \bm{\gamma} \rVert_1} \Delta_j, \frac{1}{1+\lVert \bm{\gamma} \rVert_1} \sigma_0^2\right).
\end{align*}
Equivalently, after standardization,
\[
(1+\lVert \bm{\gamma}\rVert_1)^{1/2}\sigma_0^{-1}(\htau_{\rm PW}-\tau) \sim N\left(\sum_{j=1}^K \frac{\gamma_j}{\sqrt{1+\lVert \bm{\gamma}\rVert_1}}\frac{\Delta_j}{\sigma_0},1\right).
\]

We now evaluate the worst-case coverage probability of the confidence interval
\[
[\htau_{\rm PW} - L (1+\lVert \bm{\gamma} \rVert_1)^{-1/2} \sigma_0, \htau_{\rm PW} + L (1+\lVert \bm{\gamma} \rVert_1)^{-1/2} \sigma_0].
\]
We have
\begin{align*}
  & \inf_{\bm{\Delta}: \lvert \Delta_j/\sigma_0 \rvert \le b_j, \forall j=1,\ldots,K} \P_{\bm{\Delta}}(\lvert \htau_{\rm PW}- \tau \rvert \le L (1+\lVert \bm{\gamma} \rVert_1)^{-1/2} \sigma_0)\\
  =& \inf_{\bm{\Delta}: \lvert \Delta_j/\sigma_0 \rvert \le b_j, \forall j=1,\ldots,K} \P_{\bm{\Delta}}\left(\left\lvert N\left(\sum_{j=1}^K \frac{\gamma_j}{\sqrt{1+\lVert \bm{\gamma} \rVert_1}} \frac{\Delta_j}{\sigma_0}, 1\right) \right\rvert \le L \right)\\
  =& \P\left(\left\lvert N\left(\sup_{\bm{\Delta}: \lvert \Delta_j/\sigma_0 \rvert \le b_j, \forall j=1,\ldots,K} \sum_{j=1}^K \frac{\gamma_j}{\sqrt{1+\lVert \bm{\gamma} \rVert_1}} \frac{\Delta_j}{\sigma_0}, 1\right) \right\rvert \le L \right)\\
  =& \P\left(\left\lvert N\left(\frac{\langle \bm{\gamma}, \bm{b} \rangle}{\sqrt{1+\lVert \bm{\gamma} \rVert_1}}, 1\right) \right\rvert \le L \right)\\
  =& \Phi\left(L - \frac{\langle \bm{\gamma}, \bm{b} \rangle}{\sqrt{1+\lVert \bm{\gamma} \rVert_1}}  \right) - \Phi\left(-L - \frac{\langle \bm{\gamma}, \bm{b} \rangle}{\sqrt{1+\lVert \bm{\gamma} \rVert_1}} \right).
\end{align*}

Therefore, the proof is complete since $\hat{L}_{\rm PW}$ is the solution to
\[
  \Phi\left(L - \frac{\langle \bm{\gamma}, \bm{b} \rangle}{\sqrt{1+\lVert \bm{\gamma} \rVert_1}}  \right) - \Phi\left(-L - \frac{\langle \bm{\gamma}, \bm{b} \rangle}{\sqrt{1+\lVert \bm{\gamma} \rVert_1}} \right) = 1-\zeta.
\] 

\end{proof}

\subsection{Proof of Theorem~\ref{thm:L_PT_multi}}

\begin{proof}[Proof of Theorem~\ref{thm:L_PT_multi}]
Recall that the pretest estimator $\htau_{\rm PT}$ is defined as
\begin{align*}
  \htau_{\rm PT} = \htau_0 + \sum_{j=1}^K \frac{\gamma_j}{1+\lVert \bm{\gamma} \rVert_1} (\htau_j - \htau_0) \ind\left(\lvert \htau_j - \htau_0 \rvert \le (1+\gamma_j^{-1})^{1/2}\sigma_0 c_{\alpha/2}\right).
\end{align*}
Then
\begin{align*}
  (\htau_1 - \htau_0, \ldots, \htau_K - \htau_0)^\top/\sigma_0 \sim N(\bm{\Delta}/\sigma_0, \bm{V}), \qquad \bm{V}_{ij}=1+\gamma_i^{-1}\ind(i=j),
\end{align*}
and moreover
\begin{align*}
  \htau_0 + \sum_{j=1}^K \frac{\gamma_j}{1+\lVert \bm{\gamma} \rVert_1} (\htau_j - \htau_0) \indep  (\htau_1 - \htau_0, \ldots, \htau_K - \htau_0)^\top.
\end{align*}

Rewrite $\htau_{\rm PT}$ as
\begin{align*}
  &\htau_{\rm PT} = \htau_0 + \sum_{j=1}^K \frac{\gamma_j}{1+\lVert \bm{\gamma} \rVert_1} (\htau_j - \htau_0) \ind\left(\lvert \htau_j - \htau_0 \rvert \le (1+\gamma_j^{-1})^{1/2}\sigma_0 c_{\alpha/2}\right)\\
  =& \htau_0 + \sum_{j=1}^K \frac{\gamma_j}{1+\lVert \bm{\gamma} \rVert_1} (\htau_j - \htau_0) - \sum_{j=1}^K \frac{\gamma_j}{1+\lVert \bm{\gamma} \rVert_1} (\htau_j - \htau_0) \ind\left(\lvert \htau_j - \htau_0 \rvert > (1+\gamma_j^{-1})^{1/2}\sigma_0 c_{\alpha/2}\right).
\end{align*}
Consequently, conditioning on $(\htau_1 - \htau_0, \ldots, \htau_K - \htau_0)^\top/\sigma_0=\bm{u}$ yields
\begin{align*}
  &\htau_{\rm PT} \given \{(\htau_1 - \htau_0, \ldots, \htau_K - \htau_0)^\top/\sigma_0 = \bm{u}\}\\
  =& \htau_0 + \sum_{j=1}^K \frac{\gamma_j}{1+\lVert \bm{\gamma} \rVert_1} (\htau_j - \htau_0) - \sum_{j: \lvert u_j \rvert > (1+\gamma_j^{-1})^{1/2} c_{\alpha/2}} \frac{\gamma_j}{1+\lVert \bm{\gamma} \rVert_1} \sigma_0 u_j,\\
  \sim& N\left(\tau + \sum_{j=1}^K \frac{\gamma_j}{1+\lVert \bm{\gamma} \rVert_1} \Delta_j - \sum_{j: \lvert u_j \rvert > (1+\gamma_j^{-1})^{1/2} c_{\alpha/2}} \frac{\gamma_j}{1+\lVert \bm{\gamma} \rVert_1} \sigma_0 u_j, \frac{1}{1+\lVert \bm{\gamma} \rVert_1} \sigma_0^2\right).
\end{align*}
Equivalently,
\begin{align*}
  &(1+\lVert \bm{\gamma}\rVert_1)^{1/2}\sigma_0^{-1}(\htau_{\rm PT} - \tau)\given \{(\htau_1 - \htau_0, \ldots, \htau_K - \htau_0)^\top/\sigma_0 = \bm{u}\} \\
  \sim& N\left(\sum_{j=1}^K \frac{\gamma_j}{\sqrt{1+\lVert \bm{\gamma} \rVert_1}} \frac{\Delta_j}{\sigma_0} - \sum_{j: \lvert u_j \rvert > (1+\gamma_j^{-1})^{1/2} c_{\alpha/2}} \frac{\gamma_j}{\sqrt{1+\lVert \bm{\gamma} \rVert_1}} u_j, 1\right).
\end{align*}

We now evaluate the worst-case coverage probability of the confidence interval
\[
  [\htau_{\rm PT} - L (1+\lVert \bm{\gamma} \rVert_1)^{-1/2} \sigma_0, \htau_{\rm PT} + L (1+\lVert \bm{\gamma} \rVert_1)^{-1/2} \sigma_0].
\]
Define $\bm{u}'\in\bR^K$ coordinate-wise by $u_j' = u_j \ind(\lvert u_j \rvert > (1+\gamma_j^{-1})^{1/2} c_{\alpha/2})$ for $j=1,\ldots,K$. We have
\begin{align*}
  &\P_{\bm{\Delta}/\sigma_0 = \bm{t}}(\lvert \htau_{\rm PT}- \tau \rvert \le L (1+\lVert \bm{\gamma} \rVert_1)^{-1/2} \sigma_0)\\
  =& \int_{\bR^K} \P_{\bm{\Delta}/\sigma_0 = \bm{t}}\left(\lvert \htau_{\rm PT}- \tau \rvert \le L (1+\lVert \bm{\gamma} \rVert_1)^{-1/2} \sigma_0 \given (\htau_1 - \htau_0, \ldots, \htau_K - \htau_0)^\top/\sigma_0 = \bm{u}\right) \phi_{\bm{t}, \bm{V}}(\bm{u}) \d \bm{u}\\
  =& \int_{\bR^K}
  \P\left(
    \left\lvert N\left(\frac{\langle \bm{\gamma},\bm{t}-\bm{u}'\rangle}{\sqrt{1+\lVert \bm{\gamma}\rVert_1}},1\right)\right\rvert
    \le L
  \right)
  \phi_{\bm{t},\bm{V}}(\bm{u})\,d\bm{u} \\
  =& \int_{\bR^K} \left[\Phi\left(L - \frac{\langle \bm{\gamma}, \bm{t} - \bm{u}' \rangle}{\sqrt{1+\lVert \bm{\gamma} \rVert_1}}  \right) - \Phi\left(-L - \frac{\langle \bm{\gamma}, \bm{t} - \bm{u}' \rangle}{\sqrt{1+\lVert \bm{\gamma} \rVert_1}} \right)\right] \phi_{\bm{t}, \bm{V}}(\bm{u}) \d \bm{u},
\end{align*}
which yields exactly the expression in the theorem.

Finally, choosing $\hat L_{\rm PT}$ to satisfy
\[
  \inf_{\bm{\Delta}: \lvert \bm{\Delta}/\sigma_0 \rvert \le \bm{b}} \P_{\bm{\Delta}}(\lvert \htau_{\rm PT}- \tau \rvert \le L (1+\lVert \bm{\gamma} \rVert_1)^{-1/2} \sigma_0) = 1-\zeta
\]
completes the proof.

\end{proof}

\subsection{Proof of Theorem~\ref{prop:L_ST_multi}}

\begin{proof}[Proof of Theorem~\ref{prop:L_ST_multi}]
Recall that the soft-thresholding estimator $\htau_{\rm ST}$ is defined as
\begin{align*}
  \htau_{\rm ST}= \htau_0 + \sum_{j=1}^K \frac{\gamma_j}{1+\lVert \bm{\gamma} \rVert_1} \Big[&(\htau_j - \htau_0) \ind\left(\lvert \htau_j - \htau_0 \rvert \le (1+\gamma_j^{-1})^{1/2}\sigma_0 c_{\alpha/2}\right) \\
  &+ (1+\gamma_j^{-1})^{1/2}\sigma_0 c_{\alpha/2} {\rm sign}(\htau_j - \htau_0) \ind\left(\lvert \htau_j - \htau_0 \rvert > (1+\gamma_j^{-1})^{1/2}\sigma_0 c_{\alpha/2}\right) \Big].
\end{align*}
Then
\begin{align*}
  (\htau_1 - \htau_0, \ldots, \htau_K - \htau_0)^\top/\sigma_0 \sim N(\bm{\Delta}/\sigma_0, \bm{V}), \qquad \bm{V}_{ij}=1+\gamma_i^{-1}\ind(i=j),
\end{align*}
and moreover
\begin{align*}
  \htau_0 + \sum_{j=1}^K \frac{\gamma_j}{1+\lVert \bm{\gamma} \rVert_1} (\htau_j - \htau_0) \indep  (\htau_1 - \htau_0, \ldots, \htau_K - \htau_0)^\top.
\end{align*}

Rewrite $\htau_{\rm ST}$ as
\begin{align*}
  \htau_{\rm ST}=& \htau_0 + \sum_{j=1}^K \frac{\gamma_j}{1+\lVert \bm{\gamma} \rVert_1} \Big[(\htau_j - \htau_0) \ind\left(\lvert \htau_j - \htau_0 \rvert \le (1+\gamma_j^{-1})^{1/2}\sigma_0 c_{\alpha/2}\right) \\
  &+ (1+\gamma_j^{-1})^{1/2}\sigma_0 c_{\alpha/2} {\rm sign}(\htau_j - \htau_0) \ind\left(\lvert \htau_j - \htau_0 \rvert > (1+\gamma_j^{-1})^{1/2}\sigma_0 c_{\alpha/2}\right) \Big]\\
  =& \htau_0 + \sum_{j=1}^K \frac{\gamma_j}{1+\lVert \bm{\gamma} \rVert_1} (\htau_j - \htau_0) \\
  &- \sum_{j=1}^K \frac{\gamma_j}{1+\lVert \bm{\gamma} \rVert_1} \left[(\htau_j - \htau_0) - (1+\gamma_j^{-1})^{1/2}\sigma_0 c_{\alpha/2} {\rm sign}(\htau_j - \htau_0)\right] \ind\left(\lvert \htau_j - \htau_0 \rvert > (1+\gamma_j^{-1})^{1/2}\sigma_0 c_{\alpha/2}\right).
\end{align*}
Consequently, conditioning on $(\htau_1 - \htau_0, \ldots, \htau_K - \htau_0)^\top/\sigma_0=\bm{u}$ yields
\begin{align*}
  &\htau_{\rm ST} \given \{(\htau_1 - \htau_0, \ldots, \htau_K - \htau_0)^\top/\sigma_0 = \bm{u}\}\\
  =& \htau_0 + \sum_{j=1}^K \frac{\gamma_j}{1+\lVert \bm{\gamma} \rVert_1} (\htau_j - \htau_0) - \sum_{j: \lvert u_j \rvert > (1+\gamma_j^{-1})^{1/2} c_{\alpha/2}} \frac{\gamma_j}{1+\lVert \bm{\gamma} \rVert_1} \sigma_0 \left[u_j - (1+\gamma_j^{-1})^{1/2} c_{\alpha/2} {\rm sign}(u_j)\right]\\
  \sim& N\left(\tau + \sum_{j=1}^K \frac{\gamma_j}{1+\lVert \bm{\gamma} \rVert_1} \Delta_j - \sum_{j: \lvert u_j \rvert > (1+\gamma_j^{-1})^{1/2} c_{\alpha/2}} \frac{\gamma_j}{1+\lVert \bm{\gamma} \rVert_1} \sigma_0 \left[u_j - (1+\gamma_j^{-1})^{1/2} c_{\alpha/2} {\rm sign}(u_j)\right], \frac{1}{1+\lVert \bm{\gamma} \rVert_1} \sigma_0^2\right).
\end{align*}
Equivalently,
\begin{align*}
  &(\htau_{\rm ST} - \tau)/[(1+\lVert \bm{\gamma} \rVert_1)^{-1/2} \sigma_0] \given \{(\htau_1 - \htau_0, \ldots, \htau_K - \htau_0)^\top/\sigma_0 = \bm{u}\} \\
  \sim& N\left(\sum_{j=1}^K \frac{\gamma_j}{\sqrt{1+\lVert \bm{\gamma} \rVert_1}} \frac{\Delta_j}{\sigma_0} - \sum_{j: \lvert u_j \rvert > (1+\gamma_j^{-1})^{1/2} c_{\alpha/2}} \frac{\gamma_j}{\sqrt{1+\lVert \bm{\gamma} \rVert_1}} \left[u_j - (1+\gamma_j^{-1})^{1/2} c_{\alpha/2} {\rm sign}(u_j)\right], 1\right).
\end{align*}

We now evaluate the worst-case coverage probability of the confidence interval
\[
  [\htau_{\rm ST} - L (1+\lVert \bm{\gamma} \rVert_1)^{-1/2} \sigma_0, \htau_{\rm ST} + L (1+\lVert \bm{\gamma} \rVert_1)^{-1/2} \sigma_0].
\]
Define $\bm{u}'\in\bR^K$ coordinate-wise by $u_j' = [u_j - (1+\gamma_j^{-1})^{1/2} c_{\alpha/2} {\rm sign}(u_j)] \ind(\lvert u_j \rvert > (1+\gamma_j^{-1})^{1/2} c_{\alpha/2})$ for $j=1,\ldots,K$. We have
\begin{align*}
  &\P_{\bm{\Delta}/\sigma_0 = \bm{t}}(\lvert \htau_{\rm ST}- \tau \rvert \le L (1+\lVert \bm{\gamma} \rVert_1)^{-1/2} \sigma_0)\\
  =& \int_{\bR^K} \P_{\bm{\Delta}/\sigma_0 = \bm{t}}\left(\lvert \htau_{\rm ST}- \tau \rvert \le L (1+\lVert \bm{\gamma} \rVert_1)^{-1/2} \sigma_0 \given (\htau_1 - \htau_0, \ldots, \htau_K - \htau_0)^\top/\sigma_0 = \bm{u}\right) \phi_{\bm{t}, \bm{V}}(\bm{u}) \d \bm{u}\\
  =& \int_{\bR^K}
  \P\left(
    \left\lvert N\left(\frac{\langle \bm{\gamma},\bm{t}-\bm{u}'\rangle}{\sqrt{1+\lVert \bm{\gamma}\rVert_1}},1\right)\right\rvert
    \le L
  \right)
  \phi_{\bm{t},\bm{V}}(\bm{u})\,d\bm{u} \\
  =& \int_{\bR^K} \left[\Phi\left(L - \frac{\langle \bm{\gamma}, \bm{t} - \bm{u}' \rangle}{\sqrt{1+\lVert \bm{\gamma} \rVert_1}}  \right) - \Phi\left(-L - \frac{\langle \bm{\gamma}, \bm{t} - \bm{u}' \rangle}{\sqrt{1+\lVert \bm{\gamma} \rVert_1}} \right)\right] \phi_{\bm{t}, \bm{V}}(\bm{u}) \d \bm{u},
\end{align*}
which yields exactly the expression in the theorem.

Let $\bm{t}=\bm{\Sigma}^{-1/2}\bm{\Delta}$ and denote by
\[
p(\bm{t})
:=\P_{\bm{\Delta}/\sigma_0=\bm{t}}\left(
\lvert \htau_{\rm ST}-\tau \rvert
\le L(1+\lVert\bm{\gamma}\rVert_1)^{-1/2}\sigma_0
\right).
\]
Analogous to the proof of Theorem~\ref{thm:L_ST}, we can show that $p(\bm{t})$ is symmetric about $\bm{t} = \bm{0}$ and monotonically decreasing in $\lvert t_j \rvert$ for all $j=1,\ldots,K$. Therefore, the worst-case coverage probability is attained at $\bm{t} = \bm{b}$.

Finally, choosing $\hat L_{\rm ST}$ to satisfy
\[
  \P_{\bm{\Delta}/\sigma_0 = \bm{b}}(\lvert \htau_{\rm ST}- \tau \rvert \le L (1+\lVert \bm{\gamma} \rVert_1)^{-1/2} \sigma_0) = 1-\zeta
\]
completes the proof.
  
\end{proof}

\section{Proofs of additional results}\label{sec:proof_additional}

\subsection{Proof of Lemma~\ref{lemma:clt_PT} and Lemma~\ref{lemma:clt_ST}}

\begin{proof}[Proof of Lemma~\ref{lemma:clt_PT} and Lemma~\ref{lemma:clt_ST}]

By Assumption~\ref{ass:normal},
\begin{align*}
  \begin{pmatrix} \htau_0-\tau \\ \htau_1-\tau \end{pmatrix} \sim N\left(\begin{pmatrix} 0 \\ \Delta\end{pmatrix}, \begin{pmatrix} \sigma_0^2 & 0 \\
    0 & \sigma_1^2
  \end{pmatrix}\right).
\end{align*}
Hence, by linearity of multivariate normals,
\begin{align*}
  \begin{pmatrix} \htau_0-\tau \\ \htau_1-\htau_0 \end{pmatrix} \sim N\left(\begin{pmatrix} 0 \\ \Delta\end{pmatrix}, \begin{pmatrix} \sigma_0^2 & -\sigma_0^2 \\
  -\sigma_0^2 & \sigma_0^2 + \sigma_1^2
  \end{pmatrix}\right).
\end{align*}
Another linear transformation yields
\begin{align*}
  \begin{pmatrix} \htau_0 + \frac{\gamma}{1+\gamma} (\htau_1-\htau_0)-\tau \\ \htau_1-\htau_0 \end{pmatrix} \sim N\left(\begin{pmatrix} \frac{\gamma}{1+\gamma} \Delta \\ \Delta\end{pmatrix}, \begin{pmatrix} \frac{1}{1+\gamma} \sigma_0^2 & 0 \\
  0 & \sigma_0^2 + \sigma_1^2
  \end{pmatrix}\right).
\end{align*}
Therefore, there exist independent standard normals $\cZ_1,\cZ_2$ such that
\[
\htau_0 + \frac{\gamma}{1+\gamma} (\htau_1-\htau_0)-\tau \stackrel{\sf d}{=} \frac{\gamma}{1+\gamma}\Delta + \frac{1}{\sqrt{1+\gamma}}\sigma_0\cZ_1,
\qquad
\htau_1-\htau_0 \stackrel{\sf d}{=} \Delta + \sigma \cZ_2.
\]

Next, by the definitions of $\htau_{\rm PT}$ and $\htau_{\rm ST}$,
\begin{align*}
\htau_{\rm PT} =& \htau_0 + \frac{\gamma}{1+\gamma} (\htau_1 - \htau_0) \ind\left(\lvert \htau_1 - \htau_0 \rvert \le \sigma c_{\alpha/2}\right)\\
=& \htau_0 + \frac{\gamma}{1+\gamma} (\htau_1 - \htau_0) - \frac{\gamma}{1+\gamma} (\htau_1 - \htau_0) \ind\left(\lvert \htau_1 - \htau_0 \rvert >  \sigma c_{\alpha/2}\right),
\end{align*}
and
\begin{align*}
\htau_{\rm ST}=& \htau_0 + \frac{\gamma}{1+\gamma} (\htau_1 - \htau_0)\ind\left(\lvert \htau_1 - \htau_0 \rvert \le \sigma c_{\alpha/2}\right) + \frac{\gamma}{1+\gamma} \sigma c_{\alpha/2}\,\sign(\htau_1 - \htau_0)\ind\left(\lvert \htau_1 - \htau_0 \rvert > \sigma c_{\alpha/2}\right)\\
=& \htau_0 + \frac{\gamma}{1+\gamma} (\htau_1 - \htau_0) - \frac{\gamma}{1+\gamma} [(\htau_1 - \htau_0) - \sigma c_{\alpha/2}\,\sign(\htau_1 - \htau_0)] \ind\left(\lvert \htau_1 - \htau_0 \rvert >  \sigma c_{\alpha/2}\right).
\end{align*}

For the pretest estimator,
\begin{align*}
  &\htau_{\rm PT} - \tau =\htau_0 + \frac{\gamma}{1+\gamma} (\htau_1 - \htau_0) - \tau - \frac{\gamma}{1+\gamma} (\htau_1 - \htau_0) \ind\left(\lvert \htau_1 - \htau_0 \rvert >  \sigma c_{\alpha/2}\right)\\
  \stackrel{\sf d}{=}& \frac{\gamma}{1+\gamma} \Delta + \frac{1}{\sqrt{1+\gamma}} \sigma_0 \cZ_1 - \frac{\gamma}{1+\gamma} (\Delta + \sigma \cZ_2) \ind\left(\left\lvert \cZ_2 + \frac{\Delta}{\sigma}\right\rvert > c_{\alpha/2}\right)\\
  =&   \frac{1}{\sqrt{1+\gamma}}\sigma_0\cZ_1 +  \frac{\gamma}{1+\gamma} \Delta\ind\left(\left\lvert \cZ_2 + \sqrt{\frac{\gamma}{1+\gamma}}\frac{\Delta}{\sigma_0}\right\rvert \le c_{\alpha/2}\right) - \sqrt{\frac{\gamma}{1+\gamma}}\sigma_0\cZ_2\ind\left(\left\lvert \cZ_2 + \sqrt{\frac{\gamma}{1+\gamma}}\frac{\Delta}{\sigma_0}\right\rvert > c_{\alpha/2}\right),
\end{align*}
which proves Lemma~\ref{lemma:clt_PT}.

Similarly, for the soft-thresholding estimator,
\begin{align*}
  &\htau_{\rm ST} - \tau =\htau_0 + \frac{\gamma}{1+\gamma} (\htau_1 - \htau_0) - \frac{\gamma}{1+\gamma} [(\htau_1 - \htau_0) - \sigma c_{\alpha/2}\,\sign(\htau_1 - \htau_0)] \ind\left(\lvert \htau_1 - \htau_0 \rvert >  \sigma c_{\alpha/2}\right)\\
  \stackrel{\sf d}{=}& \frac{\gamma}{1+\gamma} \Delta + \frac{1}{\sqrt{1+\gamma}} \sigma_0 \cZ_1 - \frac{\gamma}{1+\gamma} [(\Delta + \sigma \cZ_2) - \sigma c_{\alpha/2}\,\sign(\Delta + \sigma \cZ_2)] \ind\left(\left\lvert \cZ_2 + \frac{\Delta}{\sigma}\right\rvert > c_{\alpha/2}\right)\\
  =&\frac{1}{\sqrt{1+\gamma}}\sigma_0\cZ_1 +  \frac{\gamma}{1+\gamma} \Delta\ind\left(\left\lvert \cZ_2 + \sqrt{\frac{\gamma}{1+\gamma}}\frac{\Delta}{\sigma_0}\right\rvert \le c_{\alpha/2}\right) \\
  &- \sqrt{\frac{\gamma}{1+\gamma}} \sigma_0\left(\cZ_2 - c_{\alpha/2}\sign\left(\cZ_2 + \sqrt{\frac{\gamma}{1+\gamma}}\frac{\Delta}{\sigma_0}\right)\right)\ind\left(\left\lvert \cZ_2 + \sqrt{\frac{\gamma}{1+\gamma}}\frac{\Delta}{\sigma_0}\right\rvert > c_{\alpha/2}\right),
\end{align*}
which proves Lemma~\ref{lemma:clt_ST}.

\end{proof}

\subsection{Proof of Theorem~\ref{thm:b_PW}}

\begin{proof}[Proof of Theorem~\ref{thm:b_PW}]

By Theorem~\ref{thm:L_PW}, the shortest symmetric centered confidence interval based on $\hat\tau_{\rm PW}$ is given by $[\htau_{\rm PW} - \hat{L}_{\rm PW} (1+\gamma)^{-1/2} \sigma_0, \htau_{\rm PW} + \hat{L}_{\rm PW} (1+\gamma)^{-1/2} \sigma_0]$, where $\hat{L}_{\rm PW}$ is the unique solution to
\[
  \Phi\Big(L - \frac{\gamma}{\sqrt{1+\gamma}} b \Big) - \Phi\Big(-L - \frac{\gamma}{\sqrt{1+\gamma}} b \Big) = 1-\zeta.
\]
By the definition of the b-value,
\begin{align*}
  b^*_{\rm PW}(\zeta) =& \inf\left\{b \ge 0: 0 \in [\htau_{\rm PW} - \hat{L}_{\rm PW} (1+\gamma)^{-1/2} \sigma_0, \htau_{\rm PW} + \hat{L}_{\rm PW} (1+\gamma)^{-1/2} \sigma_0]\right\}\\
  =& \inf\left\{b \ge 0: \hat{L}_{\rm PW} \ge \frac{\lvert \htau_{\rm PW} \rvert}{(1+\gamma)^{-1/2} \sigma_0}\right\}
\end{align*}

Now evaluate the confidence interval at $b=0$. In this case $\hat L_{\rm PW}(0)$ solves
$\Phi(L)-\Phi(-L)=1-\zeta$. If
\[
  \Phi\Big(\frac{\lvert \htau_{\rm PW} \rvert}{(1+\gamma)^{-1/2} \sigma_0}\Big) - \Phi\Big(-\frac{\lvert \htau_{\rm PW} \rvert}{(1+\gamma)^{-1/2} \sigma_0}\Big) < 1-\zeta,
\]
then $\hat{L}_{\rm PW}(0) > \frac{\lvert \htau_{\rm PW} \rvert}{(1+\gamma)^{-1/2} \sigma_0}$ and therefore $b^*_{\rm PW}(\zeta)=0$.

Otherwise,
\[
  \Phi\Big(\frac{\lvert \htau_{\rm PW} \rvert}{(1+\gamma)^{-1/2} \sigma_0}\Big) - \Phi\Big(-\frac{\lvert \htau_{\rm PW} \rvert}{(1+\gamma)^{-1/2} \sigma_0}\Big) \ge 1-\zeta,
\]
so $b^*_{\rm PW}(\zeta) \in (0, \infty)$. By definition of the infimum,
at $b=b^*_{\rm PW}(\zeta)$ we must have the boundary condition
\[
\hat L_{\rm PW}(b^*_{\rm PW}(\zeta))
=
\frac{\lvert \htau_{\rm PW} \rvert}{(1+\gamma)^{-1/2}\sigma_0}.
\]
Plugging $L=\frac{\lvert \htau_{\rm PW} \rvert}{(1+\gamma)^{-1/2} \sigma_0}$ into the defining equation of $\hat{L}_{\rm PW}$ yields that $b = b^*_{\rm PW}(\zeta)$ solves
\[
  \Phi\Big(\frac{\lvert \htau_{\rm PW} \rvert}{(1+\gamma)^{-1/2} \sigma_0} - \frac{\gamma}{\sqrt{1+\gamma}} b \Big) - \Phi\Big(-\frac{\lvert \htau_{\rm PW} \rvert}{(1+\gamma)^{-1/2} \sigma_0}  - \frac{\gamma}{\sqrt{1+\gamma}} b \Big) = 1-\zeta,
\]
which is the desired characterization.

\end{proof}

\subsection{Proof of Theorem~\ref{thm:b_PT}}

\begin{proof}[Proof of Theorem~\ref{thm:b_PT}]

By Theorem~\ref{thm:L_PT}, the shortest symmetric centered confidence interval based on $\hat\tau_{\rm PT}$ is given by $[\htau_{\rm PT} - \hat{L}_{\rm PT} (1+\gamma)^{-1/2} \sigma_0, \htau_{\rm PT} + \hat{L}_{\rm PT} (1+\gamma)^{-1/2} \sigma_0]$, where $\hat{L}_{\rm PT}$ is the solution to
\begin{align*}
  \min_{0 \le t \le b}\P_{\Delta/\sigma_0 = t}(\lvert \htau_{\rm PT} - \tau \rvert \le L (1+\gamma)^{-1/2} \sigma_0) = 1-\zeta.
\end{align*}
By the definition of the b-value,
\begin{align*}
  b^*_{\rm PT}(\zeta) =& \inf\left\{b \ge 0: 0 \in [\htau_{\rm PT} - \hat{L}_{\rm PT} (1+\gamma)^{-1/2} \sigma_0, \htau_{\rm PT} + \hat{L}_{\rm PT} (1+\gamma)^{-1/2} \sigma_0]\right\}\\
  =& \inf\left\{b \ge 0: \hat{L}_{\rm PT} \ge \frac{\lvert \htau_{\rm PT} \rvert}{(1+\gamma)^{-1/2} \sigma_0}\right\}
\end{align*}

At $b=0$, the defining equation for $\hat{L}_{\rm PT}$ becomes $\P_{\Delta/\sigma_0 = 0}(\lvert \htau_{\rm PT} - \tau \rvert \le L (1+\gamma)^{-1/2} \sigma_0) = 1-\zeta$. If
\[
  \P_{\Delta/\sigma_0 = 0}(\lvert \tilde\tau_{\rm PT} - \tau \rvert \le \lvert \htau_{\rm PT} \rvert \given \htau_{\rm PT}) < 1-\zeta,
\]
then $\hat{L}_{\rm PT}(0) > \frac{\lvert \htau_{\rm PT} \rvert}{(1+\gamma)^{-1/2} \sigma_0}$ and therefore $b^*_{\rm PT}(\zeta)=0$.

At $b=\infty$, the defining equation for $\hat{L}_{\rm PT}$ becomes $\min_{t \ge 0}\P_{\Delta/\sigma_0 = t}(\lvert \htau_{\rm PT} - \tau \rvert \le L (1+\gamma)^{-1/2} \sigma_0) = 1-\zeta$. If
\[
  \min_{t \ge 0}\P_{\Delta/\sigma_0 = t}(\lvert \tilde\tau_{\rm PT} - \tau \rvert \le \lvert \htau_{\rm PT} \rvert \given \htau_{\rm PT}) > 1-\zeta,
\]
then $\hat{L}_{\rm PT}(\infty) < \frac{\lvert \htau_{\rm PT} \rvert}{(1+\gamma)^{-1/2} \sigma_0}$ and therefore $b^*_{\rm PT}(\zeta)=\infty$.

Otherwise,
\[
  \min_{t \ge 0}\P_{\Delta/\sigma_0 = t}(\lvert \tilde\tau_{\rm PT} - \tau \rvert \le \lvert \htau_{\rm PT} \rvert \given \htau_{\rm PT}) \le 1-\zeta,~~ {\rm and}~~ \P_{\Delta/\sigma_0 = 0}(\lvert \tilde\tau_{\rm PT} - \tau \rvert \le \lvert \htau_{\rm PT} \rvert \given \htau_{\rm PT}) \ge 1-\zeta,
\]
so $b^*_{\rm PW}(\zeta) \in (0, \infty)$. By definition of the infimum, at
$b=b^*_{\rm PT}(\zeta)$ we must have the boundary condition
\[
\hat L_{\rm PT}(b^*_{\rm PT}(\zeta))
=
\frac{\lvert \htau_{\rm PT} \rvert}{(1+\gamma)^{-1/2}\sigma_0}.
\]
Plugging $L=\frac{\lvert \htau_{\rm PT} \rvert}{(1+\gamma)^{-1/2} \sigma_0}$ into the defining equation of $\hat{L}_{\rm PT}$ yields that $b = b^*_{\rm PT}(\zeta)$ solves
\[
  \min_{0 \le t \le b} \P_{\Delta/\sigma_0 = t}(\lvert \tilde\tau_{\rm PT} - \tau \rvert \le \lvert \htau_{\rm PT} \rvert \given \htau_{\rm PT}) = 1-\zeta,
\]
which is the desired characterization.

The explicit form of $\P_{\Delta/\sigma_0 = t}(\lvert \tilde\tau_{\rm PT} - \tau \rvert \le \lvert \htau_{\rm PT} \rvert \given \htau_{\rm PT})$ is directly given by Theorem~\ref{thm:L_PT}.

\end{proof}

\subsection{Proof of Theorem~\ref{thm:b_ST}}

\begin{proof}[Proof of Theorem~\ref{thm:b_ST}]

By Theorem~\ref{thm:L_ST}, the shortest symmetric centered confidence interval based on $\hat\tau_{\rm ST}$ is given by $[\htau_{\rm ST} - \hat{L}_{\rm ST} (1+\gamma)^{-1/2} \sigma_0, \htau_{\rm ST} + \hat{L}_{\rm ST} (1+\gamma)^{-1/2} \sigma_0]$, where $\hat{L}_{\rm ST}$ is the solution to
\begin{align*}
  \P_{\Delta/\sigma_0 = b}(\lvert \htau_{\rm ST}- \tau \rvert \le L (1+\gamma)^{-1/2} \sigma_0) = 1-\zeta.
\end{align*}
By the definition of the b-value,
\begin{align*}
  b^*_{\rm ST}(\zeta) =& \inf\left\{b \ge 0: 0 \in [\htau_{\rm ST} - \hat{L}_{\rm ST} (1+\gamma)^{-1/2} \sigma_0, \htau_{\rm ST} + \hat{L}_{\rm ST} (1+\gamma)^{-1/2} \sigma_0]\right\}\\
  =& \inf\left\{b \ge 0: \hat{L}_{\rm ST} \ge \frac{\lvert \htau_{\rm ST} \rvert}{(1+\gamma)^{-1/2} \sigma_0}\right\}
\end{align*}

At $b=0$, the defining equation for $\hat{L}_{\rm ST}$ becomes $\P_{\Delta/\sigma_0 = 0}(\lvert \htau_{\rm ST} - \tau \rvert \le L (1+\gamma)^{-1/2} \sigma_0) = 1-\zeta$. If
\[
  \P_{\Delta/\sigma_0 = 0}(\lvert \tilde\tau_{\rm ST} - \tau \rvert \le \lvert \htau_{\rm ST} \rvert \given \htau_{\rm ST}) < 1-\zeta,
\]
then $\hat{L}_{\rm ST}(0) > \frac{\lvert \htau_{\rm ST} \rvert}{(1+\gamma)^{-1/2} \sigma_0}$ and therefore $b^*_{\rm ST}(\zeta)=0$.

At $b=\infty$, the defining equation for $\hat{L}_{\rm ST}$ becomes $\P_{\Delta/\sigma_0 = \infty}(\lvert \htau_{\rm ST}- \tau \rvert \le L (1+\gamma)^{-1/2} \sigma_0) = 1-\zeta$. If
\[
  \P_{\Delta/\sigma_0 = \infty}(\lvert \tilde\tau_{\rm ST} - \tau \rvert \le \lvert \htau_{\rm ST} \rvert \given \htau_{\rm ST}) > 1-\zeta,
\]
then $\hat{L}_{\rm ST}(\infty) < \frac{\lvert \htau_{\rm ST} \rvert}{(1+\gamma)^{-1/2} \sigma_0}$ and therefore $b^*_{\rm ST}(\zeta)=\infty$.

Otherwise,
\[
  \P_{\Delta/\sigma_0 = \infty}(\lvert \tilde\tau_{\rm ST} - \tau \rvert \le \lvert \htau_{\rm ST} \rvert \given \htau_{\rm ST}) \le 1-\zeta,~~ {\rm and}~~ \P_{\Delta/\sigma_0 = 0}(\lvert \tilde\tau_{\rm ST} - \tau \rvert \le \lvert \htau_{\rm ST} \rvert \given \htau_{\rm ST}) \ge 1-\zeta,
\]
so $b^*_{\rm ST}(\zeta) \in (0, \infty)$. By definition of the infimum, at
$b=b^*_{\rm ST}(\zeta)$ we must have the boundary condition
\[
\hat L_{\rm ST}(b^*_{\rm ST}(\zeta))
=
\frac{\lvert \htau_{\rm ST} \rvert}{(1+\gamma)^{-1/2}\sigma_0}.
\]
Plugging $L=\frac{\lvert \htau_{\rm ST} \rvert}{(1+\gamma)^{-1/2} \sigma_0}$ into the defining equation of $\hat{L}_{\rm ST}$ yields that $b = b^*_{\rm ST}(\zeta)$ solves
\[
  \P_{\Delta/\sigma_0 = b}(\lvert \tilde\tau_{\rm ST} - \tau \rvert \le \lvert \hat\tau_{\rm ST}\rvert \given \hat\tau_{\rm ST}) = 1-\zeta,
\]
which is the desired characterization.

The explicit form of $\P_{\Delta/\sigma_0 = t}(\lvert \tilde\tau_{\rm ST} - \tau \rvert \le \lvert \htau_{\rm ST} \rvert \given \htau_{\rm ST})$ is directly given by Theorem~\ref{thm:L_ST}.

\end{proof}

\subsection{Proof of Theorem~\ref{thm:oneside_PW}}

\begin{proof}[Proof of Theorem~\ref{thm:oneside_PW}]

We follow the argument in the proof of Theorem~\ref{thm:L_PW}. From that proof, we know that
\begin{align*}
(1+\gamma)^{1/2} \sigma_0^{-1} (\htau_{\rm PW} - \tau) \sim N\left(\frac{\gamma}{\sqrt{1+\gamma}} \frac{\Delta}{\sigma_0}, 1\right).
\end{align*}

Consider the lower confidence bound
\[
  [\htau_{\rm PW} - L (1+\gamma)^{-1/2} \sigma_0, \infty).
\]
The worst-case coverage probability of the lower confidence bound is
\begin{align*}
  & \inf_{\Delta: \lvert \Delta/\sigma_0 \rvert \le b} \P_\Delta(\tau \in [\htau_{\rm PW} - L (1+\gamma)^{-1/2} \sigma_0, \infty))\\
  =& \inf_{\Delta: \lvert \Delta/\sigma_0 \rvert \le b} \P_\Delta(\htau_{\rm PW} - \tau \le L (1+\gamma)^{-1/2} \sigma_0)\\
  =& \inf_{\Delta: \lvert \Delta/\sigma_0 \rvert \le b} \P_\Delta((1+\gamma)^{1/2} \sigma_0^{-1} (\htau_{\rm PW} - \tau) \le L)\\
  =& \inf_{\Delta: \lvert \Delta/\sigma_0 \rvert \le b} \P_\Delta\left(N\left(\frac{\gamma}{\sqrt{1+\gamma}} \frac{\Delta}{\sigma_0}, 1\right) \le L\right)\\
  =& \P\left( N\left(\sup_{\Delta: \lvert \Delta/\sigma_0 \rvert \le b} \frac{\gamma}{\sqrt{1+\gamma}} \frac{\Delta}{\sigma_0}, 1\right) \le L \right)\\
  =& \P\left( N\left(\frac{\gamma}{\sqrt{1+\gamma}} b, 1\right) \le L \right)\\
  =& \Phi\Big(L - \frac{\gamma}{\sqrt{1+\gamma}} b \Big).
\end{align*}

Therefore, the proof is complete since $\hat{L}_{\rm PW}' = c_\zeta + \frac{\gamma}{\sqrt{1+\gamma}} b$ is the solution to
\[
  \Phi\Big(L - \frac{\gamma}{\sqrt{1+\gamma}} b \Big) = 1-\zeta.
\]  

\end{proof}

\subsection{Proof of Theorem~\ref{thm:oneside_PT}}

\begin{proof}[Proof of Theorem~\ref{thm:oneside_PT}]

We follow the argument in the proof of Theorem~\ref{thm:L_PT}. From that proof, we know that
\begin{align*}
  (1+\gamma)^{1/2} \sigma_0^{-1} (\htau_{\rm PT} - \tau) \given \{\lvert \htau_1 - \htau_0 \rvert \le \sigma c_{\alpha/2}\} \sim N\left(\frac{\gamma}{\sqrt{1+\gamma}} \frac{\Delta}{\sigma_0}, 1\right),
\end{align*}
and
\begin{align*}
  (1+\gamma)^{1/2} \sigma_0^{-1} (\htau_{\rm PT} - \tau) \given \{\sigma^{-1}(\htau_1 - \htau_0) = u, \lvert \htau_1 - \htau_0 \rvert > \sigma c_{\alpha/2}\} \sim N\left(\frac{\gamma}{\sqrt{1+\gamma}} \frac{\Delta}{\sigma_0} - \sqrt{\gamma} u, 1\right).
\end{align*}

Consider the lower confidence bound
\[
  [\htau_{\rm PT} - L (1+\gamma)^{-1/2} \sigma_0, \infty).
\]
The worst-case coverage probability of the lower confidence bound is
\begin{align*}
  & \inf_{\Delta: \lvert \Delta/\sigma_0 \rvert \le b} \P_\Delta(\tau \in [\htau_{\rm PT} - L (1+\gamma)^{-1/2} \sigma_0, \infty))\\
  =& \inf_{\Delta: \lvert \Delta/\sigma_0 \rvert \le b} \P_\Delta(\htau_{\rm PT} - \tau \le L (1+\gamma)^{-1/2} \sigma_0).
\end{align*}
For a fixed $t=\Delta/\sigma_0$, we decompose the probability
according to whether the pretest accepts or rejects:
\begin{align*}
  &\P_{\Delta/\sigma_0 = t}(\htau_{\rm PT} - \tau \le L (1+\gamma)^{-1/2} \sigma_0)\\
  =&\P_{\Delta/\sigma_0 = t}((1+\gamma)^{1/2} \sigma_0^{-1} (\htau_{\rm PT} - \tau) \le L)\\
  =&\P_{\Delta/\sigma_0 = t}((1+\gamma)^{1/2} \sigma_0^{-1} (\htau_{\rm PT} - \tau) \le L, \lvert \htau_1 - \htau_0 \rvert \le \sigma c_{\alpha/2}) \\
  &+ \P_{\Delta/\sigma_0 = t}((1+\gamma)^{1/2} \sigma_0^{-1} (\htau_{\rm PT} - \tau) \le L, \lvert \htau_1 - \htau_0 \rvert > \sigma c_{\alpha/2})\\ 
  =& \P\left( N\left(\frac{\gamma}{\sqrt{1+\gamma}} t, 1\right) \le L \right) \P\left( \left\lvert N\left(\sqrt{\frac{\gamma}{1+\gamma}} t, 1\right) \right\rvert \le c_{\alpha/2} \right)\\
  &+ \E_{U \sim N(\sqrt{\frac{\gamma}{1+\gamma}} t, 1)}\left[\P\left( N\left(\frac{\gamma}{\sqrt{1+\gamma}} t - \sqrt{\gamma} U, 1\right) \le L \right) \ind\left(\lvert U \rvert > c_{\alpha/2}\right)\right]\\
  =& \Big[\Phi\Big(c_{\alpha/2} - \sqrt{\frac{\gamma}{1+\gamma}}t \Big) - \Phi\Big(-c_{\alpha/2} - \sqrt{\frac{\gamma}{1+\gamma}}t \Big)\Big]\Phi\Big(L - \frac{\gamma}{\sqrt{1+\gamma}}t \Big)\\
  &+ \int_{-\infty}^{-c_{\alpha/2} - \sqrt{\frac{\gamma}{1+\gamma}}t} \Phi\Big(L + \sqrt{\gamma} u \Big)\phi(u) \d u + \int_{c_{\alpha/2} - \sqrt{\frac{\gamma}{1+\gamma}}t}^\infty \Phi\Big(L + \sqrt{\gamma} u\Big)\phi(u) \d u.
\end{align*}

Finally, the worst-case coverage probability over $\lvert \Delta/\sigma_0 \rvert \le b$ is attained for some $t\in[0,b]$. Choosing $\hat L_{\rm PT}'$ to satisfy
\[
  \min_{0 \le t \le b}\P_{\Delta/\sigma_0 = t}(\htau_{\rm PT} - \tau \le L (1+\gamma)^{-1/2} \sigma_0) = 1-\zeta
\]
completes the proof.

\end{proof}

\subsection{Proof of Theorem~\ref{thm:oneside_ST}}

\begin{proof}[Proof of Theorem~\ref{thm:oneside_ST}]

We follow the argument in the proof of Theorem~\ref{thm:L_ST}. From that proof, we know that
\begin{align*}
  (1+\gamma)^{1/2} \sigma_0^{-1} (\htau_{\rm ST} - \tau) \given \{\lvert \htau_1 - \htau_0 \rvert \le \sigma c_{\alpha/2}\} \sim N\left(\frac{\gamma}{\sqrt{1+\gamma}} \frac{\Delta}{\sigma_0}, 1\right),
\end{align*}
and
\begin{align*}
  &(1+\gamma)^{1/2} \sigma_0^{-1} (\htau_{\rm ST} - \tau) \given \{\sigma^{-1}(\htau_1 - \htau_0) = u, \lvert \htau_1 - \htau_0 \rvert > \sigma c_{\alpha/2}\} \sim N\left(\frac{\gamma}{\sqrt{1+\gamma}} \frac{\Delta}{\sigma_0} - \sqrt{\gamma} [u - c_{\alpha/2}\,\sign(u)], 1\right).
\end{align*}

Consider the lower confidence bound
\[
  [\htau_{\rm ST} - L (1+\gamma)^{-1/2} \sigma_0, \infty).
\]
The worst-case coverage probability of the lower confidence bound is
\begin{align*}
  & \inf_{\Delta: \lvert \Delta/\sigma_0 \rvert \le b} \P_\Delta(\tau \in [\htau_{\rm ST} - L (1+\gamma)^{-1/2} \sigma_0, \infty))\\
  =& \inf_{\Delta: \lvert \Delta/\sigma_0 \rvert \le b} \P_\Delta(\htau_{\rm ST} - \tau \le L (1+\gamma)^{-1/2} \sigma_0).
\end{align*}
For a fixed $t=\Delta/\sigma_0$, we decompose the probability
according to whether the pretest accepts or rejects:
\begin{align*}
  &\P_{\Delta/\sigma_0 = t}(\htau_{\rm ST} - \tau \le L (1+\gamma)^{-1/2} \sigma_0)\\
  =&\P_{\Delta/\sigma_0 = t}((1+\gamma)^{1/2} \sigma_0^{-1} (\htau_{\rm ST} - \tau) \le L)\\
  =&\P_{\Delta/\sigma_0 = t}((1+\gamma)^{1/2} \sigma_0^{-1} (\htau_{\rm ST} - \tau) \le L, \lvert \htau_1 - \htau_0 \rvert \le \sigma c_{\alpha/2}) \\
  &+ \P_{\Delta/\sigma_0 = t}((1+\gamma)^{1/2} \sigma_0^{-1} (\htau_{\rm ST} - \tau) \le L, \lvert \htau_1 - \htau_0 \rvert > \sigma c_{\alpha/2})\\ 
  =& \P\left( N\left(\frac{\gamma}{\sqrt{1+\gamma}} t, 1\right) \le L \right) \P\left( \left\lvert N\left(\sqrt{\frac{\gamma}{1+\gamma}} t, 1\right) \right\rvert \le c_{\alpha/2} \right)\\
  &+ \E_{U \sim N(\sqrt{\frac{\gamma}{1+\gamma}} t, 1)}\left[\P\left(N\left(\frac{\gamma}{\sqrt{1+\gamma}} t - \sqrt{\gamma} [U - c_{\alpha/2}\,\sign(U)], 1\right) \le L \right) \ind\left(\lvert U \rvert > c_{\alpha/2}\right)\right]\\
  =& \Big[\Phi\Big(c_{\alpha/2} - \sqrt{\frac{\gamma}{1+\gamma}}t \Big) - \Phi\Big(-c_{\alpha/2} - \sqrt{\frac{\gamma}{1+\gamma}}t \Big)\Big]\Phi\Big(L - \frac{\gamma}{\sqrt{1+\gamma}}t \Big)\\
  &+ \int_{-\infty}^{-c_{\alpha/2} - \sqrt{\frac{\gamma}{1+\gamma}}t} \Phi\Big(L + \sqrt{\gamma} (u+c_{\alpha/2})\Big)\phi(u) \d u + \int_{c_{\alpha/2} - \sqrt{\frac{\gamma}{1+\gamma}}t}^\infty \Phi\Big(L + \sqrt{\gamma} (u-c_{\alpha/2})\Big)\phi(u) \d u.
\end{align*}

Finally, since the coverage probability is monotonically decreasing in $\Delta$, the worst case over
$\lvert \Delta/\sigma_0\rvert\le b$ is attained at $\Delta/\sigma_0 = b$. Choosing $\hat L_{\rm ST}'$ to satisfy
\[
  \P_{\Delta/\sigma_0 = b}(\htau_{\rm ST}- \tau \le L (1+\gamma)^{-1/2} \sigma_0) = 1-\zeta
\]
completes the proof.

\end{proof}

\end{document}